\documentclass[
superscriptaddress,
nofootinbib,
 amsmath,amssymb,
 aps,
pra,
10pt,
]{revtex4-2}

\usepackage[dvipsnames]{xcolor}
\usepackage{amsmath}
\usepackage{amssymb}
\usepackage{ulem}
\usepackage{graphicx}%
\usepackage{caption}
\usepackage{subcaption}
\usepackage{booktabs}
\usepackage{dcolumn}%
\usepackage{bm}%
\usepackage{hyperref}%
\usepackage{tikz}
\usepackage{quantikz}

\usetikzlibrary{external}

\usepackage{relsize}
\usepackage{amsthm}
\usepackage{comment}
\usepackage{svg}
\usepackage{soul}
\usepackage{thm-restate}
\newtheorem{thm}{Theorem}
\newtheorem{lemma}[thm]{Lemma}
\theoremstyle{definition}
\newtheorem{defn}{Definition}
\newtheorem{cor}{Corollary}

\theoremstyle{definition}
\newtheorem{rem}{Remark}

\makeatletter
\newcommand*{\rom}[1]{\expandafter\@slowromancap\romannumeral #1@}
\makeatother

\newcommand{\sys}{\text{sys}}

\newcommand{\polylog}{\operatorname{polylog}}
\definecolor{bwcolor}{HTML}{E7115E}

\definecolor{darkbrown}{rgb}{0.4, 0.26, 0.13}
\definecolor{myDarkGreen1}{RGB}{0, 100, 80}

\newcommand{\HB}{\mathcal{H}_\beta}
\newcommand{\floor}[1]{\left\lfloor #1 \right\rfloor}
\newcommand{\ceil}[1]{\left\lceil #1 \right\rceil}

\newcommand{\norm}[1]{\left \lVert #1 \right \rVert}
\newcommand{\prtcl}[2]{\mathbf{x}_{#1}^{#2} }

\renewcommand\Re[1]{\operatorname{Re}\left( #1 \right)}

\newcommand{\grad}{\nabla}

\newcommand{\nocontentsline}[3]{}
\newcommand\stoptoc{%
   \let\origcontentsline\addcontentsline
   \let\addcontentsline\nocontentsline
}
\newcommand\resumetoc{%
   \let\addcontentsline\origcontentsline
}

\newcommand\remone{\bgroup\markoverwith{\textcolor{blue}{\rule[0.5ex]{2pt}{0.4pt}}}\ULon}

\newcommand\remtwo{\bgroup\markoverwith{\textcolor{red}{\rule[0.5ex]{2pt}{0.4pt}}}\ULon}

\newcommand\remthree{\bgroup\markoverwith{\textcolor{green}{\rule[0.5ex]{2pt}{0.4pt}}}\ULon}

\newcommand\remme{\bgroup\markoverwith{\textcolor{orange}{\rule[0.5ex]{2pt}{0.4pt}}}\ULon}

\begin{document}

\newpage

\setcounter{page}{1}

\preprint{APS/123-QED}

\title{Provable Quantum Speedups for Reaction-Rate Estimation in High-Dimensional Fokker-Planck Dynamics}%

\author{Tyler Kharazi}
\email{kharazitd@berkeley.edu}
\affiliation{Department of Chemistry, University of California, Berkeley}%

\author{Ahmad M. Alkadri}
\affiliation{Department of Chemical and Biomolecular Engineering, University of California, Berkeley}

\author{Kranthi K. Mandadapu}
\affiliation{Department of Chemical and Biomolecular Engineering, University of California, Berkeley}
\affiliation{Chemical Sciences Division, Lawrence Berkeley National Laboratory, Berkeley, California}

\author{K. Birgitta Whaley}
\email{whaley@berkeley.edu}
\affiliation{Department of Chemistry, University of California, Berkeley}
\affiliation{Chemical Sciences Division, Lawrence Berkeley National Laboratory, Berkeley, California}
\affiliation{Berkeley Quantum Information and Computation Center, University of California, Berkeley}
\date{\today}

\begin{abstract}
The Fokker-Planck equation models rare events across sciences, but {direct solution of the PDE is intractable for classical computers due to } its high-dimensional nature. Classical stochastic methods circumvent this curse-of-dimensionality, and serve as the de facto standard for practicing computational scientists. Quantum algorithms for such non-unitary dynamics often suffer from exponential decay in success probability. We introduce a quantum algorithm that overcomes this bottleneck for estimating reaction rates {and dynamical correlation functions more generally}. Using a sum-of-squares representation, we develop a Gaussian linear combination of Hamiltonian simulations (Gaussian-LCHS) to represent the non-unitary propagator with $O\left(\sqrt{t\|H\|\log(1/\epsilon)}\right)$ queries to its block encoding. Crucially, we pair this with {a} novel technique to directly estimate matrix elements without exponential decay. For $\eta$ pairwise interacting particles discretized with $N$ plane waves per degree of freedom, we estimate reactive flux to error $\epsilon$ using $\widetilde{O}\left((\eta^{5/2}\sqrt{t\beta}\alpha_V + \eta^{3/2}\sqrt{t/\beta}N)/\epsilon\right)$ quantum gates, where $\alpha_V = \max_{r}|V'(r)/r|$. We further prove that under comparable worst-case analytical guarantees, the sharpest classical bounds for estimating reaction rates via simulation of the associated overdamped Langevin dynamics scale as $O(t\eta^2 e^{\Omega(\eta)}/\epsilon^4)$, yielding an exponential improvement in $\eta$, a quartic speedup in $\epsilon$, and quadratic speedup in the time horizon $t$. While classical algorithms may outperform these bounds in practice, this work demonstrates a rigorous route toward quantum advantage for high-dimensional dissipative dynamics.
\end{abstract}

\maketitle
\newpage
\twocolumngrid
\tableofcontents
\onecolumngrid

\newpage
\section{Introduction}
\label{sec:intro}

Many questions in the sciences and engineering involve answering the question “how quickly does a system cross a barrier?”---for example: how fast a protein folds into its functional shape~\cite{Eaton2021,Chung2018}, how rapidly a crystal nucleus forms during materials processing~\cite{Karthika2016}, how often a chemical reaction proceeds in a crowded solvent~\cite{Dorsaz2010}, how quickly a glassy system relaxes \cite{HasyimGlass}, how long it takes for charged particles to escape magnetic confinement in a plasma~\cite{Heidbrink2020}, or how frequently a market variable hits a regulatory or contractual threshold~\cite{Dadachanji2015,Li2015}. Transition path theory (TPT) provides a powerful framework to answer these questions {for classical dynamical systems} \cite{Vanden-Eijnden2006,eTransitionPathTheoryPathFinding2010,e.TheoryTransitionPaths2006}. TPT is formulated in the language of the forward and backward Kolmogorov equations ({FKE and BKE, respectively}). {A graphical depiction of the relation between these equations is provided in Fig. \ref{fig:fpe-stoch} below}.  The FKE, also known as the Fokker-Planck equation, describes the deterministic evolution of the probability distribution for a stochastic process with a given initial distribution. The solution of the BKE on the other hand provides the so-called committor function, which encodes the probability of a trajectory emanating from some point in the configuration space, given that it terminates at some fixed region in configuration space. The committor function additionally provides an optimal external force to generate reactive trajectories \cite{hasyim_supervised_2022,singh_splitting_2024}. The (F)BKE framework provides a compact, operator-level description of reaction rates, reaction pathways, and transition state ensembles ~\cite{risken1996fokker,Vanden-Eijnden2006,Metzner2006}. 

\begin{figure}[ht!]
\centering
\includegraphics[width=.55\linewidth]{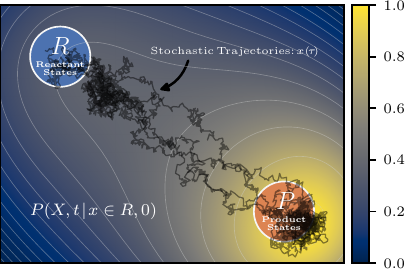}
\caption{Stochastic trajectories from an overdamped Langevin equation starting from region $R$ {evolve} in the configuration space with probabilities given by the Fokker-Planck equation. The stochastic trajectories terminate at some point $\mathbf{X}(t)$ in the configuration space with probability $P(\mathbf{X},t | \mathbf{x} \in R, 0)$, obtained by solving the Fokker-Planck equation {with an initial condition corresponding to the equilibrium distribution over $R$. {The evolution time of the stochastic trajectories is referred to as the time horizon $t$.}}}
\label{fig:fpe-stoch}
\end{figure}

{Direct numerical solution of the F(B)KE in many-body settings is essentially intractable due to the curse of dimensionality: A system of $\eta$ many particles in $d$ spatial dimensions using grid-based or spectral discretization of Eqs. \eqref{app-eq:BKE} or \eqref{eq:fpe} leads to a state space that scales as $O(N^{\eta d})$, where $N$ is the number of grid points or basis functions per degree of freedom ~\cite{Sun2014,Chen2018}. Recent classical work (e.g., tensor-network and hierarchical tensor solvers, semigroup and learned surrogates) offers important progress yet still encounters exponential or steep polynomial scaling in either state {space} dimension {$d$}, number of particles $\eta$, or time horizon {$t$} ~\cite{li_semigroup_2022,hasyim_supervised_2022,bekri_flowkac_2025,singh_splitting_2024}. }

On the other hand, numerical solution of the associated overdamped Langevin SDE does not suffer from the curse of dimensionality, and instead requires only a modest linear memory overhead in the dimensionality, i.e., $O(\eta d)$. Therefore, stochastic methods are the \textit{de facto} standard for the practicing computational scientist working on this class of problems.  In this setting, one computes observables by averaging over independent realizations of an underlying stochastic process.  Due to the independence of each trajectory, these methods can be implemented in a highly parallelizable manner. Estimating the reaction rate by sampling trajectories has cost $\mathcal{C} = \left(\text{number of trajectories}\right)\times \left(\text{cost per time step}\right) \times \left(\text{number of time steps}\right)$. The number of trajectories is $O(1/\epsilon^2)$ by standard Monte-Carlo methods. The cost per time step is na\"ively $O(\eta^2)$, but can be improved to $O(\eta \log(\eta))$ using the fast-multipole method \cite{beatson_short_nodate} or neighbor lists. The number of time steps is $t/\Delta t,$ where $\Delta t $ is the time step size and $t$ the total evolution time.

For numerical solution of the overdamped Langevin equation in a non-convex potential, the step size $\Delta t$ can be chosen according to the sharp estimates provided in Theorem 2 of Ref. \cite{chengSharpConvergenceRates2020}. {In that work}, the authors determine that a step size of $\Delta t \sim \frac{\epsilon^2 e^{-\gamma R^2}}{2^{10}R^2d \eta}$ is sufficient to control the bias resulting from the finite step-size. The parameter $R\geq \norm{x(0)}_2$ controls the radius outside of which the potential is strongly convex and $\gamma$ is the Lipschitz constant for the potential gradient $\nabla V$ over the domain (see section 2.1 of Ref. \cite{chengSharpConvergenceRates2020} and Ref. \cite{maSamplingCanBe2019}). Here, $\epsilon$ controls the error in the Wasserstein $1$ norm, which for the confining potentials that we consider in this work, is equivalent to convergence of first moments \cite{WolpertNotes}. It is natural to assume $R\in O(1)$, but the Lipschitz constant $\gamma$ may scale with the particle number $\eta$ as well as the domain size, as is the case for the pairwise potentials we consider in this work. Indeed, we show in Theorem \ref{thm:many-body-lip-const} that for the radially symmetric pair potentials made up of functions with bounded second derivative, the Lipschitz constant for the many-body potential gradient is $\Omega(\eta)$, meaning that the linear scaling of the Lipschitz constant with $\eta$ is an asymptotic lower bound.

Both stochastic and deterministic solvers face the challenge of the \textit{rare-event sampling problem}. For a barrier of height $\Delta V$ and inverse temperature $\beta$, the barrier crossings occur on a time-scale $O(\exp(\beta \Delta V))$. For stochastic trajectory-based integrators, long-time integration, poor mixing at low temperature (i.e., large $\beta$)~\cite[Chapter~22]{Bovier2015}~\cite{markowich_trend_nodate} and exponentially rare barrier crossings lead to high variance and large sample complexity, despite variance-reduction heuristics~\cite{Allen2009,Dean2009,Brhier2015}. In the deterministic case, this leads to exponentially long evolution times for the FKE, and an exponentially ill-conditioned linear system for the BKE with condition number $\kappa \in \Omega\left(e^{\beta\Delta V}\right)$ (see Fig.~\ref{fig:kappa-scaling} in Sec.~\ref{sec:prior-work}). This challenge is often characterized by a non-convex potential energy surface, leading to long-lived low-energy configurations of the system connecting two regions of the configuration space.

{In summary, three classical computational approaches to F(B)KE problems have been established in the literature: (i) solving the steady-state BKE as a linear system whose condition number $\kappa \sim \exp(\beta \Delta V)$ for metastable systems; (ii) time-evolving the FKE, which encounters the same mixing-time bottleneck through the spectral gap $\Delta_{\text{gap}}\sim \exp(\beta \Delta V)$  \cite{ashbaughMinimalMaximalEigenvalue1991}); and (iii) simulating the underlying SDE, which faces a similar mixing-time to (ii), i.e., $\sim \exp(\beta \Delta V)$, and is further compounded by the issue of choosing the correct step-size. All three approaches share a common physical bottleneck, {namely} metastability, {which manifests as different mathematical pathologies in each of the three approaches}. The quantum algorithm we develop in this paper provides an FKE-style time-evolution method that avoids both the condition number problem of the BKE linear solve and the trajectory-sampling rare-event problem of the SDE approach.
{Rather} than estimating $\nu_{RP}(t)$ as a frequency of rare events in a trajectory ensemble, it computes the propagator matrix element $\langle P | e^{t\mathcal{H}_\beta} | R\rangle$ directly via quantum amplitude estimation, achieving $O(1/\epsilon)$ additive-accuracy at fixed $t$ with no exponential dependence on the rarity of barrier crossings. Our worst-case complexity comparison is against approach (iii), {since this} is the \textit{de facto} standard in practice. }

{Crucially {however,} our {quantum} algorithm adopts approach (ii) because it admits a degree of freedom unavailable to approach (i), {namely,} the choice of simulation time $t$. In metastable systems, the generator spectrum typically exhibits an \textit{effective gap} $\Delta_{\text{eff}} \gg \Delta_{\text{gap}}$ separating the long-lived metastable manifold from bulk excitations. {This is illustrated schematically in Fig.~\ref{fig:metastability} (see} also Ch.~8.4 of Ref.~\cite{bovierMetastabilityPotentialTheoreticApproach2015}). Useful estimates of $\nu_{RP}(t)$ are then accessible at $t \sim 1/\Delta_{\text{eff}}$, well before the mixing time $1/\Delta_{\text{gap}}$ is reached. This is the chief structural advantage of the forward-time formulation over the linear-solve approach, and is what enables our worst-case complexity to depend on $t\|\mathcal{H}_\beta\|$ rather than on {the condition number} $\kappa$.}

To make this precise: implementing $e^{t_{\text{mix}}\mathcal{F}}$ for a generator $\mathcal{F}$ with $\lambda_{\max} = \norm{\mathcal{F}}_2$ requires $\Omega\!\left(\lambda_{\max}/\Delta_{\text{gap}}\right) = \Omega(\kappa)$ queries \cite{childsLimitations2010}, which is prohibitive when $\Delta_{\text{gap}} \sim e^{-\beta\Delta V}$.
In many physically relevant settings, however, sampling from the metastable manifold rather than the full stationary distribution is the scientifically relevant goal~\cite{razumovMetastableNanoprecipitatesAlloys2024,therrienMetastableMaterialsDiscovery2021,ghoshMetastableStatesProteins2020,chongPathsamplingStrategiesSimulating2017,dinnerMetastableStateFolding1998}.
For a normalized state $\ket{\psi} = \sum_{j \geq 0} c_j \ket{\phi_j}$ expanded in eigenstates of $\mathcal{F}$ (defined below in Sec.~\ref{sec:background}), with $\ket{\phi_1}$ {a state in the metastable manifold} at eigenvalue $\epsilon_1 = \Delta_{\text{gap}}$ 
 and $\ket{\phi_2}$ at eigenvalue $\epsilon_2$ satisfying $\left|\epsilon_2 - \epsilon_1\right| > \Delta_{\text{gap}}$,
\begin{align*}
    e^{-tL}\ket{\psi} &= c_0 \ket{\phi_0} + c_1e^{-t\epsilon_1}\ket{\phi_1} + c_2 e^{-t\epsilon_2}\ket{\phi_2} + O(e^{-t\epsilon_2})\\
    &= c_0 \ket{\phi_0} +  e^{-t\epsilon_1}\!\left(c_1\ket{\phi_1} + c_2 e^{-t(\epsilon_2-\epsilon_1)}\ket{\phi_2}\right) + O(e^{-t\epsilon_2}).
\end{align*}
The support on $\ket{\phi_2}$ is suppressed whenever $1/(\epsilon_2 - \epsilon_1) < t < 1/|\epsilon_1|$; when $\left|\epsilon_2 - \epsilon_1\right| \gg \left|\epsilon_1\right|$, this window is broad, and sampling from the metastable manifold {of states around $\ket{\phi_1}$} is achievable at $t \sim 1/\Delta_{\text{eff}}$, {where $\Delta_{\text{eff}} \sim (\epsilon_2 - \epsilon_1)$, rather than {at} $1/\Delta_{\text{gap}}$.
}

\begin{figure}[ht]
    \centering
    \includegraphics[width=.7\linewidth]{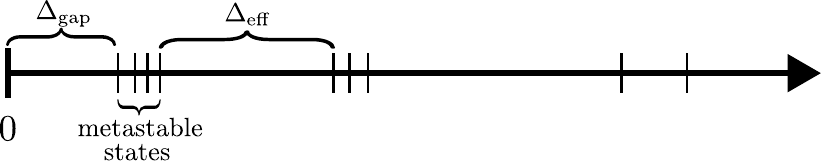}
    \caption{A diagrammatic description of the spectral characterization of metastability. There may be many low-lying eigenstates corresponding to local minima of the potential energy surface. The effective gap $\Delta_{\text{eff}}$ corresponds to the gap between the largest of these metastable states and the next largest eigenvalue of the generator.}
    \label{fig:metastability}
\end{figure}

{Many observables of interest in stochastic dynamics have a counterpart that can be expressed in terms of the associated deterministic PDE. Such examples include transport coefficients \cite{nevinsAccurateComputationShear2007,zwanzigTimeCorrelationFunctionsTransport1965}, Van Hove correlation functions  (from which dynamic structure factors can be computed) \cite{vanhoveCorrelationsSpaceTime1954, hansen2006theory}, and response functions \cite{kuboStatisticalMechanicalTheoryIrreversible1957}. Rather than attempting to prepare the full time-evolved or stationary probability density, our approach targets the direct estimation of such observables: {this approach encodes the desired physical information while also admitting} a more favorable quantum complexity.  Specifically, we focus {here} on the particular example of estimating the time-dependent reactive flux, 
\begin{equation}
    \nu_{RP}(t) = \bra{P} e^{t\mathcal{H}_\beta }\ket{R},
\end{equation}
where $\ket{R}$ and $\ket{P}$ are quantum states encoding the reactant and product regions, and the generator $\mathcal{H}_\beta$ is a self-adjoint representation of the FKE generator obtained via a similarity transformation. This approach circumvents the {two} primary quantum hurdles {mentioned above}: it avoids the condition number problem of linear solvers because it relies on time evolution, not steady state calculation, and it {also} {circumvents potentially exponential postselection overhead that would arise from preparing the full evolved state $e^{t\mathcal{H}_\beta}\ket{R}$ directly — for purely dissipative generators $H$, $\norm{e^{t H}\ket{R}}$ decays exponentially in $t$, making normalized state preparation prohibitively costly (see Sec.~\ref{subsec:sqrt-gauss-lchs}).}
{Instead, our approach estimates $\nu_{RP}(t)$ directly as a propagator matrix element.}}

\subsection{Overview and Contributions {of this work}}
Building on the correlation function formulation, we introduce {in this work} a quantum algorithm to efficiently estimate matrix elements of the Fokker-Planck propagator. To achieve this, we make three core contributions. First, {in Sec. \ref{subsec:sqrt-gauss-lchs}} we introduce an improved LCHS formula \cite{anLinearCombinationHamiltonian2023,anQuantumAlgorithmLinear2023a} for time-independent negative semi-definite matrices $H$ that can be expressed as $H = -\mathcal{A}^2$. In Sec.  \ref{subsec:block-encode-propagator} we construct a block encoding of the propagator with query complexity $\widetilde{O}\left(\sqrt{t\norm{H} \log(1/\epsilon)}\right)$, a near-quadratic improvement in the simulation time over Refs. \cite{anLinearCombinationHamiltonian2023, anQuantumAlgorithmLinear2023a, lowOptimalHamiltonianSimulation2017} {which also} saturates the fast-forwarded complexity for time evolution of this class of systems {that was} predicted in Lemma 36 of Ref. \cite{anQuantumDifferentialEquation2025}. We combine this bound with a technique that we term ``multiplexed QSP'' that implements this block encoding with gate complexity matching the block encoding query complexity of the LCHS method. 

Second, in Sec. \ref{subsec:overlap-circuit} we show that this overall sub-linear-in-$t$ dependence persists for the ubiquitous case of estimating overlaps of the non-unitary propagator. This behavior is obtained by a novel generalization of the Hadamard test for non-unitary matrices that avoids a major post-selection bottleneck by estimating the matrix elements $e^{Ht}$ directly from its LCU representation, without encountering the exponential success probability decay associated with preparing the full solution state. This circuit succeeds with probability $\Omega(1)$, independent of the evolution time $t$.  For estimating overlaps of purely dissipative dynamics, this construction provides a potentially exponential speedup over quantum algorithms that require state preparation and postselection. We believe the non-unitary overlap estimation algorithm introduced here will have many potential applications in quantum algorithms {{beyond} } the particular use case of {the FKE in} this paper. 

Third, we {derive} resource estimation costs for many of the subroutines. We apply these general techniques to the Fokker-Planck equation for $\eta$ pair-wise interacting particles in $d$ spatial dimensions, discretized with $N = 2^n$ plane wave modes per degree of freedom. Although we do not discuss the state preparation of $\ket{R}$ and $\ket{P}$ in detail, we describe in Sec. \ref{subsec:state-prep} {a method for fast thermal state preparation in locally convex regions, corresponding to local minima of the potential energy surface.} We {then} provide explicit circuits, and estimates of subnormalization factors, for block encoding the necessary operators --- covering both polynomial pair potentials directly and physically motivated singular pair potentials (Lennard-Jones, Morse, screened Coulomb) via a smoothed surrogate construction with matching asymptotic cost --- and show that the reactive flux can be estimated to additive error $\epsilon$ using $\widetilde{O} \left(\left(\eta^{5/2} \sqrt{t\beta}\alpha_V + \eta^{3/2}\sqrt{\frac{t}{\beta}} N\right)\epsilon^{-1}\right)$ quantum gates {(Secs. \ref{sec:alg-analysis} and \ref{sec:Conclusion})}. Crucially, we prove (Theorem \ref{thm:many-body-lip-const}) that for radially symmetric pair potentials with bounded second derivative, the Lipschitz constant of the potential gradient scales as $\Omega(\eta)$. 

Together with the sharp analytical bounds of Ref. \cite{chengSharpConvergenceRates2020} (see Eq. \eqref{eq:classical-complexity}) for non-convex potentials, our results imply that our quantum algorithm achieves an exponential improvement with respect to $\eta$ (when $N \sim \polylog(1/\epsilon)$), we also observe a quartic improvement with respect to $\epsilon$, and a quadratic improvement with respect to $t$, relative to these classical bounds.
We {emphasize} that our results should not be interpreted as a speedup over all classical simulation methods. Instead, our contribution is to identify a class of physically relevant observables—expressible as correlation functions of the Fokker–Planck propagator—that are provably hard to compute with guaranteed error in the worst case using classical algorithms under standard assumptions, yet can be estimated efficiently by our quantum approach.

The rest of the paper is organized as follows. In Sec. \ref{sec:prior-work} we review quantum and classical algorithms {for classical PDEs} that precede this work. We additionally perform a comparison between our quantum algorithm and the best known analytical bounds on the worst-case complexity for {simulation of} the related overdamped Langevin dynamics. In Sec. \ref{sec:background} we introduce the necessary background, detail the matrix square root construction, discuss the mapping from the Fokker-Planck equation to an imaginary time Schr\"odinger equation, and define the reaction rate in terms of overlaps of the propagator. We then review the main subroutines of the algorithm, and present the quantum circuit for estimating non-unitary overlaps. In Sec. \ref{sec:alg-analysis}, we provide an analysis of the complexity of our algorithm, as well as detailed resource estimates to construct the necessary block encodings. We then conclude in Sec. \ref{sec:Conclusion} with a summary and a discussion {of} open questions and directions for future work.

\color{black}

\section{Prior work}
\label{sec:prior-work}
\subsection{Quantum Algorithms}
There have been several distinct approaches to solving linear non-unitary differential equations on {quantum computers}. {These include} algorithms based on solving a quantum linear systems problem to prepare history states \cite{berryHighorderQuantumAlgorithm2014, an_fast-forwarding_2024}, time stepping \cite{fang_time-marching_2023}, and mapping the non-unitary propagator to a linear combination of unitaries \cite{anLinearCombinationHamiltonian2023,anQuantumAlgorithmLinear2023a,jinQuantumSimulationPartial2022}. In particular, our work leverages the LCHS technique developed in Refs. \cite{anLinearCombinationHamiltonian2023,anQuantumAlgorithmLinear2023a}, which provides a near-optimal construction for preparing the non-unitary propagator as a block encoding. 
Quantum PDE solvers {have been proposed} in a {broad range} of  settings, e.g., the wave equation \cite{costa_quantum_2019}, electrodynamics  \cite{koukoutsis_dyson_2023}, molecular dynamics \cite{ollitrault_molecular_2021}, the Fokker-Planck equation \cite{jinQuantumSimulationFokkerPlanck2024}, the simulation of classical coupled oscillators \cite{babbush_exponential_2023}, and even some non-linear PDEs \cite{Liu2021a,jin_quantum_2022,An2022,costa_further_2023}. There have also been some general studies on the efficiency of quantum differential equation solvers \cite{jin_time_2022,jennings_cost_2024}, as well as the establishment of complexity theoretic lower bounds \cite{anQuantumDifferentialEquation2025}. {However, except} in some very special circumstances where it is possible to map the underlying linear differential equation to a Schr\"odinger equation, such as was observed in Refs. \cite{babbush_exponential_2023,koukoutsis_dyson_2023}, it {has continued to be} challenging to identify problem instances where quantum computers can provide {guaranteed} exponential or even significant polynomial speedups over classical {PDE} solvers.

In addition to generic algorithmic developments for solving linear PDEs, we briefly review some prior work on the development of specific quantum algorithms for {applications closely related to our F(B)KE problems.} 
Ref. \cite{chowdhuryQuantumAlgorithmsGibbs2016} considers the problem of estimating hitting times and sampling from Gibbs states using the spectral gap amplification approach, which shares similar properties to the Gaussian-LCHS approach that we use {in this work}. 
Ref. \cite{holmes_quantum_2022} develops quantum algorithms for preparing purifications of thermal states with bounds based on fluctuation theorems from statistical mechanics. Ref. \cite{jinQuantumSimulationFokkerPlanck2024} develops a quantum algorithm for the Fokker-Planck equation applying the Schrodingerization technique and employ a finite difference discretization. {However,} as shown in their Theorem 3.2, {that algorithm} has time complexity  $\Omega(t^2)$ {unlike the $O(\sqrt t)$ scaling of our reaction rate computation.} 
Most recently, Ref. \cite{lengOperatorLevelQuantumAcceleration2025} considered preparation of thermal states of the Fokker-Planck equation, using eigenstate filtering and replica exchange to sample from non-log-concave distributions, together with the sum of squares decomposition from Ref. \cite{Witten}.

The BKE linear-solve problem is also a natural target for quantum linear-systems algorithms (QLSA). Ref.~\cite{orsucciSolvingClassesPositivedefinite2021a} develops an $O(\sqrt{\kappa})$ QLSA for the class of semi-definite matrices arising in the BKE, matching the classical conjugate-gradient baseline. Here $\kappa$ denotes the condition number of the \textit{discretized} BKE generator (a bounded operator on $\mathbb{C}^{N^{\eta d}}$); we note this clarification because the continuum generator is an unbounded operator for which $\kappa$ is not directly defined. However, by Proposition~16 of Ref.~\cite{orsucciSolvingClassesPositivedefinite2021a}, even this $O(\sqrt{\kappa})$ scaling algorithm requires a number of gates that scales exponentially with particle number $\eta$ for the sum-of-squares encoding of the Witten Laplacian.  Fig.~\ref{fig:kappa-scaling} provides numerical estimates of {the discretized BKE condition number} $\kappa$ for small problem instances {in an example with a one-dimensional quartic potential}, confirming $\kappa(N,\beta) \sim N^2 e^{\beta\Delta V}$. Any QLSA-based approach to the BKE therefore incurs gate complexity growing exponentially in the barrier height $\Delta V$.

For {quantum algorithms for non-unitary dynamics that are based on time evolution}, a different challenge arises, which we term the stability-dissipation conflict. For a linear system $\partial_t \mathbf{u}(t) = A \mathbf{u}(t)$, the solution is $\mathbf{u}(t) = e^{tA}\mathbf{u}(0)$. On the one hand, if $A$ has eigenvalues with positive real part and $\mathbf{u}(0)$ is supported on the corresponding eigenvectors, then the solution norm increases exponentially with $t$ and the system is unstable. On the other hand, if $A$ contains {a} negative real spectrum, the solution norm can decay exponentially with $t$, leading to exponentially small success probabilities in preparing $\ket{\mathbf{u}(t)}$. While block encoding methods can approximate $e^{tA}$ with $O\left(t \polylog\left(1/\epsilon\right)\right)$ queries to the block encoding of $A$, the overall complexity including the post-selection overhead becomes $\widetilde{O}\left(\frac{\norm{\mathbf{u}(0)}}{\norm{\mathbf{u}(t)}}t \polylog(1/\epsilon)\right)$. Thus when $\text{Re}(A) < 0$, {we have} $\norm{\mathbf{u}(t)} \in O(e^{-t})$, leading to an exponential dependence on the evolution time. Although alternatives based on time stepping \cite{fang_time-marching_2023} or history states \cite{berryHighorderQuantumAlgorithm2014, an_fast-forwarding_2024} can circumvent this scaling in some scenarios, the conflict between stability and dissipation is at the core of the difficulty in designing efficient quantum algorithms for the solution state of dissipative non-unitary PDEs. The resolution, presented in Sec.~\ref{subsec:overlap-circuit}, is to estimate the desired propagator matrix element $\nu_{RP}(t) = \bra{P}e^{t\mathcal{H}_\beta}\ket{R}$ directly from the LCU representation, bypassing normalized state preparation entirely.

\begin{figure}[htbp]
    \centering

    \begin{subfigure}[t]{0.48\textwidth}
        \centering
        \includegraphics[width=\linewidth]{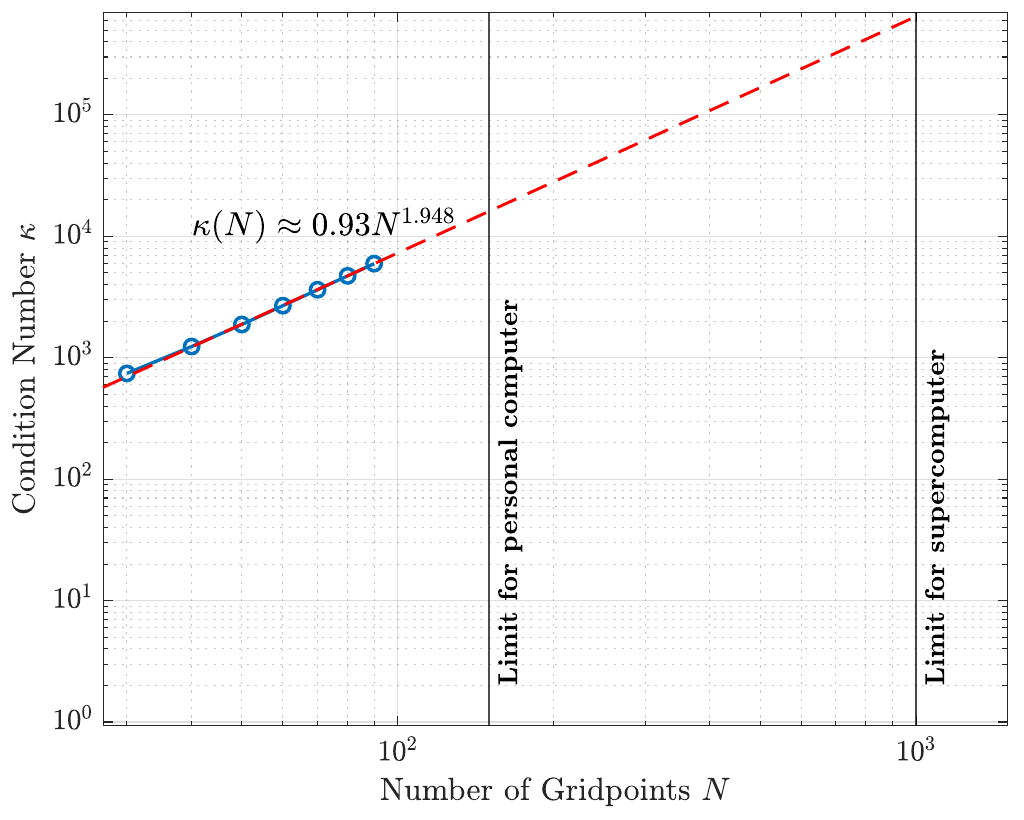}
        \phantomsubcaption
        \label{fig:kappa-vs-N}
    \end{subfigure}
    \hfill
    \begin{subfigure}[t]{0.48\textwidth}
        \centering
        \includegraphics[width=\linewidth]{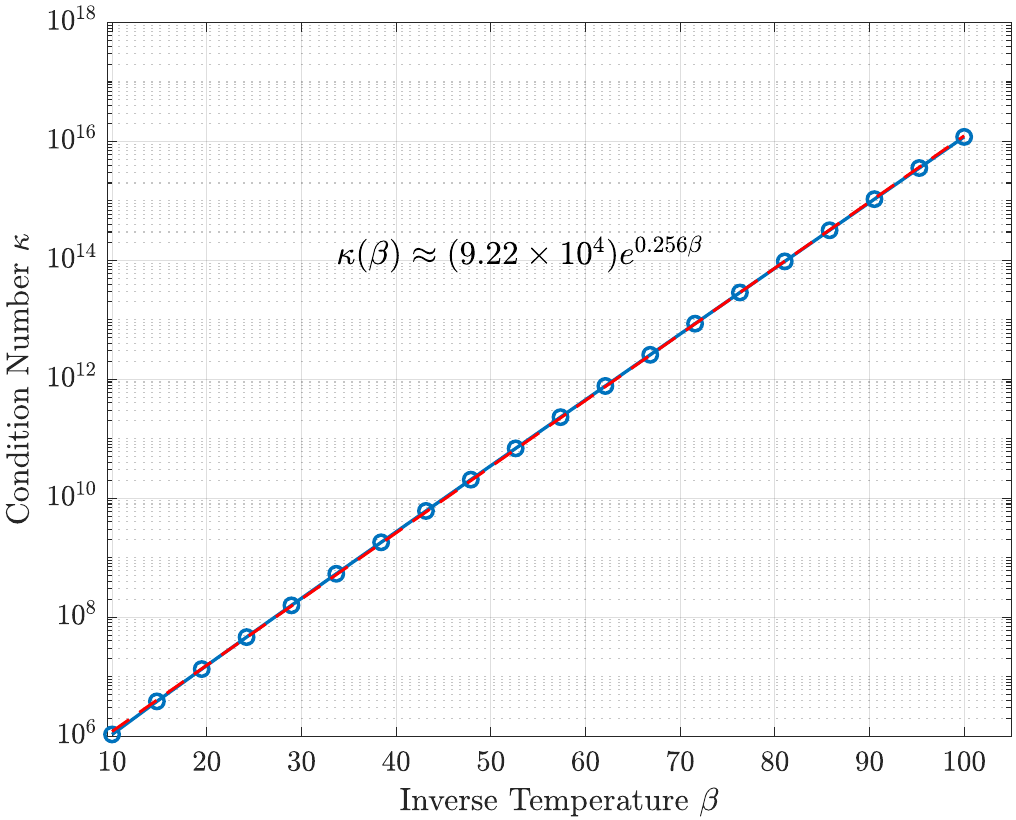}
        \phantomsubcaption
        \label{fig:kappa-vs-beta}
    \end{subfigure}
    \caption{{Scaling of the condition number $\kappa$ of the finite-difference discretized BKE generator for a one-dimensional double-well $(V(x) = x^4 - x^2)$ test problem. Condition numbers are estimated using MATLAB's \texttt{condest} function~\cite{MATLAB}. (a)~$\kappa$ versus grid resolution $N$ for particle number $\eta = 3$; a log--log least-squares fit (red dashed line) yields $\kappa(N) \approx 0.93 N^{1.948}$, indicating approximately quadratic dependence on grid resolution. (b)~$\kappa$ versus inverse temperature $\beta$ for a single particle in a double-well potential over $[-4,4]$ with $N = 2000$ grid points, over the range $\beta \in [10,100]$; an exponential fit yields $\kappa \approx (9.22 \times 10^4)e^{0.256\beta}$, nearly matching the theoretically predicted scaling $\kappa(N,\beta) \sim (N/L)^2 e^{\beta/4}$, where $L$ is the domain half-width (here $L = 4$). Vertical lines in panel~(a) indicate rough limits for feasible direct sparse solves on a personal workstation and on current exascale-class systems (e.g., the Frontier supercomputer~\cite{Frontier2025}), based on ${\sim}N^\eta$ memory growth. Consequently, the ``curse of dimensionality'' associated with grid-based BKE discretization is typically avoided classically by employing high-dimensional approximation techniques, as discussed in Sec.~\ref{sec:intro}.}}
    \label{fig:kappa-scaling}
\end{figure}

\subsection{Classical Algorithms}

In addition to the aforementioned quantum algorithms, there has been significant progress on the side of classical algorithms for {the F(B)KE problems.} Algorithms that approximate the committor function, which is obtained by solving a Dirichlet boundary value problem for the steady-state of the backward generator (BKE), have been proposed using tensor networks \cite{chen_committor_2023}, machine learning approaches \cite{hasyim_supervised_2022, khooSolvingHighdimensionalCommittor2018,wangEstimatingCommittorFunctions2026,liSemigroupMethodHigh2022}, and point-cloud discretization of the Fokker-Planck operator \cite{laiPointCloudDiscretization2018}.  
For the Fokker-Planck equation itself, similar classical algorithms have been proposed based on hierarchical tensor networks \cite{tang_solving_2024} and machine learning \cite{bekri_flowkac_2025}. These algorithms perform well under various assumptions, such as i) being efficiently representable as a matrix product operator or tensor train, ii) the solution is effectively restricted to a quasi one-dimensional region {referred to as} a ``reaction tube'', or iii) that the solution is well approximated by only a small number of low-lying eigenfunctions. 
These methods have been {quite} effectively applied to the one and two-dimensional Ginzburg-Landau model and double well potentials, but stop short of addressing high-dimensional problem instances involving multiple interacting particles. 
{Moreover,} none of these approaches can provide the same level of provable performance guarantees as a quantum algorithm performing direct numerical simulations of the underlying PDE, and are {further} unlikely to scale to {the} billions or even trillions of degrees of freedom that are needed to represent the full dynamical generator for many-body systems.

{Using the step-size bounds of Ref.~\cite{chengSharpConvergenceRates2020} (discussed in Sec.~\ref{sec:intro}) combined with our Theorem~\ref{thm:many-body-lip-const} below, the worst-case classical complexity for SDE-based trajectory simulation is}
\begin{equation}
   \mathcal{C} \sim \frac{2^{10} t R^2 \eta^2 \log(\eta) e^{R^2\gamma}}{\epsilon^4 } \in O\left(\frac{t \eta^2\log(\eta) e^{\Omega(\eta)}}{\epsilon^4}\right).
   \label{eq:classical-complexity}
\end{equation}
Therefore, our algorithm which accomplishes the same task using $\widetilde{O} \left(\left(\eta^{5/2} \sqrt{t\beta}\alpha_V + \eta^{3/2}\sqrt{\frac{t}{\beta}} N\right)\epsilon^{-1}\right)$ gates, represents a quadratic speedup with respect to the time horizon $t$, a {quartic} speedup in the accuracy $\epsilon$, and an \textit{exponential} speedup in $\eta$, relative to the best known classical analytical bounds. 

\begin{restatable}{thm}{lipschitzConst}
\label{thm:many-body-lip-const}
Let $V :\mathbb{R} \rightarrow \mathbb{R}$ and assume that $\nabla V$ is Lipschitz continuous with Lipschitz constant $\gamma>0$. Further assume that $\nabla V$ is everywhere differentiable. Let $r_{ij} = \norm{\mathbf{x}_i - \mathbf{x}_j}$ for all $\mathbf{x}_i, \mathbf{x}_j$. For $\eta > 2$, define the pair potential
\begin{align*}
    V_{\text{pair}} = \sum_{i=0}^{\eta - 1}\sum_{j<i} V(r_{ij}).
\end{align*}
Then, the Lipschitz constant $\text{Lip}\left(\nabla V_{\text{pair}}\right) \in \Omega\left(\eta \gamma\right).$
\end{restatable}
\noindent
The proof of this theorem is given in Appendix \ref{app-sec:pf-lipschitz-const}.

Our complexity comparison assumes equivalent oracles for initial state access: classical methods require sampling from the Boltzmann distribution restricted to $R$ and $P$, while our quantum algorithm analogously requires a coherent encoding of the quantum states $\ket{R}$ and $\ket{P}$.
This oracle model  {equivalence allows the complexity of the dynamical computation to be isolated}. 

We also wish to note that {in practice,} classical algorithms could outperform the worst-case estimates of Eq. \eqref{eq:classical-complexity}.
Under stronger assumptions on the convexity of the potential, when targeting weak error estimates (i.e., producing $\epsilon$ accurate expectation values), there exist classical methods (e.g., that of Ref. \cite{leimkuhlerLongtimeIntegrationStochastic2014}) that can achieve a stepsize scaling as $O(\sqrt{\epsilon})$, leading to an overall $\epsilon^{-2.5}$ scaling rather than the $\epsilon^{-4}$ scaling predicted in Eq. \eqref{eq:classical-complexity}. Although {more} stringent requirements on the convexity of the potential are required in order for the analytical statements to hold, numerical evidence suggests that these techniques may be more efficient in some circumstances, leading to a potentially reduced classical complexity. Furthermore, the $O(\epsilon^{-2})$ Monte-Carlo sampling can often be significantly improved through importance sampling techniques \cite{vanden-eijndenRareEventSimulation2012, dupuisImportanceSampling2012}. Nevertheless, this discrepancy between analytical predictions and the performance in practice of classical algorithms {further} highlights the difficulty of conducting a proper comparison of quantum and classical solvers for this task.

\section{Overview of methodology}
\label{sec:background}

\subsection{Background on Kolmogorov equations}
\label{sec:bkgrd-prior-work}
The Fokker-Planck equation (or forward Kolmogorov {equation (FKE))} describes the deterministic evolution of a probability distribution over all possible stochastic trajectories associated with an overdamped Langevin equation. Some additional details on the Kolmogorov equations and derivations of some of the facts we present here can be found in Appendix \ref{app-sec:kolmogorov-background}. 
Recall the expression for the dynamical generator of the FKE,
\begin{equation}
     \mathcal{F}[\rho] = \nabla \cdot\left( \rho \nabla V + \beta^{-1}\nabla\rho \right),
     \label{eq:fke-gen}
\end{equation}
and the associated first-order-in-time PDE,
\begin{equation}
    \begin{aligned}
        \partial_t \rho(t) &= \mathcal{F}[\rho](t) ,\\
        \rho(0) &= \rho_0 ,
    \end{aligned}
    \label{eq:fpe}
\end{equation}
with {$\beta =$1/k$_\textrm{B}$T} the inverse temperature, $\rho_0$ an initial condition, and $V$ a ``confining" potential, meaning that $V(|x|)\rightarrow \infty$ as $x\rightarrow \infty$. The assumption that $V$ is confining is technical, in that it ensures {both} that $\mathcal{F}$ is negative semi-definite, and that the equilibrium state is unique and Boltzmann distributed \cite{markowich_trend_nodate}. We note, however, that in non-confining potentials these properties may also be recovered by the enforcement of a homogeneous Neumann boundary condition. 

Under the similarity transformation generated by $\rho_\beta = \sqrt{\frac{e^{-\beta V}}{\mathcal{Z}}}$, we have
\begin{equation} \label{eq:FP_selfadjoint}
\begin{aligned}
    \text{Ad}_{\rho_\beta}\left(\mathcal{F}\right) &:= \rho_\beta^{-1}\mathcal{F}[\rho_\beta]\\
    \text{Ad}_{\rho_\beta}\left(\mathcal{F}\right) &= \beta^{-1}\Delta -\frac{\beta}{4}\norm{\nabla V}^2(x) + \frac{1}{2}\Delta V(x) =: \HB,
\end{aligned}
\end{equation}
where $\mathcal{Z}$ is the partition function and $\HB$ is an explicitly self-adjoint Schr\"odinger-type operator. In the self-adjoint representation, $\mathcal{H}_\beta$ admits a direct decomposition into a sum of squares
\begin{equation}
\begin{aligned}
    \mathcal{H}_\beta &= -\sum_{ j \in [\eta d]} A_j^\dagger A_j\\
    A_j &= -i\left(\sqrt{\beta^{-1}}\partial_{x_j} - \frac{\sqrt{\beta}}{2} F_j\right),
    \label{eq:sos-form-H}
\end{aligned}
\end{equation}
where $F_j = -\partial_{x_j} V$ is interpreted as a multiplication operator.

As described in Appendix \ref{app-sec:kolmogorov-background}, the reaction rate {at horizon time T} can {then} be expressed with respect to the generator $\HB$ as
\begin{equation}
     k_{RP}(T) := \frac{1}{T} \sqrt{\frac{\overline{p}_P}{\overline{p}_R}} \bra{P}  e^{T \HB}  \ket{R},
    \label{eq:sym-rxn-rate}
\end{equation}
where 
\begin{equation}
\begin{aligned}
    \ket{R} &= \int dx 1_R(x) \frac{e^{-\beta V(x)/2}}{\sqrt{\mathcal{Z}\bar{p}_R }}\ket{x},
    \label{eq:reactive-state}
\end{aligned}
\end{equation}
 $\bar{p}_R$ is the equilibrium probability of the region $R$,
\begin{equation}
    \overline{p}_R = \int dx\, 1_R(x) \mu(x),
\end{equation}
$1_R(x) = 1$ if $x \in R$ and $0$ otherwise is {an} indicator function {for the reactant region}, and $\mu(x) = \frac{e^{-\beta V(x)}}{\mathcal{Z}}$ is the equilibrium measure. {The product region quantities $1_P(x)$ and $\overline{p}_P = \int dx\, 1_P(x)\mu(x)$ are defined analogously.}
From Eq. \eqref{eq:sym-rxn-rate} we {can identify} the reactive flux {$\nu_{RP}(T)$ at horizon time $T$ as}
\begin{equation}
\nu_{RP}(T) := \bra{P} e^{T\HB} \ket{R} = \langle 1_P, e^{T \mathcal{F}} 1_R\rangle_\mu,
\label{eq:rxn-flux-overlap}
\end{equation}
{with $\mathcal{F}$ the Fokker-Planck generator, Eq.~{\ref{eq:fke-gen}}.}
This work will focus on developing a quantum algorithm for estimating the {reactive flux, interpreted as the propagator} overlap Eq. \eqref{eq:rxn-flux-overlap}.  With a small additional overhead to compute estimates to $\bar{p}_R$ and $\bar{p}_P$, and access to the quantum state $\ket{R}$ and $\ket{P}$, this information can {then} be used to compute the time-dependent reaction rate Eq. \eqref{eq:sym-rxn-rate}. Moreover, an estimate of reactive flux can also provide estimates of hitting times, {since}
\begin{equation}
    \tau_{RP} := \inf \{t > 0 : \nu_{RP}(t) > 0\}.
\end{equation}

The long-time limit of the reactive flux admits a simple closed form.  Since $\mathcal{H}_\beta$ is negative semi-definite with a unique zero eigenvalue and ground state $\ket{\phi_0}$ satisfying $\langle x|\phi_0\rangle \propto e^{-\beta V(x)/2}$, the eigendecomposition gives
\begin{align}
    \nu_{RP}(t)
    = \langle{P}|\phi_0\rangle \langle \phi_0| R\rangle 
      + \sum_{j \geq 1} e^{\lambda_j t}\langle P | \phi_j \rangle \langle \phi_j | R \rangle ,
    \label{eq:nu-eigendecomp}
\end{align}
where $\lambda_j < 0$ for $j \geq 1$.  As $t\to\infty$ all decaying terms vanish, leaving
\begin{align}
    \nu_{RP}(\infty) = \langle P | \phi_0 \rangle \langle \phi_0| R\rangle
    = \sqrt{\bar{p}_P}\cdot\sqrt{\bar{p}_R} = \sqrt{\bar{p}_R\,\bar{p}_P}.
    \label{eq:nu-inf-limit}
\end{align}
This long-time value is {thermodynamic}: it depends only on the equilibrium populations $\bar{p}_R$ and $\bar{p}_P$, which can be estimated classically by importance sampling or, for locally convex regions, by the thermal state preparation of Sec.~\ref{subsec:state-prep}.  
The kinetically nontrivial content of $\nu_{RP}(t)$ is the transient $\nu_{RP}(t) - \sqrt{\bar{p}_R\bar{p}_P}$, which decays at a rate set by the leading nonzero eigenvalue $\lambda_1$ of $\mathcal{H}_\beta$.  
Two systems with identical equilibrium populations $\bar{p}_R$, $\bar{p}_P$, and hence identical long-time limits, can have entirely different transients and thus entirely different rate constants.  The rate constant is encoded in $\lambda_1$, and resolving it requires accessing the full time-dependent transient. 
This is precisely the task {that} the quantum algorithm of this paper addresses.

This work may also readily be extended to compute the probabilities of transition states.
The probability of realizing a transition state $x$ at some time $t$ is expressed as
\begin{equation}
    \pi_{AB}(x,t) = \bra{B}\rho^{-1} e^{(T-t)\HB } \rho\ket{x}\bra{x} \rho e^{t\HB } \rho^{-1}\ket{A}.
\end{equation}
Using $|A| \equiv \text{vol}(A)$, with $\Pi_x = \ket{x}\bra{x}$, this can be simplified to
\begin{equation}
    \pi_{AB}(x,t) \propto \frac{1}{\sqrt{\left|A\right|\left|B\right|}}\sum_{y' \in B, y\in A}\bra{y'} e^{(T-t)\HB} e^{-\beta V/2}\Pi_x e^{-\beta V/2} e^{t\HB }\ket{y}.
\end{equation}
Thus, the computation of the transition density, {i.e., the probability of reaching a transition state located at configuration $x$ at time $t$}, can also be reduced to {computation of} the overlap of a non-unitary matrix, {albeit involving a more complex propagator correlation function}. 

\subsection{Quantum algorithmic approach}
\label{sub:algorithm-approach}

{Building on the formulation of Sec.~\ref{sec:bkgrd-prior-work}, our algorithm estimates $\nu_{RP}(t) = \bra{P}e^{t\mathcal{H}_\beta}\ket{R}$ (Eq.~\eqref{eq:rxn-flux-overlap}) via two main steps: (i) constructing an LCU representation of the propagator $e^{t\mathcal{H}_\beta}$ via Gaussian-LCHS applied to the dilated matrix square root $\mathcal{A}$ (Sec.~\ref{subsec:sqrt-gauss-lchs}), and (ii) extracting the overlap $\bra{P}e^{t\mathcal{H}_\beta}\ket{R}$ directly via the non-unitary overlap circuit (Sec.~\ref{subsec:overlap-circuit}), avoiding the stability-dissipation conflict that arises when preparing the full state $e^{t\mathcal{H}_\beta}\ket{R}$. 
The {full quantum} algorithmic pipeline is depicted in  Fig.~\ref{fig:reductions}, {with the initial construction of the Schr\"{o}dinger-type operator $\mathcal{H}_\beta$ in the left panel and the route to computation of the overlap $\bra{P}e^{t\mathcal{H}_\beta}\ket{R}$  following in the right panel.}}

\begin{figure}[ht!]
    \centering
    \includegraphics[width=0.7\linewidth]{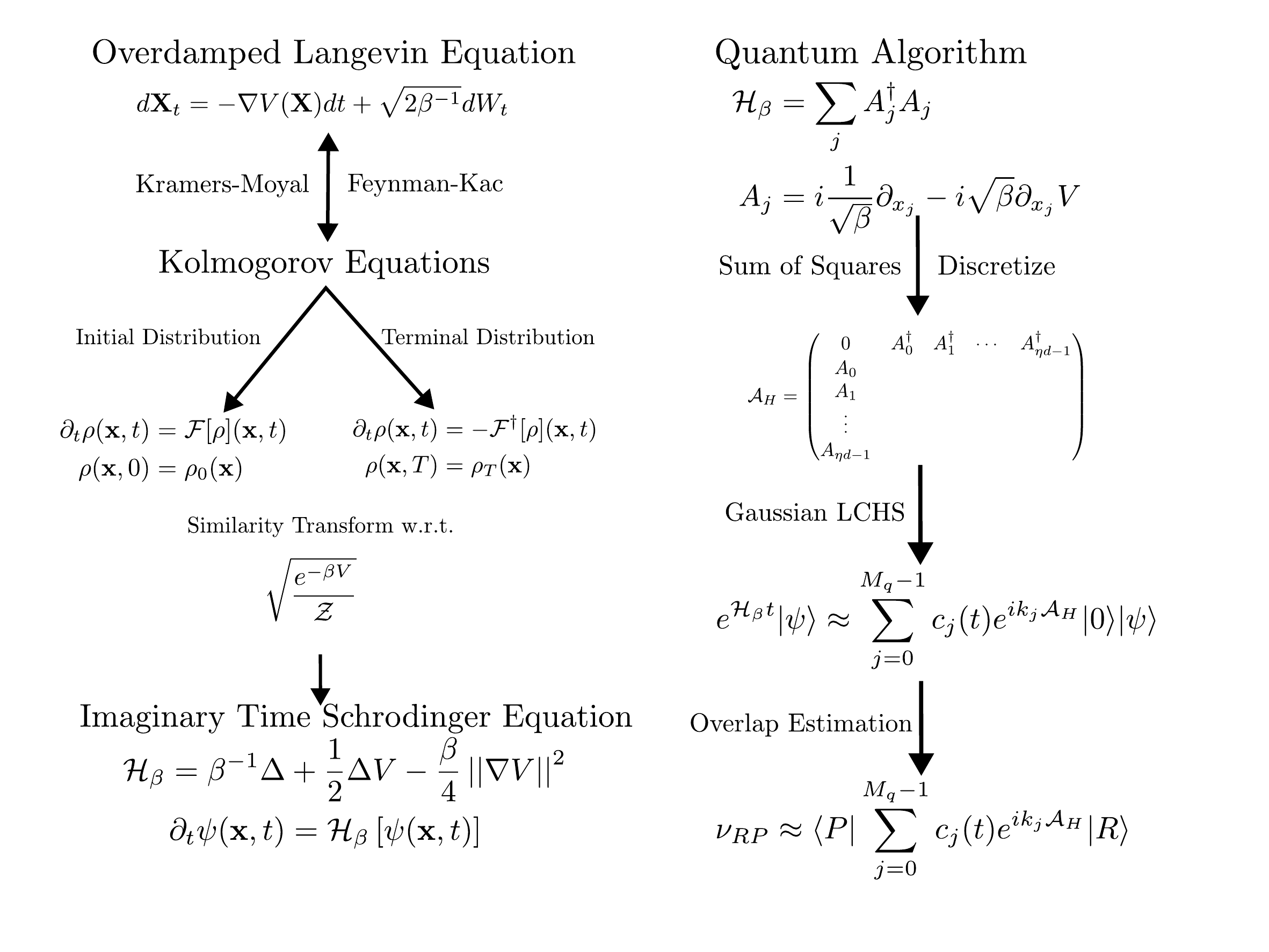}
    \caption{A graphical description of how the Hamiltonian {$\mathcal{H}_\beta$} is obtained 
    and how the {quantum} algorithm {to estimate the reactive flux $\nu_{RP}(t)$} is constructed. In the column on the left, a stochastic Langevin equation can be related to backward and forward Kolmogorov equations via the Feynman-Kac formula {and Kramers-Moyal expansion}, respectively. By performing a similarity transformation of the generator $\mathcal{F}$ with the square root of the equilibrium density, we obtain an imaginary time Schrodinger equation with Hamiltonian $\mathcal{H}_\beta$. This Hamiltonian can then be expressed as a sum-of-squares using the decomposition given by operators $A_j$ that we define below. Then, dilating into the Hermitian matrix $\mathcal{A}_H$, and performing Hamiltonian simulation with $\mathcal{A}_H$, we can approximate time evolution under $\mathcal{H}_\beta$ via the Gaussian-LCHS formula. Finally, we can apply the non-unitary overlap estimation algorithm to estimate the reactive flux, $\bra{P}e^{\mathcal{H}_\beta t}\ket{R}$.  }
    \label{fig:reductions}
\end{figure}

\subsection{Matrix square root and Gaussian-LCHS}
\label{subsec:sqrt-gauss-lchs}

Broadly speaking, the recently developed LCHS and related methods \cite{anLinearCombinationHamiltonian2023,anQuantumAlgorithmLinear2023a,pocrnicConstantFactorImprovementsQuantum2025,jinQuantumSimulationPartial2022} show how to construct a function closely related to the Fourier transform of $e^{-xt}$ that has exponentially rapid decay in Fourier space. 
Intuitively, just as the Feynman path integral provides a description of quantum states as a superposition over classical trajectories, the LCHS method {can be regarded as providing} a {quantum} description of classical dynamics {by} a linear combination over quantum trajectories. 
In this work we consider a scenario where we have access to a matrix $\mathcal{A} = \mathcal{A}^\dagger$ {such} that $-\mathcal{A}^2 = H$, for $H = H^\dagger \preceq 0$, up to an embedding into a larger block-diagonal matrix, i.e., $\mathcal{A}^2$ may be of the form
\begin{equation}
    -\mathcal{A}^2 =\begin{pmatrix}
        H & 0\\
        0 & *
    \end{pmatrix},
    \label{eq:A^2-simple}
\end{equation}
with $*$ any square matrix and the $0$ entries corresponding to rectangular matrices {with} zeros of appropriate dimension. For the Fokker-Planck generator in self-adjoint form, {Eqs.~\eqref{eq:FP_selfadjoint} - \eqref{eq:sos-form-H},} the dilated matrix square root {$\mathcal{A}$} takes the form
\begin{equation}
    \mathcal{A} = \begin{pmatrix}
        0 & A_0 ^\dagger & \cdots &A_{d\eta-1}^\dagger \\
        A_0 \\
        \vdots\\
        A_{d\eta -1}
    \end{pmatrix},
\end{equation}
with the $A_j$ given by Eq. \eqref{eq:sos-form-H}.
Notice that this is distinct from the idea of a \textit{block encoding}, where the off-diagonal entries are unknown quantities that ensure the overall unitarity of the block encoding. 
{Here it} is the block \textit{diagonal} structure of $\mathcal{A}^2$ that allows us to efficiently operate on $H$, since the top-left block of $\mathcal{A}^2$ is decoupled from the rest of $\mathcal{A}^2$. In the case of time-independent $H$ and $\mathcal{A}$, the LCHS approach naturally translates to this scenario, and efficient access to $\mathcal{A}$ can provide a significant speedup over the input model for standard LCHS. We call this approach \textit{Gaussian-LCHS}, where the usage of the phrase \textit{Gaussian} is motivated below.

We begin with the observation that the Fourier transform of the function $f_t(x) = e^{-x^2 t}$ is
\begin{equation}
    \hat{f_t}(k) = \frac{e^{-k^2/4t}}{\sqrt{2t}}.
\end{equation}
By assumption, $H$ and $\mathcal{A}$ are time-independent, and therefore the eigenvalues $x$ of $\mathcal{A}$ have no explicit dependence on $t$. Therefore the $t$ dependence of the propagator $e^{-\mathcal{A}^2t}$  can be absorbed into the kernel function $\hat{f_t}(k)$ and as a result, we have the relation 
\begin{equation}
    e^{-x^2 t} = \frac{1}{\sqrt{2\pi}}\int_{\mathbb{R}} \hat{f_t}(k) e^{-ikx}dk.
    \label{eq:FTgauss}
\end{equation}
As with standard LCHS we truncate the domain of integration to an interval of finite width,
\begin{equation}
    \widetilde{f}_t(x) = \int_{-L_G}^{L_G} \hat{f_t}(k) e^{-ikx}dk
    ,
    \label{eq:LCHS-int}
\end{equation}
where $L_G \in \mathbb{R}^+$ is the truncation wavenumber. We show in Appendix \ref{app-subsec:trunc-error-bound} that $L_G$ can be chosen $O\left(\sqrt{t\log(1/\epsilon)}\right)$.

Exploiting the relationship that $(\bra{0}\otimes I) (-\mathcal{A}^2) (\ket{0}\otimes I) = H$, where the dimension of $\ket{0}$ is the dimension of $H$ and the dimension $I$ is the remainder, we can approximate $e^{tH} = (\bra{0}\otimes I)e^{-t\mathcal{A}^2 }(\ket{0}\otimes I) $ using queries to a block encoding of $\mathcal{A}$, denoted $U_{\mathcal{A}}$. 
Since $\mathcal{A}$ is Hermitian, the spectral theorem ensures that $\mathcal{A}$ has real eigenvalues, and the Fourier transform formula Eq. \eqref{eq:FTgauss} for $x \in \mathbb{R}$ extends to the matrix-valued case by simply allowing $x$ to be any eigenvalue of $\mathcal{A}$. To obtain a discrete sum, we follow the LCHS approach and apply a quadrature scheme to approximate Eq. \eqref{eq:LCHS-int}.
Following these steps, we obtain
\begin{equation}
    e^{-\mathcal{A}^2t} \approx \frac{1}{2 \sqrt{t\pi}}\int_{-L_G}^{L_G}  e^{-k^2/4t} e^{-ik\mathcal{A}}dk\approx  \sum_{j = 0}^{M_q-1}c_j(t) e^{-i k_j \mathcal{A} },
    \label{eq:LCHS-gauss}
\end{equation}
where $M_q \in O\left(L_G\right)$ is the number of quadrature points, and the coefficients $c_j(t) = w_j e^{-k(j)^2/4t}$ with $w_j$ the quadrature weights. This is {derived} in Appendix \ref{app-subsec:quad-error-bound}. In Appendix \ref{app-subsec:bound-alpha-g} we show that $\alpha_g \equiv \sum_{j} |c_j(t)| \in O(1)$.

 When access to $\mathcal{A}$ is provided {via} a block encoding and $\alpha$ is the block encoding subnormalization factor, we can employ the Jacobi-Anger expansion, which provides a rapidly convergent polynomial approximation to $\exp(\pm i k x)$ where the required polynomial degree scales as $O\left( k \alpha +\log(1/\epsilon) \right)$ \cite{lowOptimalHamiltonianSimulation2017}. Using quantum signal processing (QSP) we can {then} efficiently obtain a block encoding of the associated matrix exponential $\exp(\pm i k \mathcal{A})$. The overall error introduced from approximating the unitaries in the Gaussian-LCHS formula Eq. \eqref{eq:LCHS-gauss} to finite precision is analyzed in Appendix \ref{app-subsec:sim-error-bounds}. It is shown {there} that the query complexity {for} the block encoding of $\mathcal{A}$ to approximate any term in the numerical quadrature scheme is upper bounded by
\begin{equation}
    Q =O\left(\alpha\sqrt{t\log\left(\frac{1}{\epsilon}\right)} + \log\left(\frac{\alpha}{\epsilon}\right)\right).
    \label{eq:poly-appx-unitary-bound}
\end{equation}
Therefore, using Gaussian-LCHS with the Hamiltonian $\mathcal{A}$, we obtain an approximation to 
\begin{equation}
    e^{-\mathcal{A}^2 t} = \begin{pmatrix}
        e^{H t} & 0 \\
        0 & *  
    \end{pmatrix},
\end{equation}
using $\widetilde{O}\left(\alpha \sqrt{t\log\left(\frac{1}{\epsilon}\right)}\right)$ queries to the block encoding of $\mathcal{A}$. Combined with the lemmas of the previous sections, this leads to the following theorem, the proof of which is given Appendix \ref{app-subsec:sim-error-bounds}.

\begin{restatable}{thm}{thmGaussLCHS}
    \label{thm:gauss-lchs-main}
    Let $\mathcal{A}$ be an $\mathcal{N}\times \mathcal{N}$ Hermitian matrix satisfying 
    \begin{equation}
        -\mathcal{A}^2 = H,
    \end{equation}
    or
    \begin{equation}
        -\mathcal{A}^2 = \begin{pmatrix}
            H & 0 \\
            0 & *
        \end{pmatrix}
    \end{equation}
    for $H$ an $N \times N$   negative definite or negative semi-definite Hermitian matrix with $N\leq \mathcal{N}$. Let $U_{\mathcal{A}}$ be a block encoding of $\mathcal{A}$ with subnormalization factor $\alpha$ and let $t>0$. 
    Then, for any $0< \epsilon <\min\{\frac{1}{\alpha},1\}$, there exists a quantum algorithm that queries $U_{\mathcal{A}}$ 
    \begin{equation}
         O\left(\alpha\sqrt{t\log\left(\frac{1}{\epsilon}\right)} +  \log\left(\frac{\alpha}{\epsilon}\right)\right)
    \end{equation}
    times {to produce} a block encoding of a matrix $\widetilde{e^{H t}}$ satisfying
    \begin{equation}
        \norm{e^{H t} - \widetilde{e^{H t}}} \leq \epsilon.
    \end{equation}
    Moreover, the block encoding can be implemented using $O\left(\log\left(\alpha \sqrt{t\log(1/\epsilon)}\right)\right)$ ancilla qubits and $O\left(\alpha \sqrt{t\log(1/\epsilon)}\right)$ Toffoli or simpler gates in addition to those used in the block encoding of $\mathcal{A}$.
\end{restatable}

Applying this LCU to any tensor product state of the form $\ket{0}\ket{x}$, with $\ket{0}$, $\ket{x}$ of appropriate {dimensions}, we will have  $e^{-\mathcal{A}^2 t}\ket{0}\ket{x} = \ket{0}e^{H t}\ket{x}$. We note that although each of the individual unitaries that make up the Gaussian-LCHS are realized with high probability, it is still the case that the preparation of a quantum state proportional to $e^{-\mathcal{A}^2t}\ket{0}\ket{x}$ can have vanishingly small success probability, as {discussed in detail} in the introduction.

\subsection{Non-unitary overlap circuit}
\label{subsec:overlap-circuit}
In this section we will explain the non-unitary overlap estimation circuit and the main operations that we perform in the algorithm. We will assume access to the following state preparation oracles which prepare normalized quantum states encoding the fixed ``reactant'' and ``product'' regions of the configuration space, $R, P \subset [0,L]^{\eta d}$. The corresponding oracles prepare the quantum states
\begin{equation}
    \begin{aligned}
        O_R\ket{0}_\sys &= \frac{1}{\sqrt{\overline{p}_R}}\sum_{x\in R} \rho(x)\ket{x} \equiv \ket{R}\\
        O_P\ket{0}_\sys &= \frac{1}{\sqrt{\overline{p}_P}}\sum_{x\in P}\rho(x)\ket{x}\equiv \ket{P},
    \end{aligned}
\end{equation}
where $\rho(x) = \frac{e^{-\beta V(x)/2}}{\sqrt{\mathcal{Z}}}$.  {We shall also assume access to the controlled versions of both oracles (see Fig.~\ref{fig:overlap-circuit}) below.}

For the purposes of demonstrating the main idea, in this section we will simply assume access to unitaries implementing $\{e^{-ik_j \mathcal{A}}\}$ for $j\in \{0,\ldots,M_q-1\}$ for $\mathcal{A}$ corresponding to the 
dilated notion of 
matrix square root for the discretized $\HB$ operator given by Eq. \eqref{eq:A^2-simple}. 
We will detail the implementation in later sections.
Using the Gaussian-LCHS formulae, we obtain an expression for the non-unitary propagator in terms of a linear combination of unitaries,
\begin{equation}
    e^{\HB t} \approx \sum_{j = 0}^{M_q-1}c_j(t) e^{-i k_j \mathcal{A} }.
\end{equation}
For notational convenience we will write $U_j \equiv e^{-i k_j \mathcal{A} T}$ to refer to {the} $j$th unitary, {where} $k_j \in [-L_G,L_G]$ and $c_j\in\mathbb{C}$ is the $j$th coefficient obtained from applying the quadrature scheme to the integrand.

\begin{figure}[ht]
    \centering
    \begin{quantikz}
        \ket{0}&\gate{H}   &&\ctrl{0}\vqw{2}&\octrl{0}\vqw{1}&\ctrl{0}\vqw{1}      &\octrl{0}\vqw{1}&&\gate{H}&\meter{}\\
        \ket{0}&\qwbundle{n}&&               &\gate{O_R}   &\gate{O_P}        &\gate[2]{\textsc{sel}_{lc}}&&&\\
        \ket{0}&\qwbundle{k}&\gate{\textsc{prep}_{lc}}       &\gate{\textsc{phase}}&        &&&&
    \end{quantikz}
    \caption{Quantum circuit for non-unitary overlap estimation. This circuit uses the standard LCU  \textsc{prep} {construction, a} controlled version of the standard \textsc{sel} routine, and an optional controlled \textsc{phase} oracle for complex and non-positive real coefficients. In conjunction with controlled state preparation, {this circuit allows us to} approximate $\frac{1}{\alpha_g}\Re{\bra{P}e^{t \HB}\ket{R}}$.}
    \label{fig:overlap-circuit}
\end{figure}

Using the standard oracles from the LCU subroutine, we define
\begin{equation}
\begin{aligned}
    \textsc{prep}_{lc}:\ket{0} \rightarrow \sum_{j}\sqrt{\frac{|c_j|}{\alpha_g}}\ket{j},\hspace{.2cm} \textsc{phase}:\ket{j}\rightarrow e^{-i\phi_j}\ket{j}
\end{aligned}
\end{equation}
where $\phi_j$ satisfies $c_j = |c_j|e^{i\phi_j}$, and the select oracle,
\begin{equation}
    \textsc{sel}_{lc} = \sum_{j} U_j\otimes \ket{j}\bra{j}.
\end{equation}
When all {coefficients $c_j$} in the formula are positive, the controlled-$\textsc{phase}$ operation is unnecessary. As shown in Fig. \ref{fig:overlap-circuit}, we make a simple modification to the Hadamard test and LCU circuits to obtain a quantum circuit {that} can be used for estimating overlaps of the non-unitary propagator. Analysis of this circuit shows that the expectation value of the Pauli-$Z$ matrix on the top qubit satisfies
\begin{equation}
    \bra{\Phi}Z\ket{\Phi} = \frac{1}{\alpha_g}\Re{\bra{P}e^{\HB t}\ket{R}}.
    \label{eq:overlap result}
\end{equation}
We provide a detailed derivation of this result in Appendix \ref{app-sec:non-unitary overlap circuit}.

\begin{restatable}{lemma}{nonUnitOverlapLem}
    \label{lem:non-unitary-overlap-main}
    Let $H$ be represented by the LCU
    \begin{equation}
        H = \sum_{l=0}^{K-1}|c_l|e^{i\phi_l}U_l,
    \end{equation}
    with $2^{k} \leq K \leq 2^{k+1}$ for some $k \geq 1$, and $\alpha = \sum_{l}|c_l|$. For any input states $\ket{\phi}$ and $\ket{\psi}$ constructed with oracles $O_\psi$ and $O_\phi$, the quantum circuit in Fig. \ref{fig:overlap-circuit} can be used to estimate 
    \begin{equation}
        \bra{\phi}\frac{H}{\alpha}\ket{\psi},
    \end{equation}
    independent of the success probability of implementing the LCU, provided the $U_l$ are  unitaries.
\end{restatable}

The circuit in Fig.~\ref{fig:overlap-circuit} succeeds with $\Omega(1)$ probability. This is because the unitaries $U_j$ in the LCU can be implemented with arbitrarily high success probability, since the circuit does not require postselection on the ancilla qubits used to construct the linear combination. This should be contrasted with approaches that attempt to prepare the normalized state $e^{t\HB}\ket{R}/\|e^{t\HB}\ket{R}\|$ directly: for dissipative generators, $\|e^{t\HB}\ket{R}\|$ decays exponentially in $t$, incurring exponential amplitude-amplification overhead as {discussed} in Sec.~\ref{subsec:sqrt-gauss-lchs}.

The $\Omega(1)$ circuit success probability should not, however, be confused with the magnitude of the overlap $\nu_{RP}(t) = \bra{P}e^{t\HB}\ket{R}$ itself. When the reaction is rare on the simulated time horizon, $\nu_{RP}(t)$ can itself be exponentially small in $t$: {in this situation} Eq.~\eqref{eq:sym-rxn-rate} {shows that} the reaction rate $k_{RP}(t) \propto \nu_{RP}(t)$ {will also be exponentially small}.
This is a faithful reflection of the underlying physics, {namely, that} the reaction is genuinely rare, and our {quantum} circuit honestly reports this. 
{Crucially,} an exponentially small $\nu_{RP}$ does \textit{not} mean {that} the circuit fails or must be repeated exponentially many times {in order} to produce a valid sample {quantum trajectory} -  {indeed,} the $\Omega(1)$ success probability of the circuit is entirely independent of the physical magnitude of $\nu_{RP}$. 
{This} distinction is {very} important: {a} small success probability means {that} the algorithm needs exponentially many repetitions to yield even one good sample; {in contrast, a} small $\nu_{RP}$ means {that} the estimated physical quantity, {in this case the estimated reaction rate}, is small. Our algorithm avoids the former; any smallness of $\nu_{RP}$ reflects the latter — the physics of the problem — {and} not a limitation of the {quantum} algorithm.

If instead our goal was to target \textit{relative}-error estimation of $\nu_{RP}$, the sample complexity {would scale} with the magnitude of the overlap, namely as $O((\nu_{RP}\epsilon)^{-1})$, whereas classical approaches would scale as $O\left((\nu_{RP}\epsilon)^{-2}\right)$. 
The structural advantage of our algorithm lies in {efficient} \textit{additive}-error estimation. {Specifically,} the $\Omega(1)$ success probability, combined with quantum amplitude estimation, yields $O(1/\epsilon)$ sample complexity \textit{independent} of the magnitude of {the reactive flux, i.e., the reaction rate} $\nu_{RP}(t)$.
Furthermore, when prior knowledge establishes that $\nu_{RP} = O(\epsilon)$ — as in the deeply metastable regime where the reaction probability within finite time $t$ is known to be small — the amplified amplitude estimation method of Ref.~\cite{simonAmplifiedAmplitudeEstimation2024} can achieve $O(1/\sqrt{\epsilon})$ scaling, an additional quadratic speedup over standard amplitude estimation and a super-quadratic improvement over classical Monte Carlo ($O(1/\epsilon^2)$). This is the practically relevant regime for many kinetic questions, {such as} certifying that a rate falls below a safety threshold, comparing rates across catalysts or conditions, or establishing that a reaction occurs with non-negligible probability within a fixed time horizon. A quantitative discussion of this distinction, with explicit sample complexity bounds, is given in Sec.~\ref{subsec:overlap-estimation}.

\color{black}

\subsection{Discussion on reactant state preparation}
\label{subsec:state-prep}
For the computation of reaction rates, we have assumed access to state preparation oracles encoding the Boltzmann-weighted probability of being in the reactant and product regions of configuration space. In full generality, the preparation of these states can be as difficult as preparing the full Boltzmann distribution in the generic non-convex setting we describe. Accordingly, our complexity statements should be interpreted as conditional on the availability of these preparations. Rather than claiming efficient state preparation for arbitrary non-convex potentials, our goal is to identify a physically motivated regime in which preparing states of the form 
\begin{equation}
    \ket{R} = \int_{x \in R} \frac{e^{-\beta V(x)/2}}{\sqrt{\mathcal{Z}_R}}\ket{x}, \hspace{.2cm } \mathcal{Z}_R = \int_{x\in R} e^{-\beta V(x)} dx,
\end{equation}
is provably efficient and can be substantially easier than preparing a state encoding the global Gibbs distribution.

In many problems of practical interest, one wishes to study the behavior of transitions between metastable configurations of the system. Such configurations correspond to local minima of the potential energy surface. Mathematically, this corresponds to the Hessian of the potential in that region being locally positive. Assume that the restriction $\left. V(x)\right|_{x\in R}$ is $m$ strongly convex,  meaning
\begin{align*}
    \inf_{x \in R}\text{Hess}\left(V(x)\right) \succeq mI
\end{align*}
for some $m > 0$. The approach we use to obtain an efficient algorithm is to augment the locally convex portion of the potential with a confining potential, and {then solve} the problem with the augmented confining potential over the whole domain.
One natural choice of an augmented potential $V_R(x)$ is,
    \begin{equation}
    V_R(\mathbf{x}) = \begin{cases}
        V(\mathbf{x}) & \mathbf{x} \in R\\
        V(\mathbf{x}) + \kappa \,\text{dist}\left(\mathbf{x},\partial R\right)^2 & \mathbf{x} \in R^c,
    \end{cases}
    \label{eq:V-augment}
    \end{equation}
where the parameter $\kappa>0$ controls the accuracy of the approximation. This construction is demonstrated in Fig. \ref{fig:local-cvx-extension} below.

\begin{figure}[h]
    \centering
    \includegraphics[width=0.5\linewidth]{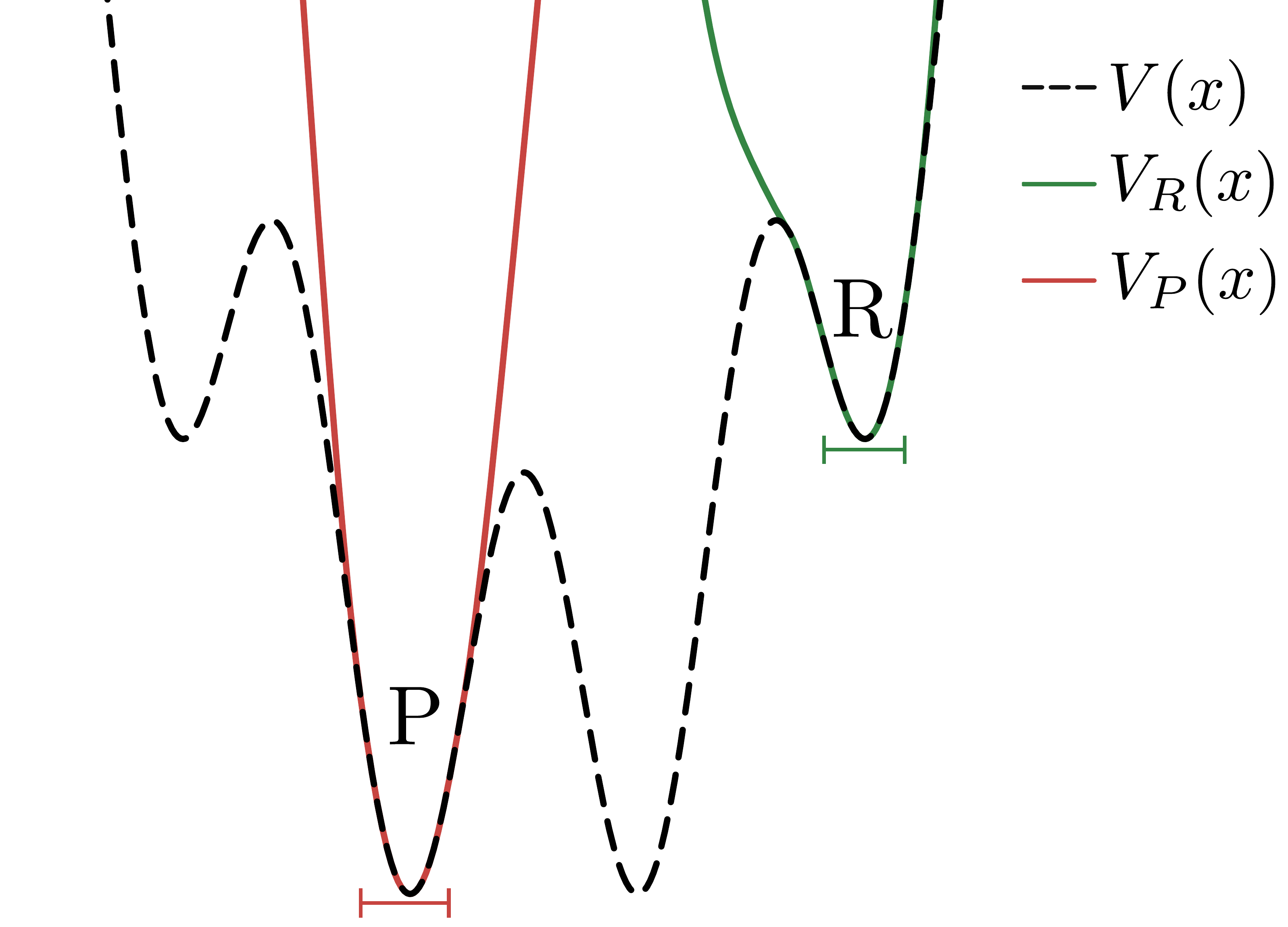}
    \caption{The potential $V(x) = \cos(2 \pi x) + x^2$ is shown in dashed black lines. The reactant and product regions, $R$ and $P$, are marked by the green and red lines under the graph respectively. The red and green curves {are generated  with a confining augmented potential having $\kappa = 20$: {these curves, denoted $V_R$ and $V_P$, respectively,}} preserve $V(x)$ over $R$ and $P$, and confine the particle trajectories to be within {these initial and final regions, respectively}. Under the dynamics generated by the Fokker-Planck equation with the convex potentials $V_R$ and $V_P$, the system decays exponentially to the equilibrium distribution of $V$ over $R$ and $P$ with controllable error determined by $\kappa$.}
\label{fig:local-cvx-extension}
\end{figure}
With these assumptions and the above construction, we can establish that there exists a quantum algorithm that can produce a quantum state encoding of the local equilibrium {state} $\ket{R}$ with polynomial overhead.

\begin{restatable}{thm}{localCvxPrep}{Fast thermal state preparation in locally convex region.}
    Let $R$ be a convex open subset of $[-L,L]^d$. Assume that,
    \begin{align*}
        \inf_{x \in R}\textup{Hess}\left(V(x)\right) \succeq mI.
    \end{align*}
    Assume the warm start condition that an initial state $\ket{R_0}$ satisfies $\langle R_0 | R \rangle \geq c$ for some constant $c$, where
    \begin{equation}
    \ket{R} = \begin{cases}
        \int_{x \in R} \frac{e^{-\beta V(x)/2}}{\sqrt{\mathcal{Z}_R}}\ket{x} & x\in R,\\
        0 & \text{else}
    \end{cases},
    \hspace{.4cm } \mathcal{Z}_R = \int_{x\in R} e^{-\beta V(x)} dx.
    \end{equation}
    Let $V_R$ be the augmented potential from Eq. \eqref{eq:V-augment}, $\alpha_V = \max_{r \leq \sqrt{d}L}\frac{\left|V'(r)\right|}{r}$, and $\kappa \in O\left(\frac{1}{\epsilon^4\mathcal{Z}_R^2 \beta}\right)$, satisfying 
    \begin{align*}
        \norm{e^{-\beta V_R/2} -  \ket{R}}_{L^2} \leq \epsilon.
    \end{align*}
    
    Moreover, there exists a quantum algorithm preparing the state $\ket{\widetilde{R}}$ such that 
    \begin{align*}
        \norm{\ket{R} - \ket{\widetilde{R}}}_{L^2} \leq \epsilon
    \end{align*}
    using
    \begin{align*}
        \widetilde{O}\left(\sqrt{\frac{\eta d}{2m}}\left(\beta \eta L (\alpha_V + \kappa) + \frac{N}{L}\right)\right),
    \end{align*}
    queries to the block encoding of $\mathcal{A}$, where the $\widetilde{O}$ notation suppresses polylogarithmic factors in $1/\epsilon$.
\end{restatable}
The proof of this theorem is {given} in Appendix \ref{app-sec:locally-cvx-stateprep}. In the end-to-end reaction-rate algorithm, these state-preparation routines are invoked $O(1/\epsilon)$ times; combined with our overlap estimation procedure, the dominant dependence on the time horizon $t$ arises through the $O\left(\sqrt{t\norm{H}}/\epsilon\right)$ queries to the block encoding of $\mathcal{A}$.

\section{Analysis and construction of core subroutines}
\label{sec:alg-analysis}
In this section we construct the core subroutines that are needed to realize the quantum circuit provided in Fig. \ref{fig:overlap-circuit}. First, we show how to block encode a generic matrix satisfying Eq. \eqref{eq:A^2-simple}. We then show how to block encode {the square root matrix} $\mathcal{A}$ for {Schr\"odinger-type operators $\mathcal{H}_{\beta}$ in the form of Eq.~\eqref{eq:FP_selfadjoint})} in first quantization, treating polynomial pair potentials directly and extending the construction to physically motivated singular pair potentials via a smoothed surrogate, and detail the block encodings of the individual terms that make up the block matrix encoding the square root. We then provide a scheme to efficiently construct the propagator using a multiplexed version of QSP combined with the Gaussian-LCHS method. To conclude this section, we bound the sampling costs as well as the sources of error in the method.

\subsection{Block encoding sum-of-squares Fokker-Planck operator}
\label{subsec:BE-A-schrod}

Recall the decomposition of the self-adjoint representation of the Fokker-Planck operator, {Eq.~\eqref{eq:FP_selfadjoint},} as 
\begin{equation}
\begin{aligned}
    \mathcal{H}_\beta &= -\sum_{ j \in [\eta]} A_j^\dagger A_j\\
    A_j &= -i\left(\frac{1}{\sqrt{\beta}}\nabla_j + \sqrt{\beta}F_j\right),
\end{aligned}
\end{equation}
where $F_j = -\nabla_j V$. We will discretize the continuum operators in $A_j$ using a plane-wave based spectral method \cite{Childs2020, Lubasch}. The differential operators become multiplication operators in the Fourier basis $\nabla_j \rightarrow -i k_l \ket{l}\bra{l}_j$, and the diagonal multiplication operator $F_j$ is treated by quantum Fourier transform. We will consider particles interacting via a radially symmetric pair-wise potential $V = \sum_{i\in[\eta]}\sum_{j<i}V_{ij}$, where $V_{ij} = V(r_{ij})$ and $r_{ij} = \norm{\mathbf{x}^i-\mathbf{x}^j}$. Under the assumption that the potential $V$ is confining, the probability of the system occupying a state near the boundaries of the simulation box is exponentially small. Therefore, we may truncate the domain to some finite length $L>0$ that depends on the growth rate of the potential as $\mathbf{x}\rightarrow \infty$ with minimal effect on the simulation accuracy, and enforce periodic boundary conditions so that the single-particle domain is defined over the $d-$dimensional torus $[-L,L]^d$.

We treat two classes of pair potentials in this section. Theorem~\ref{thm:BE-Aj-ops} addresses pair potentials $V_0(r)$ given by a polynomial of degree $2k$ with no degree-1 term, the natural setting for the plane-wave sum-of-squares construction. Pair potentials of physical interest --- Lennard-Jones, Morse, screened Coulomb --- are singular at the coincidence set $\{\mathbf{x}^i=\mathbf{x}^j\}$ and have non-confining tails, and so fall outside this class. To accommodate them, we introduce in Appendix~\ref{subsec:smoothed-surrogate} a smoothed surrogate $\widetilde{V}$ that matches the physical potential in the dynamically relevant region, is polynomial in $r^2$ near every coincidence point, and joins a polynomial confining wall near the box boundary; the surrogate's reactive flux approximates the physical one to controllable accuracy (Theorem~\ref{thm:singular-convergence}). Corollary~\ref{cor:BE-A-LJ} below extends Theorem~\ref{thm:BE-Aj-ops} to this case, and establishes that the asymptotic block-encoding cost is the same in both settings.
\color{black}

\begin{restatable}{thm}{BEAj}
    \label{thm:BE-Aj-ops}
    Let $N = 2^n$ be the number of plane wave modes per degree of freedom, for $\eta$ particles occupying a $d$-dimensional reciprocal lattice $\mathbf{G} = \mathbb{Z}^d_{N}$, corresponding to a discretization of the real space $[-L,L]^d$ for each particle. 
    Let $V$ be a potential function given by a polynomial of degree $2k$ in the pair-wise distance between any distinct pairs particles $i, j$, with no term of degree $1$, and let $\alpha_V = \max_{\norm{r}\leq \sqrt{d}L} \left|\frac{V'(r)}{r}\right|$.
    Then there exists a quantum circuit {that} block encodes the plane wave representation of the operator 
    \begin{equation}
        \mathcal{A} = \sum_{j \in [\eta]} \left| 0 \rangle \langle j+1\right|\otimes A_j^\dagger +\left| j+1 \rangle \langle 0\right|\otimes A_j
    \end{equation}
    with subnormalization factor
    \begin{equation}
        \alpha_A \in  O\left(\sqrt{\beta d} \eta^{3/2} L \alpha_V + \sqrt{\frac{\eta d}{\beta}} \frac{N}{L}\right),
    \end{equation}
    using 
    \begin{equation}
        \widetilde{O}\left(\eta d n + 2k d n^2 + (nd)^2\right)
    \end{equation}
    Toffoli or simpler gates, with the logarithmic factors hidden by the $\widetilde{O}$ notation arising from small overheads to compile finite precision rotation gates and multi-controlled operations.
\end{restatable}
The proof of {Theorem~\ref{thm:BE-Aj-ops}} can be found in Appendix \ref{app-sec:BE-A}.

\begin{cor}[Block encoding $\mathcal{A}$ for singular pair potentials via smoothed surrogates]
\label{cor:BE-A-LJ}
Let $V_0:(0,\infty)\to\mathbb{R}$ be a real-analytic pair potential that is singular at the origin and non-confining at infinity (Lennard-Jones, Morse, screened Coulomb, etc.), and let $\widetilde{V}$ be its smoothed surrogate (Definition~\ref{def:smoothed-surrogate}) with design parameters $(r_c, L', L, \kappa, p)$ chosen as in Theorem~\ref{thm:singular-convergence} so that the seam and wall Boltzmann weights both lie below the target precision. Let $\alpha_V = \max_{r \le \sqrt{d}L} \left|\widetilde{V}_0'(r)/r\right|$. Then there exists a quantum circuit that block encodes the dilated square root $\mathcal{A}$ built from $\widetilde{V}$, with subnormalization factor
\begin{equation}
    \alpha_{\mathcal{A}} \in O\left(\sqrt{\beta d}\, \eta^{3/2} L\, \alpha_V + \sqrt{\frac{\eta d}{\beta}}\, \frac{N}{L}\right),
\end{equation}
using
\begin{equation}
    \widetilde{O}\left( \eta d n + d n^2 + (nd)^2 \right)
\end{equation}
Toffoli or simpler gates. Asymptotically, the subnormalization factor and gate count match those of Theorem~\ref{thm:BE-Aj-ops}, with the polynomial-degree factor $2k$ replaced by the constant degree of the fixed Jacobi-Anger expansion used to block encode the inverse-power-law branch. The leading constant factors in the gate count are essentially identical; the only practical overhead relative to a polynomial potential of comparable bulk strength is a modest increase in $\alpha_V$ set by the short-range patch coefficients (in the Lennard-Jones case worked out in Appendix~\ref{subsec:LJ-numerics}, a factor of approximately ten at $r_c = 0.85\sigma$).
\end{cor}
The proof follows the same LCU construction as Theorem~\ref{thm:BE-Aj-ops}, with the polynomial gradient block encoding (Lemma~\ref{lem:BE-grad-V}) replaced by the smoothed-surrogate gradient block encoding constructed in Appendix~\ref{subsec:LJ-numerics}.

\subsection{Block encoding propagator with Gaussian-LCHS}
\label{subsec:block-encode-propagator}
We now show how to implement the LCHS using $\widetilde{O}(\alpha_A \sqrt{t})$ queries to the block encoding of $\mathcal{A}$, which we denote as $U_{\mathcal{A}}$. Recall the expression for the propagator obtained by the Gaussian-LCHS formula
\begin{equation}
    e^{-\mathcal{A}^2t} \approx \sum_{j=0}^{M_q-1} w_j f_t(k(j)) e^{-i k(j) \mathcal{A}},
\end{equation}
where $k(j) \in [-L_G,L_G]$, $j \in \{0,\ldots, M_q-1\}$ is the $j$th quadrature point, $w_j$ is the $j$th weight in the Gauss quadrature, and $ f_t(k(j))= \frac{e^{-\frac{k^2(j)}{4t}}}{2\sqrt{\pi t}}$. This requires forming block encodings of each of the unitary matrices $\mathcal{U}_j =e^{-i k(j) \mathcal{A}}$. We use the Jacobi-Anger expansion to obtain a polynomial approximation to $e^{i\mathcal{A}k}$ for $k \in [-L_G, L_G]$:
\begin{equation}
    e^{i k \mathcal{A} } = J_0(k) + 2\sum_{l \in {\text{evens}}>0}(-1)^{l/2}J_l(k)T_{l}(\mathcal{A}) + 2i\sum_{l \in {\text{odds}}>0}(-1)^{(l-1)/2}J_l(k)T_{l}(\mathcal{A}),
    \label{eq:jacobi-anger}
\end{equation}
{with $J_l$ and $T_l$ the $l$th Bessel function of the first kind and $l$th Chebyshev polynomial of the first kind respectively}.
Therefore, $k(j)$, which can be thought of as the evolution time, is a function of the \textit{coefficients} in the polynomial approximation, given by the Bessel $J$ functions. For a fixed $k$ and $\epsilon$, the polynomial degree $D_k$ is chosen according to the bound
\begin{equation}
    D_k = O\left( k\alpha_A + \log\left(\frac{1}{\epsilon}\right)\right).
\end{equation}
We express the polynomial obtained by the Jacobi-Anger expansion abstractly as 
\begin{equation}
    P_k(\mathcal{A}) = \sum_{l=0}^{D_k-1}  c_l(k) T_l(\mathcal{A}) \approx e^{i k \mathcal{A}},
\end{equation}
with the $c_l(k)$'s 
{given explicitly} by \eqref{eq:jacobi-anger}.

If we were to naively apply QSP to each term in the LCHS formula independently, we would require $\widetilde{O}\left(k(j)\alpha_A\right)$ queries to the block encoding of $\mathcal{A}$ for each term. Since each term in the LCU has additive cost, this would entail {an overall} cost of $\sim \sum_{j =0}^{M_q} k(j) \in O\left(L_G^2\right) \in O\left(t \log\left(\frac{1}{\epsilon}\right)\right)$. 
{We will now} demonstrate a block encoding of the above LCHS that uses {only} $O(L_G)$ queries to $U_\mathcal{A}$ and additionally does not require any controlled applications of $U_\mathcal{A}$, using a technique {that} we call \textit{multiplexed QSP}. 

\begin{restatable}{lemma}{mpxqsp}
\label{lem:mpx-qsp}
Let $U_A$ be an $(\alpha, a)$ block encoding of $A \in \mathbb{C}^{2^n\times 2^n}$. Let $\{P_{0}, \ldots, P_{M-1}\}$ be a set of polynomials of degrees $d_0, \ldots, d_{M-1}$ respectively satisfying 
\begin{enumerate}
    \item $\left|P_{i}(x)\right|^2 \leq 1$ $\forall x \in [-1,1]$ and $\forall i \in \{0,\ldots, M-1\}$
    \item $\forall i \in \{0,\ldots,M-1\}$,  $d_i \equiv p \mod{2}$ for fixed $p \in \{0,1\}$.
\end{enumerate}
Then, for any $\ket{c} = \sum_{j=0}^{M-1}c_j\ket{j}$ satisfying $\langle c| c\rangle  = 1$, and any $n$ qubit quantum state $\ket{\psi}$, there exists a quantum circuit using $d_{\max} := \max\{d_0, \ldots, d_{M-1}\}$ queries to $U_A$ and $\widetilde{O}(M)$ additional Toffoli or simpler gates to implement the quantum state
\begin{equation}
    \sum_{j=0}^{M-1} c_j P_j\left(\frac{A}{\alpha}\right)\ket{\psi}\ket{0}_{a+1}\ket{j} + \ket{\perp}.
\label{eq:mpx-qsp-state}
\end{equation}
\end{restatable}
\noindent  The proof of this lemma is provided in Appendix \ref{app-sec:pf-mpx-qsp}.

Remarkably, multiplexed QSP does not increase the overall query complexity to the block encoding, nor does it require controlled applications of the block encoding. Indeed, the total query complexity to $U_{\mathcal{A}}$ for constructing the linear combination over $U_{P_{k(j)}(\mathcal{A})}$ is $D_{\max}$ {(see Eq.~\eqref{eq:Dmax} below)} and is \textit{independent} of $M_{q}$, the number of 
terms in the linear combination. Nevertheless, there is still a dependence of $O(M_q \log(M_q))$ additional Toffoli gates needed to implement the control logic, which is subdominant to the block encoding cost. 

Since $k(j) \leq L_G \in O\left(\sqrt{t \log(1/\epsilon)}\right)$ by Lemma \ref{lem:trunc_error} in Appendix \ref{app-subsec:trunc-error-bound}, the maximum polynomial degree that is needed to approximate any of  {the operators} $\{e^{ik(j) \mathcal{A}}\}_{j=0}^{M_q-1}$ is 
\begin{equation}\label{eq:Dmax}
    D_{\max} = O\left(\alpha_A \sqrt{t\log(1/\epsilon)} + \log\left(\frac{1}{\epsilon}\right)\right).
\end{equation}
For each of the polynomials $P_{k(j)}\approx e^{i k(j) \mathcal{A}}$, there is an associated set of phase factors $\boldsymbol{\Phi}^j \in \mathbb{R}^{D_{k(j)} +1}$, where $j \in \{0,\ldots, M_q-1\}$. Let $\Phi_i^j\in\mathbb{R}$ be the $i$th phase factor associated with the $j$th polynomial. Append a register of $q = \ceil{\log(M_q)}$ ancilla qubits and let $\mathcal{S}^j_i = e^{iZ_\Pi {\Phi}^j_i}$ be provided as an oracle which implements the  $i$th rotation corresponding to the coefficients for the $j$th polynomial. The {resulting} quantum circuit implementation is provided as Fig. \ref{fig:mpx-QSP}.

\begin{figure}[h!]
    \centering
    \begin{quantikz}
        \lstick{$\ket{0}$}&\qwbundle{q}&\gate{\textsc{prep}_{g}}&\oslash\vqw{1}&&\oslash\vqw{1}&&\oslash\vqw{1} &\cdots&\oslash\vqw{1} &&\oslash\vqw{1}&\\
        \lstick{$\ket{0}$}&&&\gate{\mathcal{S}_0}\vqw{1}&&\gate{\mathcal{S}_1}\vqw{1}&&\gate{\mathcal{S}_2}\vqw{1} &\cdots&\gate{\mathcal{S}_{D_{\max}-1}}\vqw{1} &&\gate{\mathcal{S}_{D_{\max}}}\vqw{1}&\\
        \lstick{$\ket{0}$}&\qwbundle{anc}&&\octrl{0}&\gate[2]{O_\mathcal{A}}&\octrl{0}&\gate[2]{O_\mathcal{A}}&\octrl{0} &\cdots&\octrl{0} &\gate[2]{O_\mathcal{A}}&\octrl{0}&\bra{0}\\
        \lstick{$\ket{\psi}$}&\qwbundle{sys}&&&&&&& \cdots& & &\sum_{j=0}^{M}\sqrt{c_j}e^{ik(j)\mathcal{A}}\ket{\psi}\ket{j}
    \end{quantikz}
        \caption{Quantum circuit description of multiplexed QSP using the rotated block encoding $O_\mathcal{A} = U_\mathcal{A} Z_\Pi$. This circuit allows one to evaluate multiple polynomials of the block encoding of $\mathcal{A}$ simultaneously by encoding the phases for the corresponding polynomial into the phase oracle $\mathcal{S}$, so that applying $\mathcal{S}^j_i = e^{i \phi^j_i Z_\Pi}$ to the single ancilla qubit implements the $i$th phase factor for the $j$th polynomial. Remarkably, this circuit requires no controlled applications of the block encoding, and precisely $D_{\max} = \max_{l}\{D_l\}$ queries to $O_\mathcal{A}$. {Here, }  $q =\ceil{\log(M_q)}$ is the number of ancilla qubits needed to realize the coefficients for the Gaussian-LCHS formula.}
    \label{fig:mpx-QSP}
\end{figure}

Let $\alpha_g = \sum_{j=0}^{M_q-1} |c_j(t)| \leq 1 + O(\epsilon)$ be the subnormalization factor associated with the quadrature function $c_j(t)$, which is given by the Gaussian-LCHS coefficients $c_j(t) = w_jL_G\hat{f}_t(k(j))$, with $w_j$ the quadrature weight. Construct a superposition over the $q = \ceil{\log(M_q)}$ qubit ancilla register,
\begin{equation}
    \textsc{prep}_{g}:\ket{0}_q\rightarrow \sum_{j=0}^{M_q-1} \sqrt{\frac{c_j(t)}{\alpha_g}}\ket{j}_q,
\end{equation}
and apply the quantum circuit from Fig. \ref{fig:mpx-QSP} to obtain
\begin{equation}
\begin{aligned}
\ket{0}_q\ket{0}_{m+1}\ket{\psi}_{n} &\rightarrow \sum_{j=0}^{M_q-1}\sqrt{\frac{c_j(t)}{\alpha_g}}\ket{j}_q\left(\ket{0}_{m+1}P_{k(j)}(\mathcal{A})\ket{\psi}_n + \ket{\perp}\right)\\
&\approx \sum_{j=0}^{M_q-1}\sqrt{\frac{c_j(t)}{\alpha_g}}\ket{j}_q\left(\ket{0}_{m+1}e^{i k(j) \mathcal{A}}\ket{\psi}_n + \ket{\perp}\right).
\end{aligned}
\end{equation}

By inspecting the circuit in Fig. \ref{fig:mpx-QSP} it is clear that the overall query complexity to $U_{\mathcal{A}}$ is 
\begin{equation}
    D_{\max} = O\left(\alpha_A \sqrt{t\log(1/\epsilon)}\right),
\end{equation}
and the number of ancilla qubits used in addition to those used to construct $U_{\mathcal{A}}$ is $\ceil{\log(M_q)}$.
Since the unitaries $\mathcal{U}_j$ can be implemented with success probability arbitrarily close to 1, the quantum circuit in Fig. \ref{fig:mpx-QSP} prepares a \textit{superposition} of Hamiltonian simulations, which is distinct from a general linear combination, {since} the quantum state {that} we wish to sample from at the end of the circuit is approximately $l_2$ normalized.

\begin{thm}[Complexity to perform Gaussian-LCHS]
\label{thm:gauss-lchs-complexity}
    Let $\mathcal{A}$ be given by its $(\alpha_A, m)$ block encoding, and let $t>0$. Then, for any $\epsilon>0$, there exists a quantum algorithm that prepares an $\epsilon$ approximation to the quantum state
    \begin{equation} \label{eq:quantumstatel2norm}
    \sum_{j=0}^{M_q-1}\sqrt{\frac{c_j(t)}{\alpha_g}}\ket{j}_q\left(\ket{0}_{m+1}e^{i k(j) \mathcal{A}}\ket{\psi}_n + \ket{\perp}\right),
    \end{equation}
    in the $l_2$ norm, using $O\left(\alpha_A \sqrt{t \log\left(\frac{1}{\epsilon}\right)}+\log\left(\frac{1}{\epsilon}\right)\right)$ queries to $\mathcal{A}$ and $O(M_q\log(M_q) + \log\left(\frac{M_q}{\epsilon}\right))$ additional one and two qubit gates, {with a} normalization factor $\alpha_g \in O(1)$ by Lemma \ref{lem:gaus-lchs-subnorm}.
\end{thm}
\begin{proof}
    Use Hamiltonian simulation to obtain an $\epsilon$ block encoding of $e^{i k(j) \mathcal{A}}$ in the spectral norm. Since $k(j) \leq L_G \in O\left(\sqrt{t \log\left(\frac{1}{\epsilon}\right)}\right)$, this can be obtained with QSP and a polynomial of degree $O\left(\alpha_A\sqrt{t\log\left(\frac{1}{\epsilon}\right)} + \log\left(\frac{1}{\epsilon}\right)\right)$. By the construction of Lemma \ref{lem:mpx-qsp}, the maximum polynomial is also the number of queries needed to form the superposition from above. The additional $O(M_q \log(M_q) + \log\left(\frac{M_q}{\epsilon}\right))$ one and two qubit gates 
    {result} from the state preparation of the coefficient state as well as the control logic needed to implement the multiplexed operations in the circuit of Fig. \ref{fig:mpx-QSP}. This ensures that the {quantum state Eq.~\eqref{eq:quantumstatel2norm}} is $\epsilon$ close in $l_2$ norm and suffices to prove the theorem. 
\end{proof}

\begin{cor}[Gate complexity of Gaussian-LCHS applied to Fokker-Planck equation]
\label{cor:gauss-lchs-FPE}
    Let $\mathcal{H}$ be the self-adjoint representation of the Fokker-Planck operator, and let $t, \epsilon>0$. Assume that each particle is represented using {$nd$ 
    qubits}, corresponding to $N^d$ plane wave modes. Assume that $L,\alpha_V,d \in O(1)$, and $\sqrt{\beta} \gg \sqrt{\beta}^{-1}$. Then the superposition of Hamiltonian simulations can be implemented using $O\left(\left(\sqrt{\beta}\eta^{5/2} n + \eta^{3/2} n^2\right)\sqrt{t\log\left(\frac{1}{\epsilon}\right)}\right)$ Toffoli or simpler gates. 
\end{cor}
\begin{proof}
    Apply Theorem \ref{thm:gauss-lchs-complexity} to the block encoding from Theorem \ref{thm:BE-Aj-ops}, and substitute the bounds for $\alpha_A$ and the cost to perform the block encoding from the costs reported in Theorem \ref{thm:BE-Aj-ops}.
\end{proof}

\subsection{Overlap estimation algorithm}
\label{subsec:overlap-estimation}
With all of the core subroutines constructed, we now analyze the complexity of estimating the {propagator} overlaps {needed to evaluate the rate constant,} using the non-unitary overlap estimation circuit in Fig. \ref{fig:overlap-circuit} to approximate the reactive flux 
\begin{equation}
    \nu_{RP} \propto \bra{P}e^{t\HB }\ket{R}.
\end{equation}
Post-selecting on the ancilla qubits used to form the block encoding of the unitary operation $e^{-i k(j)\mathcal{A}}$, the output of this circuit is a quantum state of the form
\begin{equation}
\ket{\Phi} := \frac{1}{2}\left(\sum_{l}\sqrt{\frac{c_l}{\alpha_g}}\ket{0}\ket{0}_{sys}\left(e^{-ik(l)\mathcal{A}}\ket{R}_{sys} + \ket{P}_{sys}\right)\ket{l}_k + \sum_{l}\sqrt{\frac{c_l}{\alpha_g}}\ket{1}\ket{0}_{sys}\left(e^{-ik(l)\mathcal{A}}\ket{R}_{sys} - \ket{P}_{sys}\right)\ket{l}_k\right).
\end{equation}
The desired information is contained in a single qubit Pauli-$Z$ measurement,
\begin{equation}
    \bra{\Phi} Z_{0} \ket{\Phi} = \frac{1}{\alpha_g}\text{Re}\left(\bra{P}e^{t\HB }\ket{R}\right),
\end{equation}
where $Z_{0}$ is the Pauli-$Z$ operation on the top ancilla qubit in Fig. \ref{fig:overlap-circuit}.

We {can} directly sample from the top ancilla qubit in Fig. \ref{fig:overlap-circuit}, which attains the standard quantum limit and requires $O(1/\epsilon^2)$ circuit repetitions, {since} $\alpha_g$ and $ \text{Var}(\langle Z\rangle) \in O(1)$. 
{Sampling at this rate is sufficient if the desired error tolerance is reasonably large.}
On the other hand, if $\epsilon \ll 1$, it may be beneficial to use QAE or some variant thereof to estimate this quantity with $O(1/\epsilon)$. {QAE essentially uses the quantum circuit in Fig. \ref{fig:overlap-circuit} in a quantum phase estimation circuit in a Grover iterate, which increases the circuit depth by a factor of $O(1/\epsilon)$ and requires an additional $O(\log(1/\epsilon)$ ancillae, but reduces the sample complexity to $O(1)$.}
We {note} that since amplitude estimation produces an $\epsilon$ approximation of the \textit{square root} of the probability, this {does} increase the query complexity by a factor of $2$ due to the error propagation from squaring a number in $(0,1)$.  {However, the detailed value of the desired error tolerance will determine whether the additional circuit overhead induced by adding QAE to achieve $O(1/\epsilon)$ scaling is worthwhile for a specific instance.} 
Theorem~\ref{thm:reactivefluxcomplexity} below expresses the overall complexity for calculation of the reactive flux, including quantum amplitude estimation to arrive at the $O(1/\epsilon)$ Heisenberg scaling.

\begin{thm}[Complexity to compute reactive flux]\label{thm:reactivefluxcomplexity}
    Let $\epsilon, t>0$. Given access to quantum states $\ket{R}$ and $\ket{P}$ as defined in Eq. \eqref{eq:reactive-state}, and a block encoding of the matrix square root $\mathcal{A}$, with subnormalization factor $\alpha_A \in O\left(\sqrt{\eta d}\left(\eta L \alpha_V \sqrt{\beta} + \sqrt{\beta^{-1}}N\right)\right)$, with $L, \alpha_V$ as in Theorem \ref{thm:BE-Aj-ops} and with a potential given by a polynomial of degree $2k$. There exists a quantum algorithm that outputs an approximation $\widetilde{\nu}_{RP}$ to the quantity
    \begin{equation}
        \nu_{RP} = \bra{P}e^{t\HB }\ket{R},
    \end{equation}
    satisfying
    \begin{equation}
        \left|\widetilde{\nu}_{RP} - \nu_{RP}\right| \leq \epsilon
    \end{equation}
    that uses $O\left(\frac{\sqrt{t}\alpha_A}{\epsilon}\left(\eta d n + 2k d n^2 + (nd)^2\right)\right)$ Toffoli or simpler gates and queries the state preparation oracles an additional $O(1/\epsilon)$ times. 
\end{thm}
\begin{proof}
    Simply apply amplitude estimation to the quantum circuit in Fig. \ref{fig:overlap-circuit}, and implement the block encoding of the $\mathcal{A}$ using the costs reported in Theorem \ref{thm:BE-Aj-ops} and the time evolutions using Corollary \ref{cor:gauss-lchs-FPE}.
\end{proof}

\begin{cor}[Explicit $\epsilon$-dependence of reactive-flux estimation]
\label{cor:eps-scaling-dichotomy}
Under the hypotheses of Theorem~\ref{thm:reactivefluxcomplexity}, fix
$t>0$ and $\epsilon\in(0,1)$, and assume $L,\alpha_V,d\in O(1)$.
\begin{enumerate}
    \item[(i)] {Polynomial-in-$r^2$ regime.}
    If the pair potential $V_0$ is polynomial in $r^2$ of degree $2k$
    (Definition~\ref{def:r2-polynomial}), then
    Lemma~\ref{lem:planewave-convergence} gives
    $N\in\mathrm{polylog}(\eta,\beta,t,1/\epsilon)$, and the gate
    complexity of Theorem~\ref{thm:reactivefluxcomplexity} simplifies to
    \begin{equation}
        \widetilde{O}\!\left(
            \frac{\eta^{5/2}\alpha_V\sqrt{t\beta}
                  + \eta^{3/2}\sqrt{t/\beta}}{\epsilon}
        \right),
    \end{equation}
    with polylogarithmic dependence on $1/\epsilon$.

    \item[(ii)] {Odd-power regime.}
    If $V_0$ contains a non-trivial term $b\,r^m$ with $m\ge 3$ odd, then
    Proposition~\ref{prop:algebraic-rate} forces
    $N=\Omega(\epsilon^{-1/(m+d-2)})$, and the gate complexity becomes
    \begin{equation}
        \widetilde{O}\!\left(
            \frac{\eta^{5/2}\alpha_V\sqrt{t\beta}
                  + \eta^{3/2}\sqrt{t/\beta}\,\epsilon^{-1/(m+d-2)}}{\epsilon}
        \right),
    \end{equation}
    which carries a polynomial overhead in $1/\epsilon$ on the kinetic
    term, while the (dominant) potential term retains
    $\widetilde{O}(\epsilon^{-1})$ scaling. In the worst case
    ($m=3$, $d=1$) the kinetic contribution scales as
    $\widetilde{O}(\epsilon^{-3/2})$.
\end{enumerate}
The two regimes are illustrated numerically in
Table~\ref{tab:dichotomy}, and Proposition~\ref{prop:algebraic-rate}
establishes that the algebraic exponent in (ii) is sharp.
\end{cor}

\begin{rem}[Singular pair potentials via smoothed surrogates]
\label{rem:lj-surrogate}
Pair potentials of practical interest (Lennard-Jones, Morse, screened
Coulomb) are neither polynomial in $r^2$ nor globally confining, and so
fall outside the hypotheses of Corollary~\ref{cor:eps-scaling-dichotomy}
directly. The {smoothed surrogate} construction of
Appendix~\ref{subsec:smoothed-surrogate} replaces such $V_0$ by a
$C^2$ piecewise interpolant that is polynomial in $r^2$ near coincidence
and joins a polynomial confining wall near the box boundary;
Theorem~\ref{thm:singular-convergence} shows the surrogate's
reactive flux converges to the physical one with $N$ polylogarithmic
in $1/\epsilon$, restoring the $\widetilde{O}(\epsilon^{-1})$ scaling
of regime~(i) modulo seam contributions that decay exponentially in the
matched design parameters. The Lennard-Jones case is worked out
explicitly in Appendix~\ref{subsec:LJ-numerics}, with
numerical validation in Fig.~\ref{fig:LJ-validation}.
\end{rem}

Up to this point, our discussion has focused on achieving additive error in our estimates of the reactive flux. While relative error is necessary for precise quantification of reaction rates, many practical questions can be answered with additive error. 
For instance, determining whether a reaction rate falls below a safety threshold, comparing the relative efficiency of different catalysts,  or determining if a reaction occurs with non-zero probability within some fixed time horizon, only require that our estimate has sufficient precision to distinguish the quantity from zero or from another value. 
For such additive-error questions, our quantum algorithm {will achieve} a sample complexity of $O\left(\frac{1}{\epsilon}\right)$ via amplitude estimation, independent of the actual magnitude of the reactive flux. For example, if we wish to determine {whether} the reaction occurs with probability $< 10^{-6}$, then we require  $\Omega(10^3)$ queries to our quantum algorithm, while $\Omega(10^{6})$ queries to the classical algorithm are required to estimate the same quantity\footnote{More precisely, to determine if the probability that a reaction occurs is less than $P$, a quantum algorithm requires $\Omega\left(\frac{1}{\sqrt{P}}\right)$ queries, whereas $\Omega\left(\frac{1}{P}\right)$ queries are required classically.} .

Finally, to obtain the time-dependent reaction rate we need only to combine our estimates of the reactive flux with the ratio of probabilities of being in the reactive and product states. Since 
\begin{align*}
    k_{RP}(t) = \frac{1}{t}\sqrt{\frac{\overline{p}_P}{\overline{p}_R}}\nu_{RP}(t),
\end{align*}
with $\overline{p}_R = \frac{\mathcal{Z}_R}{\mathcal{Z}}$, we have $\sqrt{\frac{\overline{p}_R}{\overline{p}_P}} = \sqrt{\frac{\mathcal{Z}_R}{\mathcal{Z}_P}}$ suffices to obtain estimates to the partition functions in the reactant and product regions. When $R$ and $P$ correspond to locally convex regions, these quantities can be efficiently obtained with a small additive overhead.

\section{Numerical Examples and Resource Estimations}
\label{sec:Num-Examples}

{
We instantiate the quantum algorithm on the one-dimensional double-well potential
\begin{equation}
    V(x) = x^4 - x^2
    \label{eq:double-well}
\end{equation}
on the domain $[-L,L]$ with $L=2$, discretized using $N=2^n$ plane-wave modes per
degree of freedom. This potential has local minima at $x=\pm 1/\sqrt{2}$, a
barrier of height $\Delta V = 1/4$ at $x=0$, and serves as a standard benchmark
for rare-event sampling. The domain $L=2$ is chosen so that the Boltzmann weight
$e^{-\beta V(x)}$ is below $10^{-5}$ for $|x|\geq 2$ at $\beta\geq 1$, making the
periodic boundary condition negligible. We take a single particle ($\eta=d=1$)
and identify the reactant and product regions as $R = \{x < 0\}$ and
$P = \{x > 0\}$. The potential degree is $2k=4$ (so $k=2$), and the
subnormalization input is
\begin{equation}
    \alpha_V = \max_{|x|\leq L}\!\left|\frac{V'(x)}{x}\right|
    = \max_{|x|\leq L}|4x^2 - 2| = 4L^2 - 2 = 14,
\end{equation}
attained at $|x|=L$.

\subsection{Convergence with $N$.}
\label{subsec:conv-N}
We first determine the plane-wave mode count $N$ required to resolve the reactive flux numerically; this discretization decision sets the qubit count per degree of freedom and enters the gate-count analysis below. Two figures jointly support the
choice. Figure~\ref{fig:nu_vs_t} establishes 
that the discretization reproduces the correct dynamics across the full time
window from transient to plateau, while Figure~\ref{fig:algebraic-vs-exp-convergence}
quantifies the {convergence rate} in $N$ at fixed $t$ and supplies the quantitative bound used to size $N$.

\begin{figure}[ht!]
    \centering
    \includegraphics[width=\textwidth]{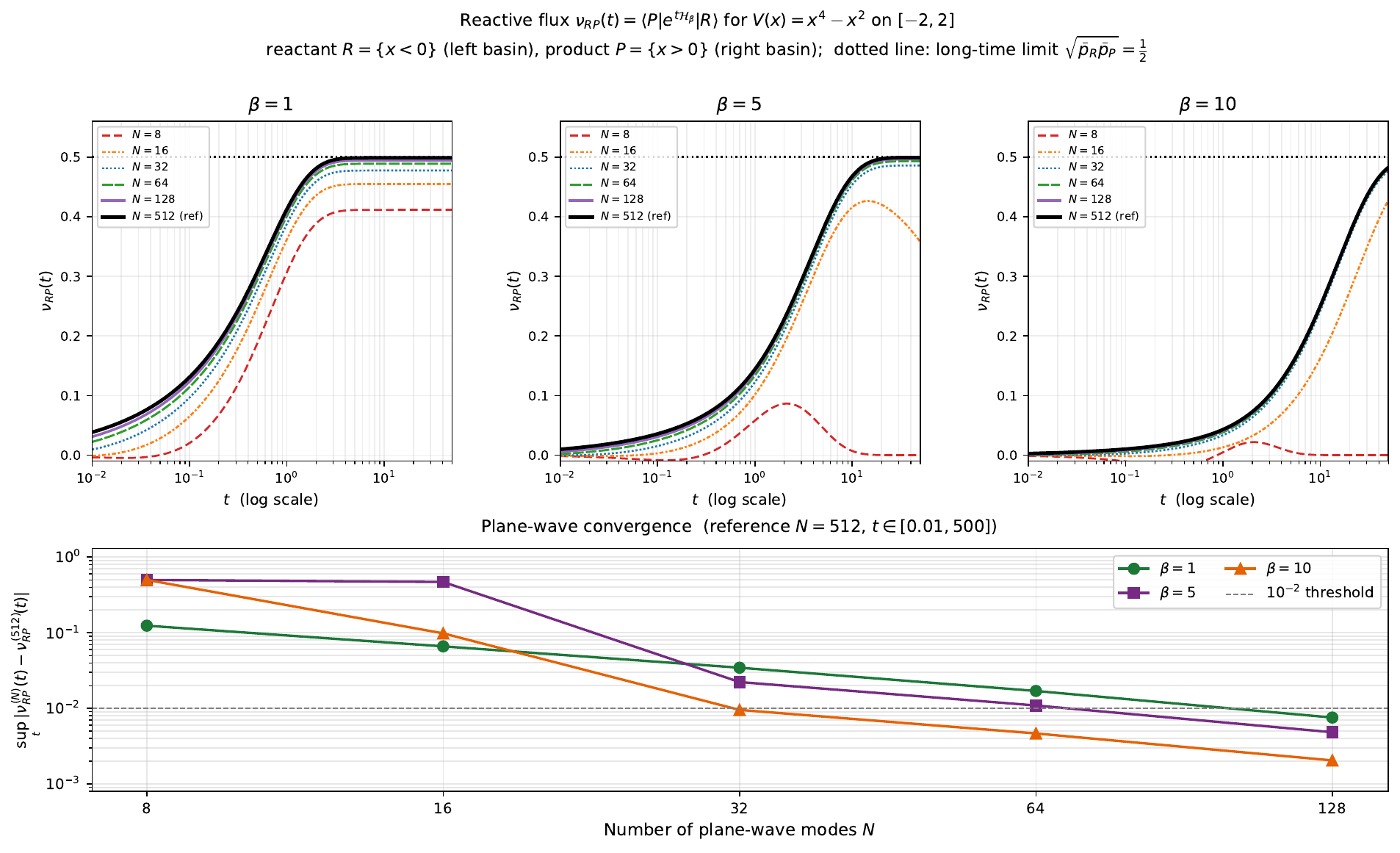}
    \caption{Reactive flux $\nu_{RP}(t)=\langle P|e^{t\mathcal{H}_\beta}|R\rangle$
    for $V(x)=x^4-x^2$ on $[-2,2]$, with reactant $R=\{x<0\}$ and product
    $P=\{x>0\}$, computed by exact diagonalization of the pseudospectral FKE
    generator. \textit{Top row:} flux curves for $\beta\in\{1,5,10\}$ and
    $N\in\{8,16,32,64,128\}$ plane-wave modes, with $N=512$ shown as the thick
    black reference. The black dotted line marks the long-time limit
    $\sqrt{\bar{p}_R\bar{p}_P}=1/2$. \textit{Bottom panel:} absolute sup-norm
    convergence error
    $\sup_{t\in[0.01,500]}|\nu_{RP}^{(N)}(t) - \nu_{RP}^{(512)}(t)|$ for each
    $(\beta,N)$ pair, with $\nu_{RP}(t)\in[0,1/2]$ across the window. At $N=32$
    the error is $3.4\%$, $2.2\%$, and $0.95\%$ for $\beta=1,5,10$ respectively;
    $N=16$ fails to reproduce the correct long-time limit for $\beta\geq 5$. The
    $O(N^{-2})$ convergence rate is set by the indicator-function jump of $R,P$
    at the dividing surface $x=0$, requiring the $N=512$ reference rather than
    saturating at machine precision.}
    \label{fig:nu_vs_t}
\end{figure}

Figure~\ref{fig:nu_vs_t} reports the reactive flux $\nu_{RP}(t) =
\langle P|e^{t\mathcal{H}_\beta}|R\rangle$ computed by exact diagonalization of the
$N\times N$ pseudospectral FKE generator on $[-2,2]$ for $N \in \{8, 16, 32, 64,
128\}$ and $\beta \in \{1, 5, 10\}$, with indicator-function reactant/product
states $R = \{x<0\}$, $P = \{x>0\}$ and $N=512$ as reference. The error reported
in the lower panel is the {absolute} sup-norm
$\sup_{t\in[0.01,500]} |\nu_{RP}^{(N)}(t) - \nu_{RP}^{(512)}(t)|$, with
$\nu_{RP}(t)$ valued in $[0, 1/2]$ across this window; the window covers multiple
Kramers crossing times for all $\beta$ (the Kramers time $\tau \sim
e^{\beta\Delta V}/\omega \approx 83$ for $\beta = 10$) and thereby tests both the
transient dynamics and the long-time limit $\nu_{RP}(\infty) = \sqrt{\bar{p}_R
\bar{p}_P} = 1/2$. The convergence rate is limited to $O(N^{-2})$ by the jump
discontinuity of the indicator states $\mathbf{1}_R$, $\mathbf{1}_P$ at the
dividing surface $x=0$, which is why a finer reference than $N=128$ is needed to
display the convergence trend over the full tested range. The results show that
$N=32$ (five qubits per degree of freedom) is the smallest grid that reproduces
the qualitative dynamics across all tested $\beta$, with absolute sup-norm
errors $3.4\%$, $2.2\%$, and $0.95\%$ for $\beta=1, 5, 10$ respectively and the
long-time limit within $1.5\%$ of the exact value $1/2$ in all cases; meeting
the $10^{-2}$ sup-norm threshold uniformly across $\beta\in\{1,5,10\}$ requires
$N=128$ (at $\beta=1$, $N=64$ still sits at $1.7\%$). At $N=16$ the
pseudospectral ground state is underresolved for $\beta\geq 5$ and the long-time
limit fails to plateau at $1/2$.

The convergence pattern is non-monotone in $\beta$ at fixed $N$: errors first
decrease with increasing $\beta$ because the slow transient dynamics (the Kramers
time $\tau\sim e^{\beta\Delta V}$ grows faster than the measured time window)
make both the test and reference curves nearly identical within the window;
errors then worsen for $\beta \gg \beta^\ast$, where $\beta^\ast \approx
(N/4L)^2$ is the crossover below which the ground-state width $\sigma \sim
(2\sqrt{\beta})^{-1}$ is no longer resolved by the grid spacing $h = 2L/N$. For
$L=2$, $\beta^\ast = N^2/64$: for $N=32$, $\beta^\ast = 16 > 10$, confirming
reliable convergence across all tested $\beta$; for $N=16$, $\beta^\ast = 4$,
consistent with the observed failure at $\beta=5$.

Figure~\ref{fig:algebraic-vs-exp-convergence} quantifies the fixed-$t$
convergence rate of $\nu_{RP}(t)$ on the analytic double well, using
basin-localized Gaussian wavepacket reactant/product states at the more
demanding box half-width $L=4$ and inverse temperature $\beta=5$:
$|\nu_{RP}^{(32)}(t{=}1) - \nu_{RP}^{(1536)}(t{=}1)| \approx \approx 1.3\times10^{-4}$,
dropping super-algebraically thereafter ($\sim 5.5\times10^{-6}$ at $N=48$,
machine precision by $N\approx 96$). The state choice matters: the
$O(N^{-2})$ rate observed in Fig.~\ref{fig:nu_vs_t} is set by the jump
discontinuity of the indicator functions $\mathbf{1}_R$, $\mathbf{1}_P$ at the
dividing surface, whereas smooth (Schwartz-class) state preparations recover the
super-algebraic rate intrinsic to the propagator on the analytic potential. The
gate-count analysis below uses the smooth-state convergence rate as the bound on
$\epsilon_{\mathrm{disc}}$: the algorithm's state-preparation oracles
$O_R, O_P$ (Sec.~\ref{subsec:state-prep}) prepare basin-localized indicators of the form of Theorem~\ref{lem:augmented-V-convexity}, for which the
fixed-$t$ matrix-element error inherits the rate of
Fig.~\ref{fig:algebraic-vs-exp-convergence}. Figure~\ref{fig:nu_vs_t} plays the
complementary role of certifying that the pseudospectral discretization does not
drift or fail to plateau over the full Kramers window, using indicator states as
the worst case; the sup-over-time error is the right metric for that diagnostic
but is dominated by the indicator's discontinuity rather than the algorithmic
state preparation.

\begin{figure}[ht!]
    \centering
    \includegraphics[width=0.9\linewidth]{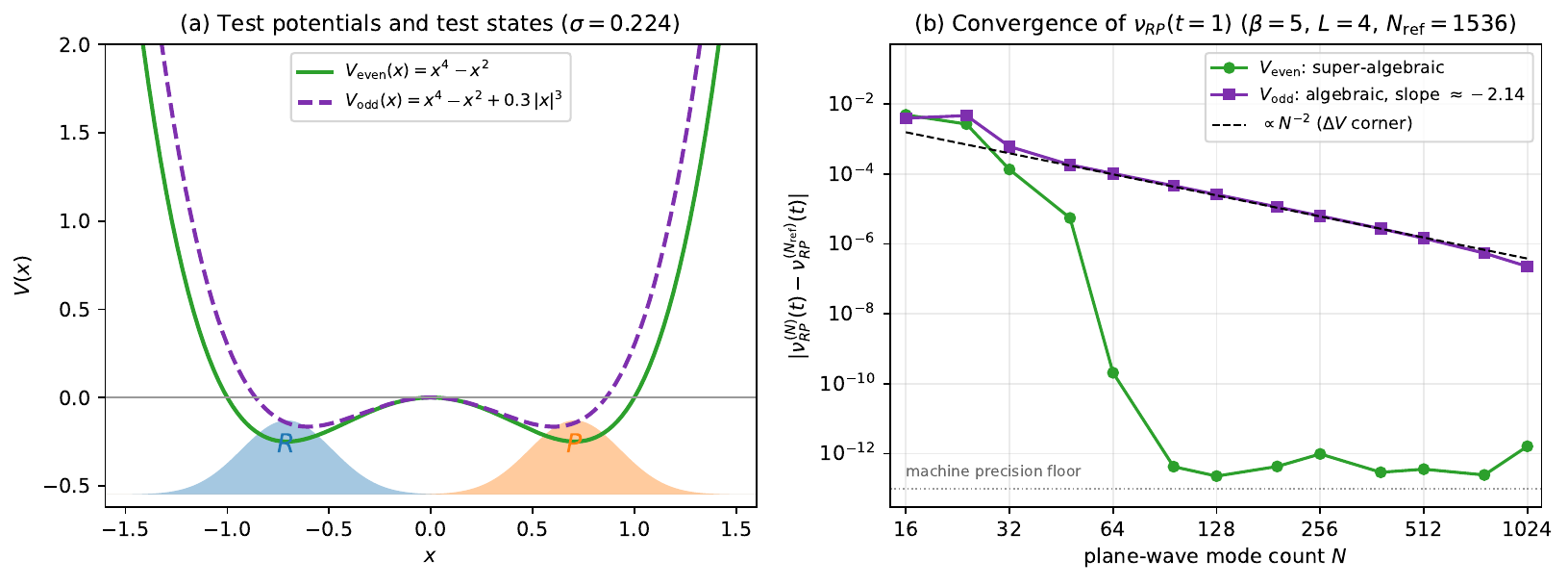}
    \caption{Plane-wave convergence dichotomy for the reactive-flux matrix
    element $\nu_{RP}(t)$, comparing a polynomial double well
    $V_{\rm even}=x^4-x^2$ with the same potential perturbed by a non-even power,
    $V_{\rm odd}=x^4-x^2+0.3\,|x|^3$ (only $C^2$ at the saddle $x=0$).
    \textit{(a)} The two potentials and the basin-localized Gaussian
    reactant/product states $R,P$. \textit{(b)} Error of $\nu_{RP}(t{=}1)$ versus
    the number of plane-wave modes $N$ ($\beta=5$, $L=4$, reference $N=1536$). The
    analytic potential converges super-algebraically: $|\nu_{RP}^{(32)} -
    \nu_{RP}^{(1536)}| \approx 1.3\times10^{-4}$ at $\beta=5,L=4$, reaching
    machine precision by $N\approx96$; the non-even power converges only
    algebraically at the rate $N^{-2}$, set by the corner that $\Delta V$ acquires
    at the saddle from the twice-differentiated $|x|^3$ term --- the
    operator-level signature of Proposition~\ref{prop:algebraic-rate}.}
    \label{fig:algebraic-vs-exp-convergence}
\end{figure}

The convergence figures also illustrate the separation of thermodynamic and
kinetic content in $\nu_{RP}(t)$
(Eqs.~\eqref{eq:nu-eigendecomp}--\eqref{eq:nu-inf-limit}): the dotted line marks
the thermodynamic baseline $\sqrt{\bar{p}_R\bar{p}_P}$, which depends only on
equilibrium populations and is classically computable; the kinetically nontrivial
content is the transient $\nu_{RP}(t) - \sqrt{\bar{p}_R\bar{p}_P}$, whose time
evolution is governed by the spectral gap $|\lambda_1|$ of $\mathcal{H}_\beta$. At
short times $t \ll \tau_{\mathrm{Kramers}}$, the transient is close to
$-\sqrt{\bar{p}_R\bar{p}_P}$, i.e.  $\nu_{RP}(t)$ is exponentially small in
$\beta\Delta V$ since few trajectories have crossed the barrier. The flux only
becomes $O(1)$ at $t \sim \tau_{\mathrm{Kramers}} \sim e^{\beta\Delta V}$.
Estimating $\nu_{RP}(t)$ to additive accuracy $\epsilon$ at any fixed $t$,
including this exponentially small regime, is precisely what the quantum
algorithm achieves via the $\Omega(1)$ success probability of the overlap circuit
(Sec.~\ref{subsec:overlap-circuit}), giving $O(1/\epsilon)$ complexity
independent of the magnitude of $\nu_{RP}(t)$. This is the source of the $\epsilon^{-1}$
versus $\epsilon^{-4}$ improvement over classical trajectory simulation.

\subsection{Validating the Gaussian-LCHS workflow}
\label{subsec:G-LCHS-num-ver}

To verify the full algorithmic chain at finite $N$, we worked through the
approximation layers of the Gaussian-LCHS construction in turn. The dilated
square root $\mathcal{A}_H$ is built in the plane-wave basis as described in
Appendix~\ref{app-sec:BE-A} (with the Nyquist mode zeroed so that $A=-iB$ for
real $B$ and $A^\dagger A = B^\top B$ is real symmetric); the identity
$(-\mathcal{A}_H^2)_{\rm top} = -A^\dagger A$ is then satisfied to machine
precision. We take this finite-$N$ operator as the reference, obtaining
$\nu_{RP}^{\rm exact}(t)$ by direct diagonalization of $-A^\dagger A$. The
linear combination $e^{-t\mathcal{A}_H^2}\approx\sum_j c_j(t)\,e^{-i k_j
\mathcal{A}_H}$ is realized by Gauss--Legendre quadrature on $[-L_G,L_G]$
with $L_G(t,\epsilon)$ set by Lemma~\ref{lem:trunc_error}, and the resulting
$\nu_{RP}^{\rm GLCHS}(t;M_q)$ is compared against the reference to locate
the smallest $M_q^*(t,\epsilon)$ that reaches the target precision.

For the 1D double well at $\beta=10$ with $N=64$ plane-wave modes and
$\epsilon=10^{-3}$, Fig.~\ref{fig:GLCHS-validation}(a) shows the GLCHS
estimates evaluated at exactly $M_q^*$ unitary calls landing on
$\nu_{RP}^{\rm exact}(t)$ across the full window. Fig.~\ref{fig:GLCHS-validation}(b)
plots $M_q^*(t)$ on log-log axes for $t\geq 5$: an OLS fit returns a slope of
$0.476\pm 0.011$ at $R^2=0.998$, statistically consistent with the
$\sqrt{t}$ growth of Eq.~\eqref{app-eq:M-bound} (the residual departure from
exactly $1/2$ is set by the sub-leading $\log(10/(\epsilon\sqrt t))$ term,
which decreases mildly with $t$). The empirical $M_q^*$ sits a roughly
constant factor of $\sim 10$ below the worst-case bound across the entire
range, indicating that Eq.~\eqref{app-eq:M-bound} is tight in scaling and
loose only in prefactor on this problem.

\begin{figure}[t]
    \centering
    \includegraphics[width=\linewidth]{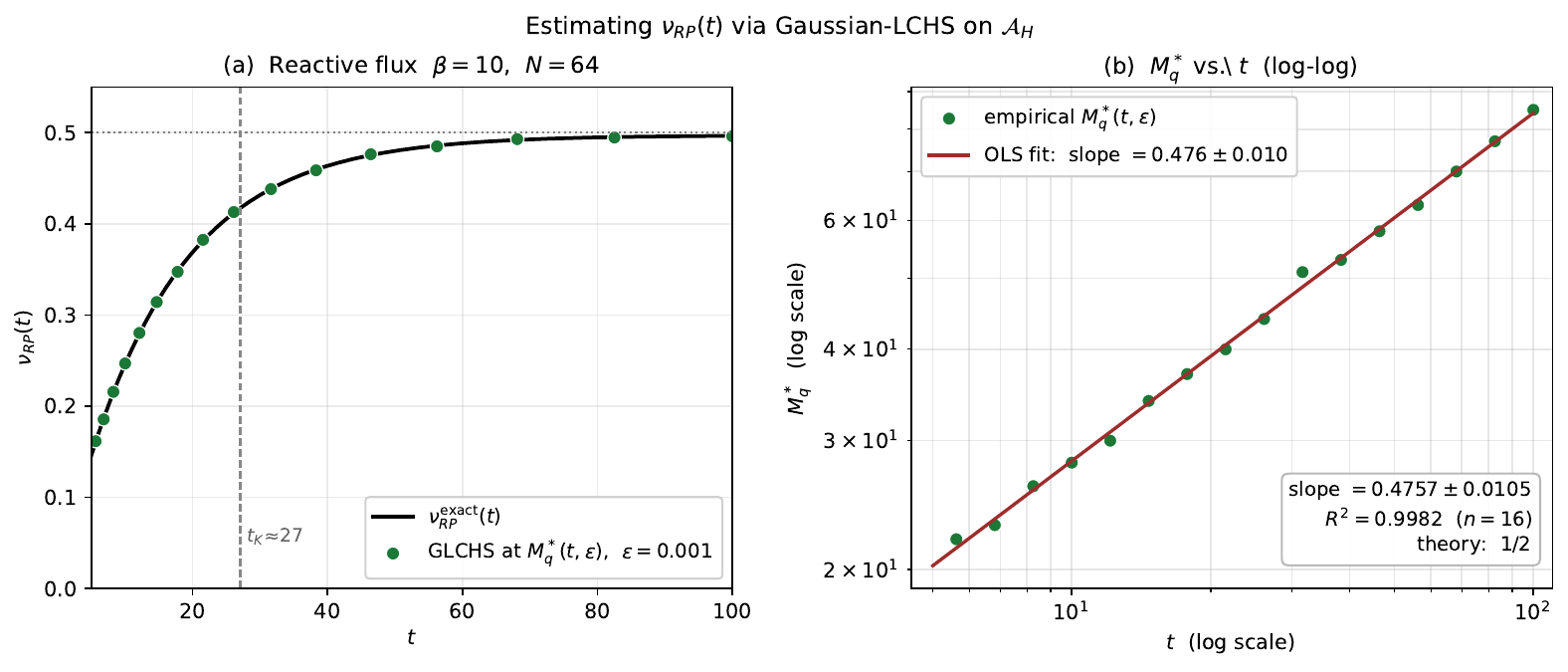}
    \caption{Numerical validation of the Gaussian-LCHS estimation pipeline for the reactive flux at $\beta=10$, $N=64$, target precision $\epsilon=10^{-3}$. \textit{(a)} $\nu_{RP}(t)$ from direct diagonalization of $\mathcal{H}_\beta=-A^\dagger A$ (solid black) versus Gaussian-LCHS estimates through the Hermitian dilation $\mathcal{A}_H$ (green markers), evaluated at exactly the $M_q^*(t,\epsilon)$ unitary calls needed to reach precision $\epsilon$. The markers lie on the exact curve across the full window, including past the Kramers crossing time $t_K\approx 27$. \textit{(b)} Log-log plot of $M_q^*(t,\epsilon)$ for $t\geq 5$ with an OLS fit. The fitted slope is $0.476\pm 0.0105$ at $R^2 = 0.998$ ($n=16$ points), consistent with the $\sqrt{t}$ growth predicted by Eq.~\eqref{app-eq:M-bound}.}
    \label{fig:GLCHS-validation}
\end{figure}

\subsection{Resource Estimates}
\label{subsec:resource-estimates}
With the discretization fixed by the convergence study, we obtain a concrete
Toffoli count for the algorithm by composing the proven per-subroutine
complexities outward from the innermost primitive, carrying the constant factors
established in Appendices~\ref{app-sec:BE-A} and~\ref{app-subsec:sim-error-bounds}.
Reading the construction as nested layers: the primitive arithmetic operations
(controlled subtraction, sums of squares, inequality testing, and the
amplitude-amplified $\mathfrak{unif}^{-1}$ uncomputation) compose into the block
encoding of the dilated square root $\mathcal{A}$ (Theorem~\ref{thm:BE-Aj-ops}),
at a per-query Toffoli cost $C_{\mathrm{BE}}$; multiplexed QSP
(Lemma~\ref{lem:mpx-qsp}) issues $D_{\max}$ queries to this block encoding, plus
$D_{\max}$ applications of the controlled-phase oracle $\mathcal{S}_i$ that
selects the QSP phase factor $\Phi_i^j$ from the $\lceil\log M_q\rceil$-qubit
quadrature register --- each $\mathcal{S}_i$ costing $\log M_q + \log(1/\epsilon)$
Toffolis (routing on $|j\rangle$ plus the finite-precision Z-rotation realized
via addition into a phase gradient register \cite{namApproximateQuantumFourier2020a}),  and a one-time
$O(M_q\log M_q)$ for the $\textsc{prep}_g$ coefficient state. The non-unitary
overlap circuit (Fig.~\ref{fig:overlap-circuit}) together with amplitude
estimation multiplies the per-circuit cost by $O(1/\epsilon)$:
\begin{equation}
  T_{\mathrm{total}}
  \;\le\;
  \frac{c_{\mathrm{AE}}}{\epsilon_{\mathrm{AE}}}
  \Bigl[\, D_{\max}\, C_{\mathrm{BE}}
       \;+\; D_{\max}\!\left(\log M_q + \log(1/\epsilon_{\mathrm{be}})\right)
       \;+\; O(M_q\log M_q) \Bigr].
  \label{eq:Ttotal}
\end{equation}
At the parameters of interest $D_{\max}$ and $M_q$ are proportional, so the $D_{\max}\log
M_q$ controlled-phase term is equivalent up to a constant to the $M_q\log M_q$
form stated in Theorem~\ref{thm:gauss-lchs-complexity} and remains subdominant
to the leading $D_{\max} C_{\mathrm{BE}}$ block-encoding cost throughout.

The composition is governed by an explicit error budget. The target additive
error $\epsilon$ on $\nu_{RP}$ is partitioned among four independent channels,
$\epsilon = \epsilon_{\mathrm{disc}} + \epsilon_{\mathrm{lchs}} +
\epsilon_{\mathrm{be}} + \epsilon_{\mathrm{AE}}$ (we take equal quarters): the
plane-wave discretization error $\epsilon_{\mathrm{disc}}$ fixes the mode count
$N$ via the convergence figures (Figs.~\ref{fig:nu_vs_t}
and~\ref{fig:algebraic-vs-exp-convergence}); the LCHS truncation, quadrature, and
simulation error $\epsilon_{\mathrm{lchs}}$ fixes $M_q$ and the degree $D_{\max}$;
the amplitude-estimation error $\epsilon_{\mathrm{AE}}$ fixes the repetition
count. Crucially, the block-encoding {internal} precision is not
$\epsilon_{\mathrm{be}}$ but $\epsilon_{\mathrm{be}}/D_{\max}$: because
$U_{\mathcal{A}}$ is queried $D_{\max}$ times and its errors add, the
simulation-error analysis underlying Theorem~\ref{thm:gauss-lchs-complexity}
(Appendix~\ref{app-subsec:sim-error-bounds}) requires each query to be
correspondingly tighter. The same analysis shows that the per-unitary budget
division across the $M_q$ quadrature terms does {not} multiply the leading
order by $M_q$; it collapses into the additive logarithm of the degree
(Eq.~\eqref{eq:poly-appx-unitary-bound}), bounded sharply by Lemma~21 of
Ref.~\cite{pocrnicConstantFactorImprovementsQuantum2025} as
\begin{equation}
  D_{\max} \;=\; \left\lceil\, \tfrac{e}{2}\,\alpha_A\, L_G
  \;+\; \ln(2\cdot 1.47762/\epsilon_{\mathrm{lchs}})\,\right\rceil,
  \label{eq:Dmax-consts}
\end{equation}
where $L_G$ is the LCHS truncation wavenumber from the sharpened form of
Lemma~\ref{lem:trunc_error}, $L_G = 2\sqrt{t\log(1/(\epsilon_{\mathrm{lchs}}\sqrt{\pi}))}$.
We assume a phase-gradient ancilla throughout, so that all single-qubit rotations
in the QFTs and QSP phase oracles reduce to Toffoli (addition) cost; the
reactant/product state-preparation oracles $O_R, O_P$ enter the per-circuit cost
{additively} rather than multiplicatively, contributing an overall factor of
at most two even in the (unrealized) worst case that their cost rivals the entire
remainder of the algorithm, and are omitted from the tally below.

We emphasize that Eq.~\eqref{eq:Ttotal} is a faithful {upper bound} on a
synthesized circuit. The analytic QSP and Jacobi--Anger degrees are worst-case
over the retained spectral band; at the precisions of interest the realized
degrees are substantially smaller, and standard gate cancellation and
phase-factor compilation reduce the constants further. Circuit synthesis is a
separate undertaking with its own optimizations; the nested estimate reported
here suffices to certify tractability and, if anything, overstates the true
optimized cost.

For the double-well parameters above ($\eta=d=1$, $k=2$, $L=2$, $\alpha_V=14$),
Theorem~\ref{thm:BE-Aj-ops} gives the subnormalization factor
\begin{equation}
    \alpha_A \;=\; 28\sqrt{\beta} + \frac{N}{2\sqrt{\beta}},
    \label{eq:alpha-A-double-well}
\end{equation}
and a per-query block-encoding cost $C_{\mathrm{BE}}$ dominated by the $(2k{-}2)$
Jacobi--Anger queries to the distance oracle and the position--momentum QFTs.
Table~\ref{tab:gate-counts} reports the composed quantities across a precision
sweep at $t=1$, $\beta=1$. The mode count $N$ in each row is the smallest power
of two satisfying the discretization sub-budget $\epsilon_{\mathrm{disc}} =
\epsilon/4$ on the fixed-$t$ matrix-element error, taken from
Fig.~\ref{fig:algebraic-vs-exp-convergence}: at the conservative dichotomy
parameters $\beta=5$, $L=4$, $N=32$ yields $|\nu_{RP}^{(32)} -
\nu_{RP}^{(\mathrm{ref})}| \approx \approx 1.3\times10^{-4}$ (sufficient for
$\epsilon\geq 10^{-3}$), $N=48$ yields $\approx 5.5\times10^{-6}$ (sufficient for
$\epsilon \geq 10^{-4}$), and $N=64$ yields $\approx 2.1\times10^{-10}$ (sufficient for all tabulated $\epsilon$); we round up to the next power of two for the QFT. This is conservative for the Sec.~V parameters ($L=2$), since the smaller box admits faster plane-wave convergence at fixed $N$.

\begin{table}[h!]
\centering
\caption{Bottom-up Toffoli resource estimate for the double-well test problem
($V = x^4 - x^2$, $\eta=d=1$, $L=2$, $\alpha_V=14$) at $t=1$, $\beta=1$, obtained
by composing the proven subroutine complexities
(Theorems~\ref{thm:BE-Aj-ops},~\ref{thm:gauss-lchs-complexity},~\ref{thm:reactivefluxcomplexity},
Lemma~\ref{lem:mpx-qsp}) with the error budget split into equal quarters.
$C_{\mathrm{BE}}$ is the per-query block-encoding cost (built at internal
precision $\epsilon_{\mathrm{be}}/D_{\max}$); $D_{\max}$ the number of queries to
$U_{\mathcal{A}}$; $M_q$ the number of quadrature points; $T_{\mathrm{total}}$
the composed estimate of Eq.~\eqref{eq:Ttotal}. The counts are conservative upper
bounds on a synthesized circuit. At $\beta=10$ the totals are larger by a factor
of $\approx 2$, entering through $\alpha_A$.}
\label{tab:gate-counts}
\begin{tabular}{ccccccc}
\hline\hline
$\epsilon$ & $N$ & $\alpha_A$ & $C_{\mathrm{BE}}$ & $D_{\max}$ & $M_q$ & $T_{\mathrm{total}}$ \\
\hline
$10^{-2}$ & 32 & 44 & 483 & 286 & 63  & $1.1\times10^{8}$  \\
$10^{-3}$ & 32 & 44 & 533 & 342 & 75  & $1.5\times10^{9}$  \\
$10^{-4}$ & 64 & 60 & 735 & 529 & 115 & $3.2\times10^{10}$ \\
$10^{-5}$ & 64 & 60 & 793 & 587 & 128 & $3.8\times10^{11}$ \\
\hline\hline
\end{tabular}
\end{table}

Two features are of practical importance. First, $T_{\mathrm{total}}$ ranges over
$10^{8}$--$10^{11}$ Toffoli gates across three orders of magnitude in target
accuracy, a range large but not unlike published fault-tolerant resource
estimates for other quantum-chemical primitives. Second, the dependence on the
mode count is mild: $\alpha_A$ is dominated by the $28\sqrt{\beta}$ potential
term, which is independent of $N$, and $N$ itself grows only polylogarithmically
with $1/\epsilon$, so that the discretization remains a modest five or six qubits
per degree of freedom across the entire sweep.

\textbf{Scaling with particle number and inverse temperature.}
The same construction, applied to $\eta$ pairwise-interacting particles with a
degree-$2k$ polynomial pair potential, exhibits the polynomial growth that
underlies the asymptotic separation from the classical worst-case bound.
Figure~\ref{fig:resource-scaling} shows the composed estimate $T_{\mathrm{total}}$
as a function of particle number $\eta$ (panel a) and inverse accuracy
$1/\epsilon$ (panel b), at $\beta=1$ and $\beta=10$. The $\eta^{5/2}$ scaling
assembles from three factors: the subnormalization $\alpha_A$ carries
$\sqrt{\eta}$ from the dimension of the configuration-space gradient and a
further $\eta$ from the $1$-norm of the many-body pair-potential gradient, giving
$\alpha_A=\Theta(\eta^{3/2})$; and the per-query block-encoding cost
$C_{\mathrm{BE}}$ carries a final factor of $\eta$ from the controlled swaps that
route the selected particle pair into and out of the work register (the
$2(\eta dn + n_\eta)$ term of Theorem~\ref{thm:BE-Aj-ops}). The product
$D_{\max}C_{\mathrm{BE}}$ therefore tends to $\widetilde{O}(\eta^{5/2})$
(Corollary~\ref{cor:gauss-lchs-FPE}). The precision scaling is
$\widetilde{O}(1/\epsilon)$ (panel b), the small super-linear steps marking the
discrete increments of $N$ required to keep the discretization error below
$\epsilon_{\mathrm{disc}}$. The inverse-temperature dependence is mild and enters
only through $\alpha_A$, displacing the curves by a constant factor.

\begin{figure}[t]
    \centering
    \includegraphics[width=\linewidth]{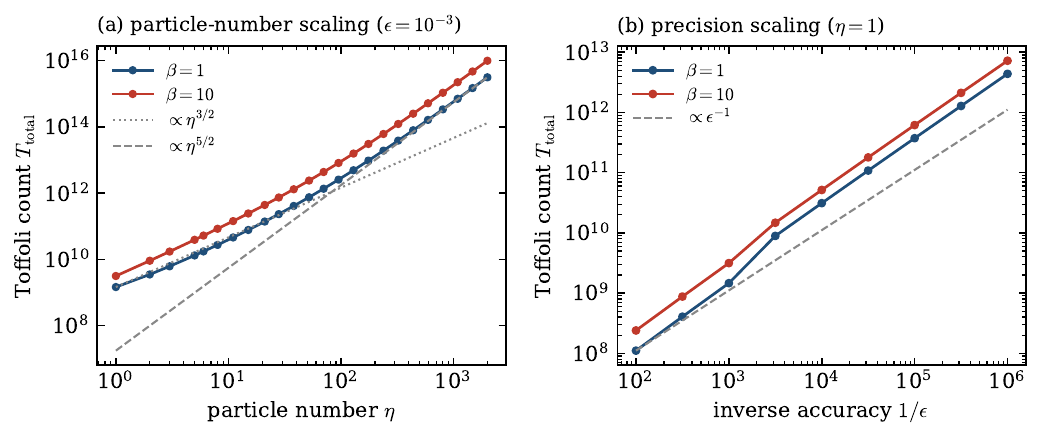}
    \caption{Bottom-up Toffoli resource estimate for the double-well test
    problem ($V=x^4-x^2$, $d=1$, $k=2$, $L=2$, $\alpha_V=14$, $t=1$), obtained by
    composing the proven subroutine complexities with the error budget of
    Eq.~\eqref{eq:Ttotal}--\eqref{eq:Dmax-consts}.
    \textit{(a)} Particle-number scaling at fixed accuracy $\epsilon=10^{-3}$:
    $T_{\mathrm{total}}$ versus $\eta$ for $\beta=1$ (blue) and $\beta=10$ (red),
    bracketed by $\eta^{3/2}$ (dotted, tangent at small $\eta$) and $\eta^{5/2}$
    (dashed, tangent at large $\eta$) reference slopes. The realized slope
    transitions from $\eta^{3/2}$ to $\eta^{5/2}$ around $\eta\sim 10^2$, reaching
    a local exponent of $2.45$ by $\eta\approx 2000$. \textit{(b)} Precision
    scaling at $\eta=1$: $T_{\mathrm{total}}$ versus $1/\epsilon$, tracking the
    $\epsilon^{-1}$ reference (dashed); the steps mark increments of the
    plane-wave mode count $N$ set by the discretization sub-budget. All counts
    are conservative upper bounds on a synthesized circuit.}
    \label{fig:resource-scaling}
\end{figure}

\textbf{Asymptotic versus pre-asymptotic regime.}
It is worth being explicit about where the $\eta^{5/2}$ scaling actually sets in. As Fig.~\ref{fig:resource-scaling}(a) shows, the realized exponent is close to$\eta^{3/2}$ for $\eta \lesssim 10^2$ and approaches $5/2$ only near $\eta \sim 10^3$, once the linear-in-$\eta$ controlled-swap network overtakes the fixed pair-arithmetic overhead in $C_{\mathrm{BE}}$ (local exponent $2.45$ by $\eta\approx 2000$). Pre-asymptotically the gate count is therefore milder than the headline $\eta^{5/2}$. Two caveats fix how this should be read. First, the many-body curve evaluates the resource {formula}; it is not a simulation of a $10^3$-particle system, the numerical demonstration in this section being the single-particle ($\eta=1$) double well. Second, a system of $\eta\sim 10^3$ interacting particles is in itself routine for classical molecular dynamics, which regularly propagates $10^3$--$10^6$ particles; the separation established in this work is not against such forward propagation but against the worst-case analytical guarantee for rare-event rate estimation (Eq.~\eqref{eq:classical-complexity}), whose exponential dependence on the Lipschitz constant $\gamma=\Omega(\eta)$ (Theorem~\ref{thm:many-body-lip-const}) renders the worst-case classical cost intractable well before the quantum cost enters its $\eta^{5/2}$ regime. The pre-asymptotic mildness of the quantum cost therefore strengthens rather than weakens the comparison: across the entire range
where the quantum estimate remains polynomial, the worst-case classical bound is already super-polynomially out of reach. We make no claim of advantage over specialized classical heuristics, whose realized cost is problem-dependent and generally far below the worst-case bound; the rigorous comparison is given in Sec.~\ref{sec:prior-work}.

Together, the convergence study and the bottom-up resource estimate show that
the algorithm is not only asymptotically efficient but concretely tractable: a
discretization of $N=32$ modes (five qubits per degree of freedom) at moderate
inverse temperature ($\beta \leq 10$) resolves the reactive flux at fixed $t$,
and the composition of Table~\ref{tab:gate-counts} places its estimation within
$\sim 10^{8}$ Toffoli gates at $\epsilon=10^{-2}$, rising polynomially in
$1/\epsilon$ and in particle number.
}%

\section{Discussion}
\label{sec:Conclusion}
We have constructed and analyzed a quantum algorithm that allows one to evaluate overlaps of the non-unitary propagator corresponding to the Kolmogorov equations using the method of Gaussian-LCHS, {allowing calculation of {classical} correlation functions and hence {of classical} reaction rates. The quantum algorithm can estimate these overlaps to additive accuracy $\epsilon$ using 
\begin{equation}
\widetilde{O}\left(\left(\eta d n + 2k d n^2 + (nd)^2\right)\sqrt{\eta d t} \left(\eta L\alpha_V\sqrt{\beta} +\sqrt{\beta^{-1}} \frac{N}{L}\right)\epsilon^{-1}\right)
\end{equation}
Toffoli or simpler quantum gates. In the regime where the dimension $d$, the polynomial degree $2k$, and the domain size $L$ are all $O(1)$, this simplifies to 
\begin{equation}
    \widetilde{O} \left(\frac{\eta^{5/2} \alpha_V\sqrt{t\beta} + \eta^{3/2}\sqrt{\frac{t}{\beta}} N}{\epsilon}\right),
\end{equation}
{with} $\alpha_V = \max_{\norm{r}\leq \sqrt{d}L} \left|\frac{V'(r)}{r}\right|$. 

As evident from the scaling, our quantum algorithm has some favorable properties. First, it does not suffer from exponentially small success probabilities as we increase the evolution time {$t$}. Second, it achieves sub-linear dependence on {both} the evolution time and {the} inverse temperature $\beta$.  Our approach is also {extendable} to more near-term approaches if the unitary evolution is implemented using splitting algorithms, such as Trotter, for the Hamiltonian simulation. {Underlying these properties is the effective-gap structure discussed in Sec.~\ref{sec:intro}. {Specifically,} for metastable systems, the FKE spectrum exhibits an effective gap $\Delta_{\mathrm{eff}} \gg \Delta_{\mathrm{gap}}$ separating the long-lived metastable manifold from higher excitations, so that the reactive flux $\nu_{RP}(t)$ reaches its plateau value on the time scale $t \sim 1/\Delta_{\mathrm{eff}}$, well before the mixing time $1/\Delta_{\mathrm{gap}} \sim e^{\beta\Delta V}$. {(See Fig.~\ref{fig:metastability} and the discussion thereof.)} It is this separation, unique to the forward-time (FKE) formulation, that allows our algorithm to answer interesting computational questions at more moderate time scales than the full equilibration time. This allows our algorithm to be potentially more efficient in practice than linear solve-based approaches, where the condition number is strictly controlled by the exponential barrier crossing time.}

{The current} work also shows a potentially \textit{exponential} separation {between} the complexity of computing overlaps of {general PDE} generators with negative logarithmic norm {and the complexity} of 
preparing solution states and estimating {the} observables of a state. Thus, it is of general interest to find additional applications in dissipative high-dimensional PDEs that can be reduced to estimating {such propagator} overlaps, {i.e., correlation functions between initial and {time-evolved} final states, rather than requiring the preparation of} solution states. {Here it is useful and motivational to note that} these {propagator} overlaps are directly related to experimentally relevant quantities such as transport coefficients and response functions, {so are generally more accessible experimentally than the full probability distribution}. 

We have provided a detailed analysis of the main costs that contribute to the overall complexity of the algorithm, including the block encoding of the square root matrix $\mathcal{A}$ and the efficient implementation of the (Gaussian) LCHS approach using the construction of {a new} multiplexed QSP.
Furthermore, we {have analyzed} the cost to estimate an observable of physical relevance in simulations of this type, {namely,} the reaction rate, 
providing an estimation {for the cost of performing} a realistic quantum computation of this type in an end-to-end fashion, {with the} block encoding, {time} evolution, and observable estimation costs {all} analyzed in detail. We have additionally extended these constructions beyond strictly polynomial pair potentials: Corollary~\ref{cor:BE-A-LJ} and the smoothed surrogate of Appendix~\ref{subsec:smoothed-surrogate} cover the physically ubiquitous class of singular pair potentials (Lennard-Jones, Morse, screened Coulomb) at no asymptotic cost, with the Lennard-Jones case worked out explicitly. {We also showed that this} work can also be readily extended to the computation of hitting times and estimating the probability of particular transition states. We have analyzed one technique to efficiently construct quantum states encoding the reactant and product states when {these} correspond to local minima of the potential energy surface. However, we have left the detailed implementation of this subroutine as a subject {for} future work.

Under the assumptions of Theorem \ref{thm:many-body-lip-const} and the sharpest known analytical bounds for time evolution of the stochastic dynamics of the corresponding overdamped Langevin in non-convex potentials, the classical complexity of estimating the reaction rate to additive error $\epsilon$ using trajectory-based methods {was shown in Ref.~\cite{chengSharpConvergenceRates2020} to scale} as
\[
\mathcal{C}=O\!\left(\frac{t \eta^{2}\log(\eta)\,e^{\gamma }}{\epsilon^{4}}\right).
\]
When the Lipschitz constant $\gamma$ scales linearly with $\eta$ (as proven for radially symmetric pair potentials {in this work (Theorem~\ref{thm:many-body-lip-const})} ), this {classical} bound becomes exponential in $\eta$. In contrast, our quantum algorithm scales polynomially in $\eta$, yielding an {exponential separation in particle number with respect to these worst-case classical bounds}, together with a quartic separation in accuracy $\epsilon$, (when $N\sim \polylog(1/\epsilon)$ as in Corollary \ref{cor:eps-scaling-dichotomy}) and a quadratic separation in simulation time over the standard LCHS method. {The classical} exponential scaling with $\eta$ is a result of {the classical method} requiring an exponentially small step size in the Lipschitz constant {$\gamma$}, which is used as a measure of the non-convexity of the potential. {In this context,} it is remarkable that for finite time simulations, our quantum algorithm has negligible dependence on the convexity of the potential. {However, we further note} that for sampling from the equilibrium distribution, both quantum and classical algorithms must {show} a dependence on the convexity of the potential, {which} can be viewed as resulting from the {long times} required to escape from metastable configurations. It is therefore unlikely that an exponential quantum advantage persists for equilibrium sampling {of {general} many-body states} using the methods we propose in this work.

We also emphasize that these {rigorous} classical bounds are derived under general assumptions and may be {overly} pessimistic for specific problems, {as noted earlier in this paper.} Indeed, specialized classical heuristics often perform better in practice than the worst-case {analytic treatment} predicts. Consequently, {direct empirical comparisons} on concrete problem instances will be necessary to fully quantify the achievable quantum advantage {in practice}. Nevertheless, the polynomial-vs-exponential scaling in $\eta$ {found here for the calculation of reaction rates of classically interacting many-body systems as described by the Kolmogorov equations} highlights a compelling route toward quantum advantage in regimes where classical methods are provably strained. {This is particularly important} because quadratic speedups alone are often insufficient for {achieving} practical advantage~\cite{babbushFocusQuadraticSpeedups2021}.

There are several avenues for practical improvements to {the current} quantum algorithm {in this work}. {First,} the $\eta^{5/2}$ scaling presents a target for further optimization. There are two sources of this scaling. One factor of $\eta^{3/2}$ comes from the subnormalization factor in the block encoding of $\mathcal{A}$. The other factor of $\eta$ comes from the number of elementary operations we needed to realize a block encoding of the many-body potential. This is suboptimal when compared {to} classical methods, which achieve $O(\eta^2)$ scaling for force evaluations, before applying optimizations such as the fast multipole method \cite{beatson_short_nodate}, which {further} reduce the scaling to $O(\eta \log(\eta))$. One reason for the difficulty in improving over the $\eta^{5/2}$ {quantum} scaling is that the quantum algorithm, which is fundamentally discretizing a many-body differential operator, is evaluating the forces applied between particles occupying all possible configurations simultaneously. Thus, it is unclear how methods based on the fast-multipole method could be applied in the current framework.

{Another avenue to potentially further improve these results is to develop methods that improve the convergence rate of the plane wave approximation to the $A_j$ operators. { When the potentials exhibit singular or non-analytic structure, this greatly degrades the planewave convergence rate, and one may require a large number of plane wave modes {in order} to discretize the system to the desired accuracy. In addition to the smoothed surrogate techniques we used for the Lennard-Jones potential in Sec. \ref{subsec:LJ-numerics},  
{work to improve the convergence rate of the planewave pseudospectral method to the $A_j$ operators} {could} include the development of efficient quantum circuitry to perform a discrete circular convolution of the potential with a mollifier, but additional work is needed to obtain efficient circuit constructions.}

Finally, we note that this problem setting is likely amenable to the methods of Ref. \cite{simonAmplifiedAmplitudeEstimation2024}, which achieve an asymptotic scaling of $O(1/\sqrt{\epsilon})$ to estimate operator expectation values. In principle, this would further increase the asymptotic separation in $\epsilon$ between the quantum and classical algorithm to a power of $8$.

We conclude by noting that the construction of the generalized matrix square root given by $\mathcal{A}$ mirrors the construction of a Hamiltonian of jump operators that has been proposed for quantum simulations of the Lindblad equation \cite{cleveEfficientQuantumAlgorithms2019a,dingSimulatingOpenQuantum2023}. Since the dynamics of the Fokker-Planck equation are purely dissipative, there is no coherent term in our problem. As was recently shown in Ref. \cite{fangMixingTimeOpen2024}, 
{a coherent dynamical contribution} can play an important role in improving the mixing time {of} open quantum systems. 
{Noting} the recent work {of} Ref. \cite{dingSingleancillaGroundState2023}, where jump operators were constructed to induce dissipation to the ground state, it is an {interesting avenue for} future work to investigate if the converse may be possible.  
That is, can one include a well-designed ``coherent term" in $\mathcal{A}$ so as to meaningfully accelerate the preparation of the equilibrium distribution?  {Here the recently designed ``shortcut to Zeno'' dynamics where a system under continuous monitoring in a time-dependent basis is subjected to a coherent term designed to cancel out non-adiabatic transitions \cite{lewalle2024optimal} offers one example. Another design route under continuous monitoring is offered by feedback, e.g., noise-canceling quantum feedback~\cite{karmakar2025noise}. It will be interesting to explore analogous methods for digital quantum algorithms.}

\vspace{-.5cm}

\begin{acknowledgments}
{T.D.K., A.M.A., K.K.M., and K.B.W. were supported by the U.S. Department of Energy, Office of Science, Office of Advanced Scientific Computing Research under Award Numbers DE-SC0023273 and DE-SC0025526. T.D.K. was also supported by a Siemens FutureMakers Fellowship. T.D.K. would also like to thank Aditya Singh for insightful discussions in the early portions of this work.} 
\end{acknowledgments}

\newpage
\bibliography{references}

\appendix
\newpage
\normalem
\section{Bound on Lipschitz constant (Proof of Theorem \ref{thm:many-body-lip-const})}
\label{app-sec:pf-lipschitz-const}

We will first perform some elementary computations. Let
\begin{equation}
    V_{\text{pair}} = \frac{1}{2}\sum_{i=0}^{\eta-1}\sum_{j\neq i}V(r_{ij})
    \label{eq:vpair}
\end{equation}
First, denote the partial derivative with respect to particle $i$'s $k$'th direction as $\partial_{i,k} := \partial_{x_{i,k}}$. Then, the derivatives satisfy
\begin{align*}
    \partial_{i,k} V(r_{ij}) &= V'(r_{ij}) \frac{x_{i,k}-x_{j,k}}{r_{ij}}\\
    \partial_{i,k}V_{\text{pair}} &=\sum_{j\neq i}V'(r_{ij}) \frac{x_{i,k}-x_{j,k}}{r_{ij}}.
\end{align*}
For the second derivatives, we have
\begin{align*}
    \partial_{i,l}\partial_{i,k}V(r_{ij}) &= \partial_{i,l}\left( V'(r_{ij}) \frac{x_{i,k}-x_{j,k}}{r_{ij}}\right),
\end{align*}
where
\begin{align*}
    \partial_{i,l}\left( V'(r_{ij}) \frac{x_{i,k}-x_{j,k}}{r_{ij}}\right) &= \frac{(x_{i,k}-x_{j,k})(x_{i,l}-x_{j,l}) V''(r_{ij})}{r_{ij}^2} - \frac{(x_{i,k}-x_{j,k})(x_{i,l}-x_{j,l})V'(r_{ij})}{r_{ij}^3}.
\end{align*}
We also have
\begin{align*}
    \partial_{j,l}\partial_{i,k}V_{\text{pair}} &= -\frac{(x_{i,k}-x_{j,k})(x_{i,l}-x_{j,l}) V''(r_{ij})}{r_{ij}^2} + \frac{(x_{i,k}-x_{j,k})(x_{i,l}-x_{j,l})V'(r_{ij})}{r_{ij}^3}\\
    &= - \partial_{i,l}\partial_{i,k} V(r_{ij}).
\end{align*}

Therefore, the Hessian for a single $V(r_{ij})$ takes the form,
\begin{align}
\label{eq:Hij}
    H(V(r_{ij})) &= \begin{pmatrix}
        \nabla_i \nabla_i^T V(r_{ij}) & \nabla_i \nabla_j^T V(r_{ij})\\
        \nabla_j \nabla_i^T V(r_{ij}) & \nabla_j \nabla_j^T V(r_{ij})
    \end{pmatrix}\\
    &=\begin{pmatrix}
        A_{ij} & -A_{ij}\\
        -A_{ij} & A_{ij}
    \end{pmatrix},
\end{align}
where,
\begin{align*}
    A_{ij} = V''(r_{ij}) \frac{(x_i-x_j)(x_i-x_j)^T}{r_{ij}^2} + \frac{V'(r_{ij})}{r_{ij}}\left(I_d - \frac{(x_i - x_j)(x_i-x_j)^T}{r_{ij}^2}\right).
\end{align*}
Due to the assumption the $\nabla V$ is $\gamma$-Lipschitz, we also have
\begin{equation}
    \sup_{r_{ij}<\infty} \norm{H(V(r_{ij}))} = \gamma.
\end{equation}

Now, let us also find the matrix elements of the Hessian for the full potential. Consider the case for $\eta = 3$. Then, the Hessian will take the form
\begin{align*}
    H_{\text{pair}} &= \begin{pmatrix}
        \nabla_0 \nabla_0^T (V_{01} + V_{02}) & \nabla_0 \nabla_1^T V_{01} & \nabla_0 \nabla_2^T V_{02}\\
        \nabla_0 \nabla_1^T V_{01} & \nabla_1 \nabla_1^T (V_{01} + V_{12}) & \nabla_1 \nabla_2^T V_{12}\\
        \nabla_0 \nabla_2^T V_{02} & \nabla_1 \nabla_2^T V_{12} & \nabla_2 \nabla_2^T (V_{02} + V_{12})
    \end{pmatrix}\\
    &=\begin{pmatrix}
        \nabla_0 \nabla_0^T V_{01} & \nabla_0 \nabla_1^T V_{01} & 0\\
        \nabla_1\nabla_0^T V_{01} & \nabla_1 \nabla_1^T V_{01} & 0\\
        0 & 0 &0
    \end{pmatrix} + \begin{pmatrix}
        \nabla_0 \nabla_0^T V_{02} & 0 &\nabla_0 \nabla_2^T V_{02}\\
        0 & 0 &0\\
        \nabla_2\nabla_0^T V_{0,2} & 0 & \nabla_2\nabla_2^T V_{0,2}
    \end{pmatrix} + \begin{pmatrix}
        0 & 0 & 0\\
        0 &\nabla_1\nabla_1^T V_{1,2} & \nabla_1 \nabla_2^T V_{1,2}\\
        0 & \nabla_2\nabla_1^T V_{1,2} & \nabla_2\nabla_2^T V_{1,2}
    \end{pmatrix}.
\end{align*}
Now, consider some arbitrary vector $\mathbf{u} = \begin{pmatrix}
    u_0 & u_1 & u_2
\end{pmatrix}^T \in \mathbb{R}^{3 d}$ (i.e. each $u_j \in \mathbb{R}^d$). Then, it is straightforward to observe that
\begin{align*}
    H_{\text{pair}}\begin{pmatrix}
        u_0\\
        u_1\\
        u_2
    \end{pmatrix} 
    &= H_{0,1}\begin{pmatrix}
        u_0\\
        u_1
    \end{pmatrix} + H_{0,2}\begin{pmatrix}
        u_0\\
        u_2
    \end{pmatrix} + H_{1,2}\begin{pmatrix}
        u_1\\
        u_2
    \end{pmatrix},
\end{align*}
where $H_{i,j}$ is the $2d\times 2d$ Hessian for $V(r_{ij})$ in Eq. \eqref{eq:Hij}. This pattern persists for arbitrary $\eta$, in particular, for any $\mathbf{u} = \begin{pmatrix}
    u_0&\cdots&u_{\eta-1}
\end{pmatrix}^T \in \mathbb{R}^{\eta d}$ with $u_i \in \mathbb{R}^d$ we have 
\begin{align*}
    H_{\text{pair}}\mathbf{u} = \frac{1}{2}\sum_{i=0}^{\eta -1} \sum_{j \neq i} H_{i,j}\begin{pmatrix}
        u_i\\
        u_j
    \end{pmatrix}.
\end{align*}

Then,
\begin{align*}
    \mathbf{u}^TH_{\text{pair}}\mathbf{u} &= \frac{1}{2}\sum_{i=0}^{\eta -1} \sum_{j \neq i} \begin{pmatrix}
        u_i^T & u_j^T
    \end{pmatrix}\begin{pmatrix}
        A_{ij} & -A_{ij}\\
        -A_{ij} & A_{ij}
    \end{pmatrix}\begin{pmatrix}
        u_i\\
        u_j
    \end{pmatrix}\\
    &=\frac{1}{2}\sum_{i=0}^{\eta -1} \sum_{j \neq i} u_i^TA_{ij}u_i - u_i^TA_{ij}u_j + u_j^T A_{ij}u_j - u_j^T A_{ij} u_i,
\end{align*}
which, due to the symmetry of $A_{ij} \in \mathbb{R}^{d\times d}$ resulting from the equality of mixed partials, gives
\begin{equation}
\begin{aligned}
    \mathbf{u}^TH_{\text{pair}}\mathbf{u} &= \frac{1}{2}\sum_{i=0}^{\eta -1} \sum_{j \neq i} u_i^TA_{ij}u_i - 2 u_i^TA_{ij}u_j + u_j^T A_{ij}u_j\\
    &= \frac{1}{2}\sum_{i=0}^{\eta -1} \sum_{j \neq i} (u_i^T - u_j^T)A_{ij}(u_i-u_j).
    \label{eq:hess-quad-form}
\end{aligned}
\end{equation}

We now construct a particular configuration to obtain the lower bound and prove the following theorem.

\lipschitzConst*

\begin{proof}
    Let $V:[0,\infty) \rightarrow \mathbb{R} $ be a twice differentiable function with bounded derivatives. Assume that $V'$ is Lipschitz continuous with Lipschitz constant $\gamma>0$,
    \begin{align*}
        \left|V'(r)-V'(s)\right| \leq \gamma \left|r-s\right| \hspace{.2cm} \forall r,s \geq 0.
    \end{align*}
    Since $V''$ is continuous, this implies $\sup_{r\geq 0}\left|V''(r)\right| = \gamma$, and there exits $r_0 > 0$ such that $V''(r_0) \geq \gamma/2$.

    Consider a system of $\eta$ particles in $\mathbb{R}^d$, with positions $\mathbf{x} = \left(x_0, \ldots, x_{\eta-1}\right) \in \mathbb{R}^{\eta d}$. Define the pairwise potential
    \begin{align*}
        V_{\text{pair}}(\mathbf{x}) = \frac{1}{2}\sum_{i\neq j}V(r_{ij}), \hspace{.2cm} \text{ where } r_{ij} = \norm{\mathbf{x}_i -\mathbf{x}_j}.
    \end{align*}
    The gradient $\nabla V_{\text{pair}}:\mathbb{R}^{\eta d} \rightarrow \mathbb{R}^{\eta d}$ is Lipschitz continuous if there exists $\gamma'>0$ such that 
    \begin{align*}
        \norm{\nabla V_{\text{pair}}(\mathbf{x}) - \nabla V_{\text{pair}}(\mathbf{y})} \leq \gamma' \norm{\mathbf{x}-\mathbf{y}} \hspace{.2cm} \forall x,y \in \mathbb{R}^{\eta d}.
    \end{align*}

    Since $V_{\text{pair}}$ is $C^2$, the optimal $\gamma'$ is the supremum of the operator norm of the Hessian
    \begin{align*}
        \gamma' &= \sup_{\mathbf{x}}\norm{H(\mathbf{x})} \text{ where,}\\
        H(\mathbf{x}) &= \nabla \nabla^T V_{\text{pair}}(\mathbf{x}).
    \end{align*}
    From Eq. \eqref{eq:hess-quad-form} at the beginning of this section, for any vector $\mathbf{u} \in \mathbb{R}^{\eta d}$ with $u_i \in \mathbb{R}^d$, the quadratic form
    \begin{align*}
        \mathbf{u}^T H(\mathbf{x}) \mathbf{u} = \frac{1}{2}\sum_{i=0}^{\eta - 1}\sum_{j\neq i}(u_i^T-u_j^T)A_{ij}(u_i-u_j).
    \end{align*}

    The following argument can be made much simpler if one is allowed to place multiple particles at the same position. Indeed, nothing in our stated assumptions prevents this from being the case. Nevertheless, it is quite common to restrict particles from occupying the same location, and so we proceed with an argument lower bounding the Lipschitz constant in this case.
    
    Assume $\eta$ is even. The case for $\eta $ odd is similar.
    
    Let $r_0$ such that $\left|V''(r_0)\right|\geq \gamma/2$. Such an $r_0$ exists because $\sup_{r\geq 0}\left|V''(r)\right| = \gamma$ and $V''$ is continuous. Let ${r}_0 = {A}-{B}$, with ${A}, {B} \in \mathbb{R}^d$. Place $\eta/2$ particles co-linearly with $\mathbf{r}_0$ within a radius $\delta$ of ${r}_0$ surrounding the point $A \in \mathbb{R}^d$ and the remaining $\eta/2$ particles co-linearly with $\mathbf{r}_0$ in a $\delta$ ball surrounding the point $B \in \mathbb{R}^d$. In addition, we assume $\delta$ is chosen such that $\delta<r_0/2$.  Define $B_\delta(A) = \{ \mathbf{x}: \norm{\mathbf{x} -A} \leq \delta\}$, and $B_\delta(B)$ analogously. 

    Notice that for any $\mathbf{x} \in B_\delta(A)$ and $\mathbf{y} \in B_\delta(B)$, that 
    \begin{align*}
        r_0 - 2\delta \leq \norm{\mathbf{x} - \mathbf{y}} \leq r_0 + 2\delta.
    \end{align*}
    Fix an $\epsilon \leq \gamma/4$. We now wish to show that we can choose a $\delta>0$ so that $\left|V''(r_0 \pm 2\delta)\right|\geq \gamma/2 - \epsilon \geq \gamma/4$. This follows from the continuity of  $V''$. In particular, there exists $\delta > 0$ such that for all $y$ satisfying $\left|y-r_0\right|\leq \delta$, we have 
    \begin{align*}
        \left|V''(y) - V''(r_0)\right| < \epsilon,
    \end{align*}
    thus,
    \begin{align*}
    -\epsilon + \frac{\gamma}{2}<V''(y).
    \end{align*}
    With the guarantee of such a $\delta$, we can proceed with lower bounding the Lipschitz constant.

    Define a test vector $\mathbf{v} = (v_0,\ldots, v_{\eta-1})$ by
    \begin{align*}
        v_i = \begin{cases}
            e, & \text{if particle $i$ is in $B_\epsilon(A)$ }\\
            -e, & \text{if particle $i$ is in $B_\epsilon(B)$},
        \end{cases}
    \end{align*}
    where $e = \frac{A-B}{\norm{A-B}}$ is a unit vector. Then, $\norm{v_i}=1$ $\forall i$, so $\norm{\mathbf{v}}^2 = \sum_{i=0}^{\eta -1}\norm{v_i}^2 = \eta$.

    Now, evaluating the quadratic form,
    \begin{align*}
        \mathbf{v}^T H(\mathbf{x}) \mathbf{v} = \frac{1}{2}\sum_{i=0}^{\eta - 1}\sum_{j\neq i}(v_i^T-v_j^T)A_{ij}(\mathbf{x})(v_i-v_j),
    \end{align*}
    we can see that if $i$ and $j$ are in the same cluster that $v_i - v_j = 0$, whereas if $i$ and $j$ are in different clusters, $v_i - v_j = \pm 2e$. Thus, in this case
    \begin{align*}
    (v_i^T-v_j^T)A_{ij}(\mathbf{x})(v_i-v_j) &=(v_i^T-v_j^T)\left(V''(r_{ij}) \frac{(x_i-x_j)(x_i-x_j)^T}{\norm{x_i-x_j}^2} + \frac{V'(r_{ij})}{\norm{x_i-x_j}}\left(I_d - \frac{(x_i - x_j)(x_i-x_j)^T}{\norm{x_i-x_j}^2}\right)\right)(v_i-v_j)\\
    &=\frac{4 r_0^T}{\norm{r_0}^2}\left(V''(r_{ij}) \frac{(x_i-x_j)(x_i-x_j)^T}{\norm{x_i-x_j}^2} + \frac{V'(r_{ij})}{\norm{x_i-x_j}}\left(I_d - \frac{(x_i - x_j)(x_i-x_j)^T}{\norm{x_i-x_j}^2}\right)\right)r_0\\
    &=\frac{4 }{\norm{r_0}^2}\left(V''(r_{ij}) \frac{r_0^T(x_i-x_j)(x_i-x_j)^Tr_0}{\norm{x_i-x_j}^2} + \frac{V'(r_{ij})}{\norm{x_i-x_j}}\left(\norm{r_0}^2 - \frac{r_0^T(x_i - x_j)(x_i-x_j)^Tr_0}{\norm{x_i-x_j}^2}\right)\right)\\
    &=\frac{4 }{\norm{r_0}^2}\left(V''(r_{ij}) \frac{\norm{r_0}^2 \norm{x_i-x_j}^2}{\norm{x_i-x_j}^2} + \frac{V'(r_{ij})}{\norm{x_i-x_j}}\left(\norm{r_0}^2 - \frac{\norm{r_0}^2 \norm{x_i-x_j}^2}{\norm{x_i-x_j}^2}\right)\right)\\
    &=4 V''(r_{ij}).
    \end{align*}
    
    Since $r_0 - \delta_0 \leq r_{ij} \leq r_{0}+\delta_0 $, we know that 
    \begin{align*}
        V''\left(\norm{\mathbf{x}_i -\mathbf{x}_j}\right) \geq \frac{\gamma}{4} .
    \end{align*}
    There are $\eta^2/4$ such terms. Thus,
    \begin{align*}
        \gamma' \geq \frac{\mathbf{v}^T H (\mathbf{x}) \mathbf{v}}{\norm{\mathbf{v}}^2} \geq \frac{\gamma \eta^2 }{4\eta } = \frac{\gamma\eta}{4} \in \Omega\left(\gamma \eta \right).
    \end{align*}

\end{proof}

\section{Technical Background on Kolmogorov Equations}
\label{app-sec:kolmogorov-background}
Recall the Fokker-Planck operator,
\begin{equation}
     \mathcal{F}[\rho] = \nabla \cdot\left( \rho \nabla V + \beta^{-1}\nabla\rho \right).
     \label{app-eq:fokker-planck-operator}
\end{equation}
$F = -\nabla V$ is the force generated by the potential $V$, and $\beta$ is the inverse temperature. The operator is the dynamical generator for the forward Kolmogorov equation
\begin{equation}
    \mathcal{F}[P](\mathbf{x},t|\mathbf{x}',0) = \partial_{t}P(\mathbf{x},t|\mathbf{x}',0).
\end{equation}
Therefore, the operator $\mathcal{F}$ generates diffusive and convective dynamics, and is closely related to the convection-diffusion equation. 
The differential operator $\mathcal{F}$ is negative semi-definite and the eigenfunction associated to the zero-mode is the equilibrium Boltzmann distribution,
\begin{equation}
    \mu(\mathbf{x}) = \frac{e^{-\beta V(\mathbf{x})}}{\mathcal{Z}},
\end{equation}
where $\mathcal{Z}$ is the associated normalization factor or partition function,
\begin{equation}
    \mathcal{Z} = \int e^{-\beta V(\mathbf{x})} d\mathbf{x}.
\end{equation}
$\mathcal{F}$ is self-adjoint with respect to the measure generated by $\mu$, that is, for any smooth functions $u$ and $v$ the following holds,
\begin{equation}
    \langle u, \mathcal{F}[v]\rangle_{\mu} \equiv \frac{1}{\mathcal{Z}}\int e^{-\beta V(\mathbf{x})}  u(\mathbf{x})  \mathcal{F}[v](\mathbf{x}) d\mathbf{x}= - \langle \nabla u, \nabla v\rangle_{\mu} ,
\end{equation}
with the last equality resulting from applying integration by parts and using the confining property of $V$ to evaluate the boundary term. As the operator $\mathcal{F}$ is not self-adjoint under the standard $L^2$ inner product, its adjoint is given by
\begin{equation}
    \mathcal{F}^\dagger = \nabla V \cdot \nabla + \beta^{-1}\Delta,
    \label{app-eq:BKE generator}
\end{equation}
and is the generator for the closely related backward Kolmogorov equation
\begin{equation}
    -\mathcal{F}^\dagger [Q](\mathbf{x}',t'|\mathbf{x},t) = \partial_t Q(\mathbf{x}', t'|\mathbf{x},t).
    \label{app-eq:BKE}
\end{equation}

The operator $\mathcal{F}$ (and its adjoint) can be related to an explicitly self-adjoint operator via a similarity transformation generated by the square root of the equilibrium measure,
\begin{equation}
    \rho_\beta(\mathbf{x}) = \sqrt{\mu(\mathbf{x})}.
\end{equation}
We define 
\begin{equation}
    \HB[u] = \rho_\beta ^{-1} \mathcal{F}[ \rho_\beta u],
    \label{app-eq:sim-transf}
\end{equation}
and through straightforward calculations find that the resulting operator can be expressed as
\begin{equation}
    \HB =  \beta^{-1}\Delta + U_\beta(\mathbf{x}),
\end{equation}
where 
\begin{equation}
    U_\beta(\mathbf{x}) = -\frac{\beta}{4}\norm{\nabla V}^2(\mathbf{x}) + \frac{1}{2}\Delta V(\mathbf{x}),
    \label{app-eq:sym-pot}
\end{equation}
is a scalar multiplication operator.
As a result, $\HB$ is a \textit{Schr\"odinger}-type operator, and we thereby refer to this as the self-adjoint representation of $\mathcal{F}$. Since the spectrum of any operator is preserved under similarity transformation, and we know that $\mathcal{F}$ is negative semi-definite, $\HB$ inherits this property as well. We note that this relationship between the Fokker-Planck operator and Schr\"odinger-type operator is somewhat well known, see e.g. Refs. \cite{wadia_solution_2022, tong_university_nodate, lengOperatorLevelQuantumAcceleration2025}.

For a many-body potential $V_{\text{tot}}$, expressed in terms of the joint variable $\mathbf{x} = \left(\mathbf{x}^0, \mathbf{x}^2, \ldots, \mathbf{x}^{\eta-1}\right)$ given by a summation over pair-wise interactions governed by $V$, we have,
\begin{equation}
    V_{\text{tot}}\left(\mathbf{x}\right) = \sum_{i=0}^{\eta -1}\sum_{j<i}V\left(\mathbf{x}^i,\mathbf{x}^j\right),
\end{equation}
where $\mathbf{x}^i \in \mathbb{R}^d$. We will consider functions $V$ that are symmetric functions of their arguments, so that $\partial_{\mathbf{x}^i}V\left(\mathbf{x}^i,\mathbf{x}^j\right) = \partial_{\mathbf{x}^i}V\left(\mathbf{x}^j,\mathbf{x}^i\right)$. In particular, we assume spherically symmetric pair-wise interactions of the form $V\left(r_{ij}\right)$ where $r_{ij} = \norm{\mathbf{x}^i - \mathbf{x}^j}$ is the Euclidean distance in $d$ dimensions. The terms in the effective potential for the self-adjoint representation of the Fokker-Planck operator $U_\beta$ take the form,
\begin{equation}
\begin{aligned}
    \norm{\nabla V_{\text{tot}}}^2(\mathbf{x}) &= \sum_{k=0}^{\eta -1}\sum_{j<k}\norm{\partial_{\mathbf{x}^k} V\left(\mathbf{x}^k,\mathbf{x}^j\right)}^2 = \sum_{k=0}^{\eta -1 }\sum_{j<k}\sum_{i=0}^{d-1}\left|\partial_{x^k_i}V\left(\mathbf{x}^k,\mathbf{x}^j\right)\right|^2\\
    \Delta V(\mathbf{x}) &= \sum_{k=0}^{\eta -1}\sum_{j<k}\sum_{i=0}^{d-1}\partial_{x^k_i}^2 V(\mathbf{x}^k,\mathbf{x}^j),
\end{aligned}
\end{equation}
where we write $\partial_{\mathbf{x}^k}$ to mean the $d$-dimensional gradient vector of partial derivatives for particle $k$ and $\partial_{x^k_i}^{(2)}$ to mean the first (second) derivative of particle $k$ with respect to its $i$th dimension.

The effective potential for the self-adjoint representation of the Fokker-Planck operator is explicitly written as,
\begin{equation}
    U_{\beta}\left(\mathbf{x}\right) = \frac{-\beta}{4}\sum_{i=0}^{\eta - 1} \norm{\partial_{\mathbf{x}^k} V\left(\mathbf{x}^k,\mathbf{x}^j\right)}^2 + \frac{1}{2}\sum_{k=0}^{\eta -1}\sum_{j<k}\sum_{i=0}^{d-1}\partial_{x^k_i}^2 V\left(\mathbf{x}^k,\mathbf{x}^j\right),
\end{equation}
and the two particle and many-body interaction terms are 
\begin{equation}
\begin{aligned}
    U_{ij} &= \frac{-\beta}{4}\norm{\partial_{\mathbf{x}^i}V\left(\mathbf{x}^i,\mathbf{x}^j\right)}^2+\frac{1}{2}\sum_{k=0}^{d-1}\partial_{x^i_k}^2V\left(\mathbf{x}^i,\mathbf{x}^j\right)\\
    U_{tot} &= \sum_{i=0}^{\eta}\sum_{j<i} U_{ij}
\end{aligned}
\end{equation}

An important quantity of interest is the so-called \textit{reaction rate}, which describes the probability of trajectories beginning in some region of configuration space ending in another region of configuration space. This describes an abstract notion of ``reactant'' and ``product'' states, e.g. in a chemical reaction, and we refer to the set of states where the trajectories start $(R)$ and terminate $(P)$ as the reactants and products. Although the reaction rate is typically defined as a steady state property, there exists a related time-dependent quantity that is more useful for this work. The time-dependent reaction rate $k_{RP}(T)$ between subsets of the configuration space, $R$ and $P$ is defined as the time-average quantity
\begin{equation}
    k_{RP}(T) := \frac{1}{T\overline{p}_R}\int dx \int dx' 1_{P}(x')P(x',T|x,0)\mu(x)1_{R}(x),
    \label{app-eq:td-rxn-rate}
\end{equation}
with
\begin{equation}
    1_{R}(x) = \begin{cases}
        1 & x\in R\\
        0 & \text{else}
    \end{cases},
\end{equation}
and
\begin{equation}
    \overline{p}_R = \int dx 1_R(x) \mu(x).
\end{equation}
$P(x',T|x,0)$ is the probability of being in state $x'$ at time $T$ given the state $x$ at time 0 and is obtained by the propagator for the FPE,
\begin{equation}
    P(x',T|x,0) = \bra{x'}e^{T\mathcal{F}}\ket{x}.
\end{equation}
The time-dependent reaction rate can be interpreted as the probability of a trajectory starting in $R$ to terminate in $P$ at time $T$.

We may define an analogous expression for the time-dependent reaction rate $k_{RP}(T)$ in terms of the propagator in the self-adjoint representation, $\widetilde{P}(x', t| x, 0) = \bra{x'}e^{\HB t}\ket{x}$. By applying \eqref{app-eq:sim-transf} and \eqref{app-eq:td-rxn-rate} it is straightforward to show that,
\begin{equation}
     k_{RP}(T) = \frac{1}{T \overline{p}_R}\int dx \int dx' \bra{x'}1_{P}(x')\rho(x') e^{T \HB}  \rho(x) 1_{R}(x)\ket{x}.
    \label{app-eq:sym-rxn-rate}
\end{equation}
We define the normalized quantum states
\begin{equation}
\begin{aligned}
    \ket{R} &= \int1_R(x) \frac{e^{-\beta V(x)/2}}{\sqrt{\overline{p}_R}}\ket{x},
\end{aligned}
\end{equation}
and see that 
\begin{align*}
    k_{RP}(T) = \frac{1}{T}\sqrt{\frac{\overline{p}_P}{\overline{p}_R}} \bra{P}e^{T \HB} \ket{R}.
\end{align*}
From the above, we extract the reactive flux 
\begin{equation}
\nu_{RP}(T) := \bra{P} e^{T \HB } \ket{R}.
\label{app-eq:rxn-flux}
\end{equation}

\section{Equilibrium state preparation in locally convex region}
\label{app-sec:locally-cvx-stateprep}

Before we begin our discussion of the particular construction that we use to demonstrate efficient state preparation, we briefly review some elementary results on the rate of convergence to equilibrium in the convex setting. First, we establish exponential decay rate $\lambda$ to the steady state $(e^{-\beta V})$ when the steady state satisfies a Poincar\'e inequality with constant $\lambda$.

\begin{lemma}[Convergence rate is Poincar\'e constant (adapted from Ref. \cite{markowich_trend_nodate})]
\label{lem:poincare}
    Let $L$ be the generator for the Fokker-Planck equation with steady state $\mu= e^{-\beta V}$, a suitably normalized distribution. Denote $L^2(\mu)$ as the weighted $L^2$ norm. 
    Suppose $\mu$ satisfies a Poincar\'e inequality with constant $\lambda \beta $, i.e.
    \begin{equation}
    \forall g \in L^2(\mu), \quad \int g(x) e^{-\beta V(x)} dx = 0 \implies \int \norm{\nabla g(x)}^2 \mu dx \geq \lambda \beta  \int g^2(x)\mu dx.
        \label{eq:poincare-ineq}
    \end{equation}
    Then, for any $L^1$ normalized initial state $\rho_0 = h_0 e^{-\beta V} = u_0 e^{-\beta V/2}$ the following inequalities hold,
    \begin{equation}
    \begin{aligned}
        \norm{h(t) - 1}_{L^2(\mu)} &\leq e^{-2\lambda t}\norm{h_0 - 1}_{L^2(\mu)}\\
        \chi^2(\rho(t), \mu) &\leq e^{-2\lambda t}\chi^2(\rho_0, \mu)\\
        \norm{u(t) - e^{-\beta V/2}}_2^2 &\leq e^{-2\lambda t} \norm{u_0 - e^{-\beta V/2}}_2^2.
    \end{aligned}
    \end{equation}
\end{lemma}
\begin{proof}
    The Fokker-Planck equation \eqref{app-eq:fokker-planck-operator} under the change of variable $\rho = h e^{-\beta V}$, reads
    \begin{equation}
        \frac{\partial h}{\partial t}(x,t) = -\nabla V \cdot \nabla h + \beta^{-1} \Delta h,
        \label{eq:fpe-cov}
    \end{equation}
    where we extract $L = -\nabla V \cdot \nabla + \beta^{-1} \Delta$. For confining potentials $V$, the operator $L$ is self-adjoint with respect to the equilibrium measure, and in particular, it can be shown that
    \begin{align*}
        \langle f, L[g]\rangle_\mu :=  \int f(x)  L[g](x) e^{-\beta V(x)} dx = -\beta^{-1} \langle \nabla f, \nabla g\rangle_\mu,
    \end{align*}
    for any $f, g \in C^2 \cap L^2(\mu)$.

    Using Eq. \eqref{eq:fpe-cov} one can perform the computation,
    \begin{align*}
        \frac{d}{dt}\int (h(x,t)-1)^2 \mu(x) dx  &= -2\beta^{-1} \int \norm{\nabla h}^2(x,t)\mu(x) dx.
    \end{align*}
    Now for $\rho = h e^{-\beta V}$ with $\int \rho(x) dx = 1$, we have that $\int (h(x) - 1)e^{-\beta V(x)}dx = 0$, therefore $h-1$ satisfies the Poincar\'e inequality \eqref{eq:poincare-ineq}
    \begin{align*}
        \int \norm{\nabla h}^2(x,t) \mu(x) dx \geq \lambda \beta  \int (h(x,t)-1)^2 \mu(x) dx.
    \end{align*}
    This entails the estimate,
    \begin{align*}
        -\frac{d}{dt}\int (h(x,t)-1)^2 \mu(x) dx &=2\beta^{-1} \int \norm{\nabla h}^2(x,t)\mu(x) dx \geq 2\lambda \int (h(x,t)-1)^2 \mu(x) dx\\
    \end{align*}
    or that 
    \begin{align*}
        \int (h(x,t)-1)^2 \mu(x) dx \leq e^{-2\lambda t}  \int (h(x,0)-1)^2 \mu(x)dx.
    \end{align*}
    Thus, for any initial state $\rho_0 = h_0 e^{-\beta V}$, with $\int \rho_0 dx = 1$, we have by Gr\"onwall's inequality
    \begin{align*}
        \norm{h(x,t) - 1}_{L^2(\mu)} \leq e^{-2\lambda t}\norm{h(x,0) - 1}_{L^2(\mu)}.
    \end{align*}
    Equivalently,
    \begin{align*}
        \chi^2(\rho, e^{-\beta V}) &= \int \left|\rho - e^{-\beta V}\right|^2e^{\beta V} dx\\
        &= \int \left|h e^{-\beta V} -1 e^{-\beta V}\right|^2e^{\beta V}dx\\
        &= \int \left|h -1\right|^2 e^{-\beta V}dx\\
        &= \norm{h-1}^2_{L^2(\mu)}.
    \end{align*}
    We further have, that for $\rho = e^{-\beta V/2} u$,
    \begin{align*}
    \chi^2(\rho, e^{-\beta V}) &= \int \left|\rho - e^{-\beta V}\right|^2e^{\beta V} dx\\
    &= \int \left|e^{-\beta V/2} u - e^{-\beta V}\right|^2e^{\beta V} dx\\
    &= \int  \left|u -e^{-\beta V/2}\right|^2 dx\\
    &= \norm{u - e^{-\beta V/2}}_2^2,
    \end{align*}
    thus, convergence of $h$ to $1$ is equivalent to the convergence of $\rho$ to $e^{-\beta V}$ in $\chi^2$ divergence, and convergence of $\rho$ to $e^{-\beta V}$ in $\chi^2$ divergence is equivalent to convergence of $u$ to $\sqrt{e^{-\beta V}}$ in the standard $L^2$ norm.
\end{proof}

The above lemma guarantees that when the measure $L^2(e^{-\beta V})$ supports a space of functions with a positive Poincar\'e constant $\lambda$ that an exponential convergence rate of $\lambda$ to equilibrium can be established. An equivalent condition, provided by the Bakry-Emery criteria \cite{bakryDiffusionsHypercontractives1985}, shows that the Poincar\'e constant can be related to the convexity of the potential $V$. In particular, Theorems 1 and 2 of Ref. \cite{markowich_trend_nodate} establish that if $V$ satisfies 
\begin{equation*}
    \text{Hess}(V) \geq \lambda I,
\end{equation*}
then $e^{-\beta V}$ satisfies the Poincar\'e inequality \eqref{eq:poincare-ineq} with constant $\frac{\lambda}{\beta}$. Thus, the assumption that the  potential is $m$-strongly convex is sufficient to prove that the dynamics converge to the equilibrium distribution at a rate of $O(e^{- t m})$, and thus establishes a mixing time estimate $t_\text{mix} \in O\left(\frac{1}{m}\right)$.

\subsection{Quantum reactant (product) state preparation }

Our goal is to prepare a quantum state corresponding to $\sqrt{\pi_R(x)}$ where
$\pi_R(x) \propto e^{-\beta V(x)}1_{R}(x)$, where the restriction of the potential $\left. V(x)\right|_{x\in R}$ is $m$-strongly convex.  We remark that weaker conditions can also be used to obtain exponential convergence. Throughout this section, $R$ can be replaced with $P$ for the preparation of the product distribution. We would like to establish that even if the potential $V(x)$ has no global convex structure that the mixing in the region $R$ is exponentially fast, and that there exists a quantum algorithm that prepares an encoding of $\pi_R(x)$ into a quantum state with complexity proportional to the mixing time $t_{\text{mix}} \in O\left(\frac{1}{m}\right)$.

Consider the Fokker-Planck operator restricted to the region $R$, an open subset of $\mathbb{R}^{\eta d}$, endowed with homogeneous Neumann boundary conditions,
\begin{equation}
\begin{aligned}
    \nabla \cdot \left(\rho \nabla V + \beta^{-1}\nabla \rho\right) &= \partial_t \rho\\
    \mathbf{n}\cdot \left(\nabla V \rho(\mathbf{x}) + \beta^{-1} \nabla \rho(\mathbf{x})\right) &= 0 \text{ for } \mathbf{x} \in \partial R
\end{aligned}
\end{equation}
where $\mathbf{n}$ is an outward pointing unit vector on $\partial R$. Under the change of variable $\rho = e^{-\beta V/2}\psi$, we obtain a Robin boundary value problem with the self-adjoint operator $\mathcal{H}_\beta$,
\begin{equation}
\begin{aligned}
\label{app-eq:steady-state-bvp}
    (\beta^{-1}\Delta - \frac{\beta}{4} \norm{\nabla V}^2 + \frac{1}{2}\Delta V)\psi &= \partial_t \psi\\
    \mathbf{n}\cdot \left(\frac{1}{2} \psi(\mathbf{x}) \nabla V(\mathbf{x}) + \frac{1}{\beta}\nabla \psi(\mathbf{x}) \right) &= 0 \text{ for } \mathbf{x} \in \partial R.
\end{aligned}
\end{equation}

Since $V$ is not necessarily confining on $R$, more effort is needed to enforce the Robin boundary condition using the plane wave discretization we employ in the full configuration space. One straightforward method is to introduce an auxiliary confining potential,
\begin{equation}
    V_R(\mathbf{x}) = \begin{cases}
        V(\mathbf{x}) & \mathbf{x} \in R\\
        V(\mathbf{x}) + \kappa \text{dist}\left(\mathbf{x},\partial R\right)^2 & \mathbf{x} \in R^c,
    \end{cases}
\end{equation}
where $\kappa > 1$ is a stiffness parameter that improves the accuracy of the Robin boundary value problem as $\kappa \rightarrow \infty$.

 A demonstration of this construction is provided in Fig. \ref{fig:local-cvx-extension} below.
Since $V_R$ and $V$ match at $\partial R$, we know that $V_R$ is at least a $C^0$ extension of $V$, however, we also have that $\nabla V$ and $\nabla V_R$ match at $\partial R$,
\begin{align*}
    \nabla V_R(\mathbf{x}) = \begin{cases}
        \nabla V(\mathbf{x}) &\mathbf{x} \in R\\
        \nabla V(\mathbf{x})+2 \kappa \left(\mathbf{x} - \partial R\right) & \mathbf{x} \in R^c.
    \end{cases}
\end{align*}
Thus as $\mathbf{x} \rightarrow \partial R$, $\nabla V_R \rightarrow \nabla V(\mathbf{x})$ and we have continuity in the first derivative, and $V_R$ is a $C^1$ extension of $V$. This in turn allows us to at least establish a convergence rate of $O(N^{-1})$ in the plane wave discretization and avoid problematic Gibbs phenomenon.

With this augmented potential, we now wish to solve the problem
\begin{equation}
\begin{aligned}
\label{app-eq:augmented-problem}
    \left(\beta^{-1} \Delta - \frac{\beta}{4}V_R (\mathbf{x}) + \frac{1}{2}\Delta V_R\right)\psi &= \partial_t \psi \\
    \psi(0) &= \psi_0.
\end{aligned}
\end{equation}
In the following Lemma \ref{lem:kappa-scaling-steady-state}, we establish bounds on how large $\kappa$ needs to be to ensure that the steady state of Eq. \eqref{app-eq:augmented-problem} well approximates the steady state of Eq. \eqref{app-eq:steady-state-bvp}. The result takes as a simplifying assumption that $R$ defines a spherical open subset of $\mathbb{R}^{\eta d}$, but this is not strictly necessary, and serves as an example where bounds can be explicitly computed. We expect similar results would hold in less symmetric settings. With a well-behaved enough potential, for example, we may be able extend $R$ onto a larger, but simpler, domain.

\begin{lemma}
    \label{lem:kappa-scaling-steady-state}
    Let $V_R$ and $V$ be as above. Let $R$ define a spherical subset of $\mathbb{R}^{\eta d}$. That is for some $\mathbf{x}_0 \in \mathbb{R}^{\eta d}$ and some $r_0 > 0$,
    \begin{align*}
        R = \{\mathbf{x} \in \mathbb{R}^{\eta d}: \norm{\mathbf{x} - \mathbf{x}_0} < r_0\}.
    \end{align*}
    Then, for any $\epsilon > 0$, we can choose $\kappa \in O\left(\frac{1}{\mathcal{Z}_R^2 \beta \epsilon^4}\right),$ to ensure that 
    \begin{align*}
        \norm{e^{-\beta V_R(\mathbf{x})} - e^{-\beta V(\mathbf{x})}}_2 \leq \epsilon,
    \end{align*}
    where $\norm{\cdot}_2$ is the $L^2$ norm over $\mathbb{R}^{\eta d}$.
\end{lemma}
\begin{proof}
    Let $\rho_R(\mathbf{x}) = \frac{e^{-\beta V_R(\mathbf{x})/2}}{\sqrt{\mathcal{Z}}}$ and $\rho(\mathbf{x}) = \frac{e^{-\beta V(\mathbf{x})/2}}{\sqrt{\mathcal{Z}_R}}$, where 
    \begin{align*}
        \mathcal{Z} &= \int e^{-\beta V_R(\mathbf{x})} d \mathbf{x}\\
        \mathcal{Z}_R &= \int_R e^{-\beta V(\mathbf{x})} d \mathbf{x}
    \end{align*}
    are normalizing factors. We also write $\mathcal{Z} = \mathcal{Z}_R + \mathcal{Z}_{R^c}$. Trivially extend $\rho$ as,
    \begin{align*}
        \rho(\mathbf{x}) = \begin{cases}
            \rho(\mathbf{x}) & \mathbf{x} \in R\\
            0 & \text{ else }.
        \end{cases}
    \end{align*}

    And compute the $L_2$ norm of their difference
    \begin{align*}
        \norm{\rho_R  - \rho}^2_2 &= \int \left|\rho_R(\mathbf{x})  - \rho(\mathbf{x})\right|^2 d\mathbf{x}\\
        &=\int_{x\in R} \left|\frac{e^{-\beta V(\mathbf{x})/2}}{\sqrt{\mathcal{Z}}}  -\frac{e^{-\beta V(\mathbf{x})/2}}{\sqrt{\mathcal{Z}_R}}\right|^2 d\mathbf{x} + \int_{x \in R^c}\frac{e^{-\beta V_R(\mathbf{x})}}{\mathcal{Z}} d\mathbf{x}\\
        &=\left|\frac{1}{\sqrt{\mathcal{Z}}}  -\frac{1}{\sqrt{\mathcal{Z}_R}}\right|^2\int_{x\in R} e^{-\beta V(\mathbf{x})} d\mathbf{x} + \frac{1}{\mathcal{Z}}\int_{x \in R^c}e^{-\beta V_R(\mathbf{x})} d\mathbf{x}\\
        &=\left|\frac{1}{\sqrt{\mathcal{Z}}}  -\frac{1}{\sqrt{\mathcal{Z}_R}}\right|^2\mathcal{Z}_R+ \frac{\mathcal{Z}_{R^c}}{\mathcal{Z}}\\
        &=1-2\sqrt{\frac{\mathcal{Z}_R}{\mathcal{Z}}} + \frac{\mathcal{Z}_R}{\mathcal{Z}} + \frac{\mathcal{Z}_{R^c}}{\mathcal{Z}}\\
        &= 2\left(1-\sqrt{\frac{\mathcal{Z}_R}{\mathcal{Z}}} \right).
    \end{align*}
    Thus, in order to ensure that 
    \begin{align*}
        \norm{\rho_R  - \rho}^2_2  \leq \epsilon^2,
    \end{align*}
    we require 
    \begin{align*}
        \frac{\mathcal{Z}_R}{\mathcal{Z}_R+\mathcal{Z}_{R^c}} \geq \left(1 -\frac{\epsilon^2}{2}\right)^2.
    \end{align*}
    Expanding the left hand side around $\mathcal{Z}_{R^c} = 0$,  we have 
    \begin{align*}
        \frac{\mathcal{Z}_R}{\mathcal{Z}_R+\mathcal{Z}_{R^c}} = 1-\frac{\mathcal{Z}_{R^c}}{\mathcal{Z}} + O\left(\left(\frac{\mathcal{Z}_{R^c}}{\mathcal{Z}}\right)^2\right).
    \end{align*}

    Now, we obtain a lower bound on $\mathcal{Z}_{R^c}$. Assume that $\min_{\mathbf{x} \in \partial R} e^{-\beta V(\mathbf{x})} \geq C > 0$.
    \begin{align*}
        \mathcal{Z}_{R^c} &\geq C  \int_{R^c} e^{-\beta \kappa \text{ dist}(\mathbf{x},\partial R)^2}d \mathbf{x},
    \end{align*}
    since the point $\mathbf{x}_0$ defining the center of the spherical region $R$ entails an overall translation of the coordinates, and the above integral is invariant under this translation, we may take $\mathbf{x}_0 = 0$ without loss of generality. Therefore, 
    \begin{align*}
        \int_{R^c} e^{-\beta \kappa \text{ dist}(\mathbf{x},\partial R)^2}d \mathbf{x} = \omega_{\eta d}\int_{r_0}^{\infty} e^{-\beta \kappa(r-r_0)^2}r^{\eta d - 1}dr,
    \end{align*}
    and $\omega_n = \frac{2 \pi^{n/2}}{\Gamma(n/2)}$. Let $u = r-r_0$, then
    \begin{align*}
        \int_{r_0}^{\infty} e^{-\beta \kappa(r-r_0)^2}r^{\eta d - 1}dr&=\int_{0}^{\infty} e^{-\beta \kappa u^2}(u+r_0)^{\eta d - 1}du.
    \end{align*}
    The term $(u+r_0)^{\eta d - 1}$ can be expanded using the binomial theorem, and using standard Gaussian integrals we can obtain an exact expression for the integral
    \begin{equation}
        \int_{R^c} e^{-\beta \kappa \text{ dist}(\mathbf{x},\partial R)^2}d \mathbf{x} = \frac{\omega_{\eta d}}{2}\sum_{k=0}^{\eta d -1}\binom{\eta d -1}{k}r_0^{\eta d - 1 -k} (\beta \kappa)^{-(k+1)/2}\Gamma\left(\frac{k+1}{2}\right).
    \end{equation}
    For $\beta\kappa$ large and $r_0$ some constant, we have the following asymptotic bound on the integral in terms of the dimension-dependent constant $C_{\eta d}$,
    \begin{equation}
        \frac{\pi^{(\eta d + 1)/2}}{\Gamma(\eta d /2)} (\beta \kappa)^{-1/2}r_0^{\eta d-1} = C_{\eta d} (\kappa \beta)^{-1/2} \in O\left(\sqrt{\kappa \beta}\right).
    \end{equation}
    
    With this asymptotic bound $\mathcal{Z}_R$, it therefore suffices to choose 
    \begin{equation}
        \kappa \in O\left(\frac{1}{\mathcal{Z}_R^2 \beta \epsilon^4}\right),
    \end{equation}
    with constants depending on the behavior of the potential on $\partial R$ and the dimension $\eta d$.
\end{proof}

The $\kappa \sim O\left(\frac{1}{\epsilon^4}\right)$ scaling from Lemma \ref{lem:kappa-scaling-steady-state} entails a complexity of $\Omega(\epsilon^{-2})$ for quantum state preparation algorithms employing the sum-of-squares decomposition we use in this work, as the block encoding subnormalization factor for the potential will scale with this factor. Thus, the procedure based on augmenting the potential to obtain a confining potential is likely suboptimal. Nevertheless, it satisfies our purpose of demonstrating a method that prepares the initial reactant and product states with polynomial complexity in the locally convex setting.

Finally, it remains to establish the convexity of the augmented potential $V_R(x)$ given the original potential $V(x)$ restricted to $R$ is $m$-strongly convex. By the following Lemma, the augmented potential $V_R$ has Poincar\`e constant $m + 2\kappa$, and has a corresponding mixing time of $t_{\text{mix}} \sim \frac{\beta}{m+2\kappa}$.
\begin{lemma}
\label{lem:augmented-V-convexity}
    Suppose $V$ is $m$-strongly convex for all $x \in R$, for a convex set $R$. Then with $V_R$ defined as in Eq. \eqref{eq:V-augment}, we have that $\text{Hess}\left(V_R\right)\geq m I$ for any $\kappa \geq 1/2$.
\end{lemma}
\begin{proof}
    The assumption that $V$ is $m$-strongly convex on $R$ is equivalent to the statement that for all $x, y \in R$,
    \begin{align*}
        \left(\nabla V(\mathbf{x}) - \nabla V(\mathbf{y})\right)^T(\mathbf{x}-\mathbf{y}) \geq \frac{m}{2} \norm{\mathbf{x}-\mathbf{y}}^2.
    \end{align*}
    Clearly, for all $\mathbf{x}, \mathbf{y} \in R^c$, $V_R$ satisfies
    \begin{align}
       \left(\nabla V_R(\mathbf{x}) - \nabla V_R(\mathbf{y})\right)^T(\mathbf{x}-\mathbf{y}) \geq \left(\frac{m}{2} + \kappa\right) \norm{\mathbf{x}-\mathbf{y}}^2.
    \end{align}
    Thus, it remains to show that for $\mathbf{x} \in R$ and $\mathbf{y} \in R^c$ that 
    \begin{align*}
        \left(\nabla V(\mathbf{x}) - \nabla V(\mathbf{y})\right)^T(\mathbf{x}-\mathbf{y}) \geq \frac{m}{2} \norm{\mathbf{x}-\mathbf{y}}^2.
    \end{align*}
    Noting that,
    \begin{align*}
    \nabla V_R(\mathbf{x}) = \begin{cases}
        \nabla V(\mathbf{x}) &\mathbf{x} \in R\\
        \nabla V(\mathbf{x})+2 \kappa \left(\mathbf{x} - \partial R\right) & \mathbf{x} \in R^c.
    \end{cases}
    \end{align*}
    We have,
    \begin{align*}
        &\left(\nabla V(\mathbf{x}) - \nabla V(\mathbf{y}) - 2\kappa(\mathbf{y} -\partial R)\right)^T(\mathbf{x}-\mathbf{y})\\
        &=\left(\nabla V(\mathbf{x}) - \nabla V(\mathbf{y})\right)^T(\mathbf{x}-\mathbf{y}) - 2\kappa(\mathbf{y} -\partial R)^T(\mathbf{x}-\mathbf{y})\\
        &=\left(\nabla V(\mathbf{x}) - \nabla V(\mathbf{y})\right)^T(\mathbf{x}-\mathbf{y})  -2\kappa(\mathbf{y} -\mathbf{x})^T(\mathbf{x}-\mathbf{y}) -2\kappa (\mathbf{x} -\partial R)^T(\mathbf{x}-\mathbf{y})\\
        &\geq \frac{m}{2}\norm{\mathbf{x} - \mathbf{y}}^2 + 2\kappa \norm{\mathbf{x} - \mathbf{y}}^2 -2\kappa(\mathbf{x} - \partial R)^T(\mathbf{x}-\mathbf{y}).
    \end{align*}
    Then, since $\norm{\mathbf{x}-\mathbf{y}}\geq \norm{\mathbf{x} - \partial R}$, we have
    \begin{equation}
        \left(\nabla V(\mathbf{x}) - \nabla V(\mathbf{y}) - 2\kappa(\mathbf{y} -\partial R)\right)^T(\mathbf{x}-\mathbf{y}) \geq (m + 2\kappa - 1) \norm{\mathbf{x} - \mathbf{y}}^2,
    \end{equation}
    Thus for any $\kappa \geq 1/2$, we obtain the desired result.
\end{proof}

Under the given assumption that $V(x)$ is strongly convex on $R$, and the construction of Lemma \ref{lem:kappa-scaling-steady-state}, we can use our block encoding of the Witten Laplacian (see Appendix \ref{app-sec:BE-A}) with the augmented potential to efficiently prepare the equilibrium distribution over reactant and product regions that correspond to local minima of the potential energy surface. The state preparation can be performed by either performing direct dynamics under the Witten Laplacian, or by applying the eigenstate filtering algorithm of Ref. \cite{lengOperatorLevelQuantumAcceleration2025}.

\localCvxPrep*
\begin{proof}
    First, we employ  Lemma \ref{lem:kappa-scaling-steady-state}, to obtain a construction of the augmented confining potential $V_R$, where $\kappa\sim \frac{1}{\mathcal{Z}_R^2 \beta \epsilon^4}$. Then, using the technique described in Lemma 9 of Ref. \cite{lengOperatorLevelQuantumAcceleration2025}, we can choose an isometry $I_N$ that acts as an embedding from $\mathbb{C}^{N^{\eta d}}$ into $L^2$, thus translating error in $\ell^2[\mathbb{C}^{N^{\eta d}}]$ into error in the standard $L^2$ norm. Thus, with this construction and Lemmas \ref{lem:poincare} and \ref{lem:augmented-V-convexity}, then the dynamics generated by $\HB$ satisfy,
    \begin{equation}
        \norm{I_N e^{\HB t}\ket{R_0} - \ket{R}}_2 \leq e^{-2t\frac{m}{\beta}}\norm{I_N \ket{R_0} - \ket{R}}_{L^2} \leq e^{-2t\frac{m}{\beta}}C,
    \end{equation}
    for some constant $C \in O(1)$.
    Thus, for any $\epsilon > 0$, we can ensure
    \begin{align}
        \norm{I_N e^{\HB t}\ket{R_0} - \ket{R}}_2 \leq \epsilon,
    \end{align}
    by choosing 
    \begin{equation}
        t \geq \frac{\beta}{2m}O\left(\log\left(\frac{1}{\epsilon}\right)\right).
    \end{equation}
     The subnormalization factor associated with the block encoding of the matrix square root $\mathcal{A}$ is shown in Theorem \ref{thm:BE-Aj-ops} from Appendix \ref{app-sec:BE-A}, to be $\alpha_A =  O\left(\sqrt{\beta d} \eta^{3/2} L \alpha_V + \sqrt{\frac{\eta d}{\beta}} \frac{N}{L}\right)$. This, combined with a query complexity of $\widetilde{O}\left(\alpha_A \sqrt{t}\right) $ where $\sqrt{t} \in \widetilde{O}\left(\sqrt{\frac{\beta}{2m}}\right)$ leads to an overall complexity of
    \begin{align*}
        \widetilde{O}\left(\sqrt{\frac{\eta d}{2m}}\left(\beta \eta L (\alpha_V + \kappa) + \frac{N}{L}\right)\right),
    \end{align*}
    queries to the block encoding of $\mathcal{A}$ as desired.
\end{proof}

\section{Background on quantum algorithms}
\label{app-sec:q-alg-background}
\subsection{Block encoding}
\label{app-subsec:BE-defn}
Block encoding is a quantum subroutine that implements the action of a non-unitary matrix in a subspace of a unitary matrix. We will exclusively use the linear combination of unitaries (LCU) approach to constructing block encodings, which we describe here. Suppose we have a matrix operation $A$ where we have an expression for $A$ as 
\begin{equation}
    A = \sum_{l=0}^{M-1}c_l U_l
\end{equation}
where the $U_l$ are unitary matrices of the same dimension as $A$ and the $c_l \in \mathbb{C}$ are the coefficients in the expansion. We first define the $\textsc{prep}$ routine on an ancilla register with $m = \ceil{\log_2(M)}$ ancilla qubits,
\begin{equation}
    \textsc{prep}:\ket{0}_m\rightarrow \sum_{l=0}^{M-1}\sqrt{\frac{c_l}{\alpha}}\ket{l}_m
    \label{eq:prep}
\end{equation}
with $\alpha = \sum_{l=0}^{M-1}|c_l|$, the so-called \textit{subnormalization factor}. Then, we apply the unitaries $U_l$ in a controlled manner using the $\textsc{sel}$ routine
\begin{equation}
\textsc{sel} = \sum_{l=0}^{M-1}U_l\otimes \ket{l}\bra{l}.
\end{equation}
In a quantum circuit we may sometimes express the $\textsc{sel}$ operation using the $\oslash$ notation, which is to be interpreted as a controlled operation applied on all relevant bitstrings in the control register.
Combining these routines, we define the block encoding operation
\begin{equation}
    U_A = (I\otimes \textsc{prep}^\dagger) \textsc{sel} (I\otimes \textsc{prep}).
\end{equation}
We say that $U_A$ is an $(\alpha,m,\epsilon)$ block encoding of the matrix $A$, if it satisfies
\begin{equation}
\label{eq:defn-BE}
    \norm{A-\alpha(I\otimes \prescript{}{m}{\bra{0}})U_A (I\otimes \ket{0}_m)}\leq \epsilon.
\end{equation}

\subsection{Quantum Signal Processing}
\label{app-subsec:QSP}
Quantum signal processing (QSP) and its variants \cite{lowHamiltonianSimulationQubitization2019,lowOptimalHamiltonianSimulation2017,gilyenQuantumSingularValue2019, motlaghGeneralizedQuantumSignal2023} provides a systematic procedure for implementing a class of polynomial transformations to block encoded matrices. 
For Hermitian matrices $A$, the action of a scalar function $f$ can be determined by the eigendecomposition of $A = UDU^\dagger$, as
\begin{equation}
    f(A) := U \sum_{i}f(\lambda_i)\ket{i}\bra{i}U^\dagger,
    \label{eq:mat func}
\end{equation}
where $\lambda_i$ is the $i$th eigenvalue, and $U\ket{i}$ is the $i$th eigenvector. When $A$ is not Hermitian, a generalized procedure based on the singular value transform \cite{gilyenQuantumSingularValue2019} can be used to implement these functions. To implement a degree $d$ polynomial, QSP  uses $O(d)$ applications of the block encoding and therefore its efficiency depends closely on the rate {at which} a given polynomial approximation converges to the function of interest, {as well as on} the circuit complexity {required} to implement the block encoding. 

The following, originally from Ref. \cite{lowHamiltonianSimulationQubitization2019}, characterizes the complexity of performing Hamiltonian simulation with QSP in terms of the number of queries to the block encoding of the Hamiltonian matrix.
\begin{lemma}[Polynomial degree for Hamiltonian simulation (adapted from Ref.\cite{lowHamiltonianSimulationQubitization2019} Theorem 1) ]
\label{lem:qsp-ham-sim}
    Let $A$ be a Hermitian matrix accessed via a block encoding $U_A$ with subnormalization factor $\alpha \geq \norm{A}$. Then, for any $t$ we can construct a block encoding of $U_{e^{i A t}}$ using a polynomial of degree 
    \begin{equation}
        O\left(\alpha t + \log(1/\epsilon)\right).
        \label{app-eq:qsp-poly-degree}
    \end{equation}
\end{lemma}

{QSP uses} repeated applications of the block encoding circuit {to implement} a polynomial of the block encoded matrix. We assume that $A$ is a Hermitian matrix that has been suitably subnormalized by some factor $\lambda$ so that $\norm{A/\lambda}\leq 1$. QSP exploits the fact that the block encoding $U_A$ can be expressed as a direct sum over one and two-dimensional invariant subspaces which correspond directly to the eigensystem of $A$, without knowledge of the eigendecomposition. For each eigenvector $\ket{x}$ of $A$, there is a $2\times 2$ rotation matrix $O_x$ representing the $2d$ subspace associated with $\ket{x}$. In turn this allows one to \textit{obliviously} apply a class of polynomial functions to the eigenvalues in each subspace, which builds up the polynomial transformation of the block encoded matrix $A/\lambda$. QSP is characterized by the following theorem, which is a rephrasing of Theorem 4 of Ref. \cite{lowOptimalHamiltonianSimulation2017}.

\begin{lemma}[Quantum Signal Processing (Theorem 7.21 \cite{Lin2022})]
\label{lem:qsp}
    There exists a set of phase factors $\boldsymbol{\Phi} := (\phi_0, \ldots, \phi_{d}) \in \mathbb{R}^{d+1}$ such that
    \begin{equation}
    \begin{aligned}
        U_\phi(x) &= e^{i \phi_0 Z}\prod_{j=1}^d[O(x)e^{i\phi_j Z}]\\
        &=
        \begin{pmatrix}
            P(x) & - Q(x)\sqrt{1-x^2}\\
            Q^*(x)\sqrt{1-x^2} & P^*(x)
        \end{pmatrix},
    \end{aligned}
    \end{equation}
    where
    \begin{equation}
        O(x) = \begin{pmatrix}
            x & -\sqrt{1-x^2}\\
            \sqrt{1-x^2}& x
        \end{pmatrix},
    \end{equation}
    if and only if $P,Q \in \mathbb{C}[x]$ satisfy
    \begin{enumerate}
        \item $deg(P) \leq d, deg(Q) \leq d-1$,
        \item $P$ has parity $d \mod 2$ and $Q$ has parity $d-1 \mod 2$, and
        \item $|P(x)|^2 + (1-x^2)|Q(x)|^2 =1 \hspace{.2cm}\forall x \in[-1,1]$.
    \end{enumerate}
\end{lemma}

Given a scalar function with $\norm{f(x)}\leq 1$, we can use QSP to approximately implement $f(A/\lambda)$ by using the block encoding of $A$, a polynomial approximation to the desired scalar function, and a set of phase factors corresponding to the approximating polynomial. If $f$ is not of definite parity, we can use the technique of linear combination of block encodings to obtain a polynomial approximation to $f$ that is also of indefinite parity, or use generalized QSP from Ref. \cite{motlaghGeneralizedQuantumSignal2023}. 
 
{We note that} there are efficient algorithms for computing the phase factors {of QSP} for very high degree polynomials.  Algorithms for finding phase factors for polynomials of degree $d = O(10^7)$ have been reported in the literature (see e.g. \cite{motlaghGeneralizedQuantumSignal2023, dong_efficient_2021, Ying2022stablefactorization}) and are surprisingly numerically stable even to such high degree. Given that it is straightforward to obtain polynomial approximations to scalar functions and to obtain the corresponding phase factors for QSP, the block encoding circuitry combined with {the QSP framework immediately provides a nearly fully explicit circuit description of the block encodings in this work.}

\subsection{The original LCHS method}
\label{app-subsec:og-LCHS}
The original LCHS method  \cite{anLinearCombinationHamiltonian2023} describes the non-unitary propagator $e^{-t H}$, for $H\succeq 0$, in terms of the expectation value over a family of unitaries $e^{-i t kH}$, where the probability measure can be obtained in terms of the normalized Fourier coefficients of the exponential function $e^{-t |x|}$. This is described by the following relationship 
\begin{equation}
    e^{-t |x|} =  \frac{1}{\pi}\int_{-\infty}^{\infty} dk \frac{e^{ it k |x|}}{1+k^2},
    \label{eq:LCHS-cauchy}
\end{equation}
which is a Cauchy distribution over the Fourier modes. Although the LCHS method can accommodate both time-dependent and non-Hermitian matrices, our work is restricted to the case of approximating the propagator for a time-independent Hermitian matrix, which greatly simplifies its usage. Analogously, the LCHS method represents the dissipative, non-unitary propagator $e^{-t H}$ as a summation over the set of unitaries $\{e^{i t H k}\}$, over the variable $k$,
\begin{equation}
    e^{-t H} = \frac{1}{\pi}\int_{-\infty}^{\infty} dk \frac{e^{-it k H}}{1+k^2}.
\end{equation}
One then truncates the infinite interval to a finite interval $[-K,K]$, choosing $K = O(1/\epsilon)$ and approximating the truncated integral using a quadrature scheme to obtain an $\epsilon$ accurate approximation \ref{eq:LCHS-cauchy} in the spectral norm.

The follow-up work of Ref. \cite{anQuantumAlgorithmLinear2023a} showed that there exists a generalized family of weight functions which can be used instead of the Cauchy distribution, some of which exhibit rapid decay properties away from the origin. In turn, this leads to an improved bound with respect to the cutoff parameter $K$. For a function $f$ satisfying the conditions of Theorem 5 of \cite{anQuantumAlgorithmLinear2023a}, the LCHS formula can be expressed in the more general form
\begin{equation}
    e^{- t H} = \int_{-\infty}^{\infty}dk\frac{f(k)}{1-ik}e^{-itkH},
    \label{eq:LCHS-general}
\end{equation}
where a near-optimal choice of $f$ was shown to be
\begin{equation}
    f_\sigma (k) = \frac{1}{2 \pi e^{-2^\sigma}e^{(1+ik)^{\sigma}}}
    \label{eq:near-opt-kernel}
\end{equation}
for a parameter $\sigma$ satisfying $0 < \sigma < 1$. As discussed in Appendix C of the same work, a practical choice is $\sigma \in [.7,.8]$. 

The integral \eqref{eq:LCHS-general} can be approximated using numerical quadrature. First the range of integration is truncated to $[-K,K]$, where $K$ is chosen to be $O(\log^{1/\sigma}(1/\epsilon))$ when using the near-optimal choice of kernel function \eqref{eq:near-opt-kernel}. Then one employs a composite Gaussian quadrature scheme, where the large interval $[-K,K]$ is broken-up into a sequence of small intervals $[mh,(m+1)h]$, with $m \in \{-K/h, \ldots, K/h -1\}$, then we perform Gaussian quadrature with $Q$ Gauss points on each sub-interval. Ultimately, this provides an expression for the non-unitary propagator in terms of an LCU
\begin{equation}
    e^{-t H} = \sum_{m=-K/h}^{K/h-1}\sum_{q=0}^{Q-1}g_q w(k_{q,m})e^{-i t k_{q,m}H},
    \label{eq:LCHS-quad}
\end{equation}
where $k_{q,m}$ is the $q$th quadrature point on the $m$th sub-interval, $w$ the weight function, and $g_q$ is quadrature weight on the $q$th point for Gaussian quadrature on $Q$ points. Lemmas 9 and 10 of \cite{anQuantumAlgorithmLinear2023a} bound the size of each sub-interval $h = \frac{1}{e t \norm{H}}$ and the number of quadrature points $Q = O(\log(1/\epsilon))$. Therefore, the total number of unitaries in the LCU \eqref{eq:LCHS-quad} is 
\begin{equation}
    M = \frac{2 K Q}{h} = O\left(t \norm{H} \log^{1+1/\sigma}(1/\epsilon)\right).
\end{equation}

\section{Error Bounds on Gaussian-LCHS}
\label{app-sec:gauss-lchs-error}
We now perform the error estimates on the Gaussian-LCHS formula introduced from domain truncation and quadrature. We analyze error introduced at nearly every stage of the algorithm, sans errors in preparation of any coefficient states that are used to prepare LCUs. We first analyze the error from truncating the domain of integration in the Gaussian-LCHS formula in Sec. \ref{app-subsec:trunc-error-bound}. We then bound the error introduced by applying the quadrature scheme on the truncated integration domain in Sec. \ref{app-subsec:quad-error-bound}. Then we bound the error introduced by approximating the unitaries in the LCHS to finite precision in Sec. \ref{app-subsec:sim-error-bounds}.

\subsection{Truncation error bounds}
\label{app-subsec:trunc-error-bound}
The error from truncating the domain of integration is characterized by the
following lemma.

\begin{lemma}[Truncation Error]
\label{lem:trunc_error}
For any $\epsilon > 0$ there exists a $K > 2\sqrt{t}$ such that
\begin{equation*}
\left|\int_{-\infty}^{\infty}\frac{e^{-k^2/4t}}{2\sqrt{\pi t}}e^{ikx}\,dk
- \int_{-K}^{K}\frac{e^{-k^2/4t}}{2\sqrt{\pi t}}e^{ikx}\,dk\right| \leq \epsilon,
\end{equation*}
and the choice
\begin{equation}
K \;=\; 2\sqrt{t\,\log\!\left(\tfrac{1}{\epsilon\sqrt{\pi}}\right)}
\;\in\; O\!\left(\sqrt{t\,\log(1/\epsilon)}\right)
\label{app-eq:K-bound}
\end{equation}
suffices.
\end{lemma}

\begin{proof}
Let
\begin{equation*}
E_K \;=\; \left|\int_{-\infty}^{\infty}\frac{e^{-k^2/4t}}{2\sqrt{\pi t}}e^{ikx}\,dk
- \int_{-K}^{K}\frac{e^{-k^2/4t}}{2\sqrt{\pi t}}e^{ikx}\,dk\right|.
\end{equation*}
The truncation tail satisfies
\begin{align*}
E_K &= \left|\int_{K}^{\infty}\frac{e^{-k^2/4t}}{2\sqrt{\pi t}}\left(e^{-ikx}+e^{ikx}\right)dk\right|
\;\leq\; \frac{1}{\sqrt{\pi t}}\int_{K}^{\infty} e^{-k^2/4t}\,dk
\;=\; \operatorname{erfc}\!\left(\tfrac{K}{2\sqrt{t}}\right).
\end{align*}
For $x > 0$ the complementary error function obeys the sharp upper bound
\begin{equation*}
\operatorname{erfc}(x) \;\leq\; \frac{2\,e^{-x^2}}{\sqrt{\pi}\!\left(x+\sqrt{x^2+\tfrac{4}{\pi}}\right)}.
\end{equation*}

 With $x = K/(2\sqrt{t})$ and using $\sqrt{x^2+\tfrac{4}{\pi}}\geq x$ to
bound the denominator below by $2x$,
\begin{equation*}
\operatorname{erfc}(x) \;\leq\; \frac{e^{-x^2}}{x\sqrt{\pi}},
\end{equation*}
which, under the lemma's assumption $K \geq 2\sqrt{t}$ (i.e. $x\geq 1$), implies
the further simplification
\begin{equation*}
E_K \;\leq\; \operatorname{erfc}(x) \;\leq\; \frac{e^{-x^2}}{\sqrt{\pi}}
\;=\; \frac{e^{-K^2/4t}}{\sqrt{\pi}}.
\end{equation*}
The condition $E_K\leq\epsilon$ is therefore implied by
$e^{-K^2/4t}\leq \epsilon\sqrt{\pi}$, equivalently
\begin{equation*}
K \;\geq\; 2\sqrt{t\,\log\!\left(\tfrac{1}{\epsilon\sqrt{\pi}}\right)}.
\end{equation*}
For all $\epsilon < 1/\sqrt{\pi}$ this choice also satisfies $K\geq 2\sqrt{t}$,
closing the assumption. 
\end{proof}

\subsection{Quadrature Error Bounds}
\label{app-subsec:quad-error-bound}
Now, we would like to bound the size of the error introduced from the quadrature scheme on $M$ grid points. This is characterized by the following lemma. 
\begin{lemma}[Quadrature Error]
\label{lem:quad_error}
    For any $\epsilon>0$, there exists a positive integer $M \geq 2$, a set of weights $w_j \in \{0, \ldots, M\}$, and quadrature points $k(j) \subset [-K,K]$ for $j \in \{0,\ldots,M\}$ such that 
    \begin{equation*}
        \left|\sum_{j=0}^{M}w_j f_t\left(k(j)\right) - \int_{-K}^{K}f_t(k)e^{i kx}dk\right| \leq \epsilon,
    \end{equation*}
    where
    \begin{equation*}
    f_t(k) = \frac{e^{-k^2/4t}}{{2\sqrt{\pi t}}}e^{ikx}.
\end{equation*}
    Furthermore, for a block encoding with subnormalization factor $\alpha$ of the square root matrix, $M$ can be chosen 
    \begin{equation*}
         M \in O\left(\alpha \sqrt{t \log(1/\epsilon)}\right).
    \end{equation*}
\end{lemma}

\begin{proof}
By Theorem 19.3 of Ref. \cite{trefethen2013approximation}, for a number of quadrature points $M\geq 2$, we have the error estimates for the Gauss quadrature
\begin{equation}
    \left|I_M - \int_{-K}^{K}\frac{e^{-k^2/4t}}{{2\sqrt{\pi t}}}e^{-i kx}dk\right| \leq \frac{144}{35}\frac{ \Gamma \rho^{-2M}}{\rho^2-1},
\end{equation}
where $\rho$ controls the size of the Bernstein ellipse and $\Gamma$ is the maximum that the integrand attains on the ellipse. With 
\begin{equation}
    g(z) = \frac{K}{2}\left(z+z^{-1}\right),
\end{equation}
parameterizing the Bernstein ellipse over $[-K,K]$ with $z\in\mathbb{C}$, we have that 
\begin{equation}
    \frac{\partial }{\partial \overline{z}}e^{-g^2(z)/4t}e^{-i g(z)x} = 0,
\end{equation}
so the integrand satisfies the Cauchy-Riemann equations and is analytic for any $\rho$. Therefore, we can choose $\rho$ arbitrarily large, but this also causes $\Gamma$ to grow correspondingly. For some fixed $\rho > 1$ we have
\begin{equation}
    \Gamma = \frac{1}{2\sqrt{\pi t}}\max_{\theta, x\in [-\alpha,\alpha]} e^{-g^2(\rho e^{i\theta })/4t-i g(\rho e^{i\theta})x},
\end{equation}
which is maximized when $-\frac{g^2(\rho e^{i\theta })}{4t}-i g(\rho e^{i\theta})x$ is maximized. Furthermore, we only need to maximize the real part of the expression.

It can be shown that 
\begin{equation}
    \text{Re}\left(-\frac{g^2(\rho e^{i\theta })}{4t}-i g(\rho e^{i\theta})x\right) = \frac{-K\left(K \text{Re}\left(e^{-2i\theta}(1+e^{i\theta}\rho)^4\right) +16 tx \rho (\rho^2-1)\sin(\theta)\right)}{64 t \rho ^2},
\end{equation}
which achieves its maximum at $\theta = -\pi/2$. At $\theta = -\pi/2$, we have
\begin{equation}
    \frac{K( 16 tx\rho (\rho^2-1) + K(1-6 \rho^2 + \rho^4))}{64 t \rho^2}.
\end{equation}
Now, we can see that for $1 <\rho \leq 1 +\sqrt{2}$ that 
\begin{equation}
    \delta \rho(\rho^2-1) +  1- 6 \rho^2 + \rho^4 \leq \delta \rho(\rho^2-1)
\end{equation}
for any $\delta > 0$, as $\rho = 1+\sqrt{2}$ is a zero of $(1-6 \rho^2 + \rho^4)$ and is strictly negative on $[1,1+\sqrt{2})$. Therefore, with the choice of $\rho = 1+\sqrt{2}$, we have that the exponent is upper bounded by 
\begin{equation}
\frac{K x }{2 },
\end{equation}
or that 
\begin{equation}
    \Gamma \geq  \frac{e^{K \alpha /2}}{2 \sqrt{\pi t}}.
\end{equation}

Then, for any $\epsilon>0$, we choose $M$ such that 
\begin{equation}
    \frac{144 \Gamma \rho^{-2M}}{35 (\rho^2-1)}\leq \epsilon,
\end{equation}
where it can be shown that it suffices to choose
\begin{equation}
    M \geq \frac{1}{2\log(1+\sqrt{2})}\left(\frac{K \alpha}{2} + \log\left(\frac{10}{\epsilon\sqrt{t}}\right)\right),
\end{equation}
since $K \sim \sqrt{t \log(1/\epsilon_K)}$, and $\epsilon_K \sim \epsilon$, we have the estimate
\begin{equation}
    M \in O\left(\alpha \sqrt{t \log(1/\epsilon)} \right),
    \label{app-eq:M-bound}
\end{equation}
as desired.

\end{proof}

\subsection{Bound on LCHS subnormalization factor}
\label{app-subsec:bound-alpha-g}
The following lemma characterizes the subnormalization factor from implementing the Gaussian-LCHS using Gauss quadrature  as a linear combination of unitaries. 
\begin{lemma}[Bound on subnormalization factor]
\label{lem:gaus-lchs-subnorm}
The subnormalization factor  $\alpha_g$ resulting from applying the Gausssian-LCHS formula using Gauss quadrature satisfies
\begin{equation*}
    \alpha_g \in O(1).
\end{equation*}
\end{lemma}
\begin{proof}
We have
\begin{equation}
    I_M = \sum_{j=0}^{M}w_j f_t\left(k(j)\right),
\end{equation}
where $k(j)$ is a Gaussian quadrature point and
\begin{equation}
\begin{aligned}
    f_t(k) &= \frac{e^{-k^2/4t}}{2\sqrt{\pi t}} e^{ik x}.
\end{aligned}
\end{equation}
Then the subnormalization factor 
\begin{align*}
    \alpha_g &= \sum_{j=0}^{M}\left|w_j\frac{e^{-k(j)^2/4t}}{2\sqrt{\pi t}}\right|\\
    &=\frac{1}{2\sqrt{\pi t}}\sum_{j=0}^{M}\left|w_je^{-k(j)^2/4t}\right|\\
    &=\frac{1}{2\sqrt{\pi t}}\sum_{j=0}^{M}w_je^{-k(j)^2/4t}\\
    &\leq \frac{1}{2\sqrt{\pi t}} \int_{-\infty}^{\infty}e^{-k^2}{4t} + O(\epsilon)\\
    &= 1 + O(\epsilon)\\
    &\in O(1).
\end{align*}
The third equality results from the fact that the Gauss quadrature weights satisfy $w_j>0$, and the last equality completes the lemma.
\end{proof}

\subsection{Simulation error bounds}
\label{app-subsec:sim-error-bounds}
Now we will bound the effect of implementing the unitary terms to finite error in the Gaussian-LCHS formula.
\begin{lemma}[Bound on simulation error and polynomial degree]
\label{lem:simul-error-poly}
    Let $\mathcal{A}$ be provided as a block encoding with subnormalization factor $\alpha$ and let $K$ and $M$ be as in Lemmas \ref{lem:trunc_error} and \ref{lem:quad_error} respectively. Let $\widetilde{e^{ik(j)\mathcal{A}}}$ be an approximate block encoding of the true unitary evolution $e^{ik(j)\mathcal{A}}$ obtained via polynomial approximation.
    For any $\epsilon>0$, there exists $K > 1$, $M \geq 2$ and $\{\epsilon_j>0\}_{j=0}^M$ satisfying
    \begin{equation*}
        \norm{e^{ik(j)\mathcal{A}} - \widetilde{e^{ik(j)\mathcal{A}}}} \leq \epsilon_j,
    \end{equation*}
    such that 
    \begin{equation}
        \norm{e^{-t \mathcal{A}^2}-\sum_{j=0}^{M}w_jf_t(k(j))\widetilde{e^{ik(j)\mathcal{A}}} } \leq \epsilon.
        \label{app-eq:local-error}
    \end{equation}
    Moreover, there exists a polynomial approximation to $e^{i k(j)\mathcal{A}}$ of degree
    \begin{equation}
        O\left(\alpha\sqrt{t\log\left(\frac{1}{\epsilon}\right)} +  \log\left(\frac{\alpha}{\epsilon}\right)\right)
        \label{app-eq:total-error}
    \end{equation}
    such that the inequalities Eqs. \eqref{app-eq:local-error} and \eqref{app-eq:total-error} hold.
\end{lemma}
\begin{proof}

Let 
\begin{equation}
    U_j = e^{i k(j) \mathcal{A}}.
\end{equation}
When each $U_j$ is implemented with zero error, choosing $K$ and $M$ according to the bounds in \ref{app-subsec:trunc-error-bound} and \ref{app-subsec:quad-error-bound} with error $\epsilon/4$, in the spectral norm we have
\begin{equation}
\begin{aligned}
    \norm{e^{\HB t} - \sum_{j=0}^{M} f_t(k(j)) e^{ik(j)\mathcal{A}} } &\leq \norm{e^{\HB t} - \int_{-K}^K f_t(k) e^{i k(j) \mathcal{A}}dk} + \norm{\int_{-K}^K f_t(k) e^{i k(j) \mathcal{A}}dk-\sum_{j=0}^{M} f_t(k(j)) e^{ik(j)\mathcal{A}}}\\
    &\leq \frac{\epsilon}{4} + \frac{\epsilon}{4}\\
    &= \frac{\epsilon}{2}.
\end{aligned}
\end{equation}

Now suppose each $U_j$ is approximated with a unitary $\widetilde{U}_j$ to error $\epsilon_j > 0$ in the spectral norm
\begin{equation}
    \norm{U_j - \widetilde{U}_j} \leq \epsilon_j.
\end{equation}
Then,
\begin{align*}
    E & =\norm{\sum_{j=0}^{M}w_jf_t(k_j) (U_j - \widetilde{U}_j)}\\
    &\leq \sum_{j=0}^{M}w_jf_t(k_j)\norm{U_j - \widetilde{U}_j}\\
    &\leq \sum_{j=0}^{M}w_jf_t(k_j)\epsilon_j.
\end{align*}
If we desire, $E < \epsilon'$ for some $0 < \epsilon' \leq \epsilon/2$, it is sufficient to choose each $\epsilon_j \leq \frac{\epsilon'}{(M+1) w_j f_t(k_j)}$, as
\begin{align*}
    E & \leq \sum_{j=0}^{M}w_jf_t(k(j))\epsilon_j\\
    & \leq \sum_{j=0}^{M}w_jf_t(k(j))\frac{\epsilon'}{(M+1)w_jf_t(k(j))}\\
    & = \sum_{j=0}^{M}\frac{\epsilon'}{(M+1)}\\
    & = \epsilon'.
\end{align*}
With $K$ and $M$ chosen as above and $\epsilon' < \epsilon/2$, we see that
\begin{equation}
\begin{aligned}
    &\norm{e^{-t\mathcal{A}^2} - \sum_{j=0}^{M}w_jf_t(k(j))e^{i k(j) \mathcal{A}}}\\
    &\leq \frac{\epsilon}{2} + \norm{\sum_{j=0}^{M}w_jf_t(k_j) (U_j - \widetilde{U}_j)}\\
    &\leq \epsilon
\end{aligned}
\end{equation}
and find that the error in the Gaussian-LCHS formula with finite precision implementation of the $U_j$ is bounded by $\epsilon$.

Now, for any $\epsilon_j$ we have by Lemma \ref{lem:qsp-ham-sim} that $e^{i k(j) \mathcal{A}}$  can be approximated by a polynomial of degree 
\begin{equation}
    O\left(\alpha k(j) + \log\left(\frac{1}{\epsilon_j}\right)\right).
\end{equation}
Now, substitute $\epsilon_j \sim \frac{\epsilon}{M w_j f_t(k(j))}$ into the above logarithm,
\begin{align*}
    \log\left(\frac{1}{\epsilon_j}\right)&= \log\left(\frac{M w_j f_t(k(j))}{\epsilon'}\right)\\
    &=\log\left(\frac{M w_j e^{-k^2(j)/4t}}{\sqrt{t}\epsilon}\right).
\end{align*}
From Eq. \eqref{app-eq:M-bound}, we have
\begin{equation}
    M \in O\left(\alpha \sqrt{t \log(1/\epsilon)}\right).
\end{equation}
Therefore,
\begin{align*}
    \log\left(\frac{M w_j e^{-k^2(j)/4t}}{\sqrt{t}\epsilon'}\right)&\sim\log\left(\frac{\alpha \sqrt{t \log(1/\epsilon)} w_j e^{-k^2(j)/4t}}{\sqrt{t}\epsilon}\right)\\
    &=\log\left(\frac{\alpha \sqrt{\log(1/\epsilon)} w_j e^{-k^2(j)/4t}}{\epsilon}\right)\\
    &\sim\log\left(\frac{\alpha  w_j e^{-k^2(j)/4t}}{\epsilon}\right)\\
    &= \log\left(\frac{\alpha w_j }{\epsilon}\right) - \frac{k^2(j)}{4t}\\
    &\leq \log\left(\frac{\alpha w_j }{\epsilon}\right)\\
    &\in O\left(\log\left(\frac{\alpha}{\epsilon}\right)\right).
\end{align*}
Then, with $k(j) \leq K \in O\left(\sqrt{t\log\left(\frac{1}{\epsilon}\right)}\right)$ per Lemma \ref{lem:trunc_error}, the polynomial degree $n$ can be chosen
\begin{equation}
\begin{aligned}
    n &\in O\left(\alpha\sqrt{t\log\left(\frac{1}{\epsilon}\right)} +  \log\left(\frac{\alpha}{\epsilon}\right)\right)
    \label{app-eq:sim-error-bounds}
\end{aligned}
\end{equation}
\end{proof}

Lemmas \ref{lem:trunc_error}, \ref{lem:quad_error}, and \ref{lem:simul-error-poly} prove the following theorem.
\thmGaussLCHS*

\section{Multiplexed QSP (Proof of Lemma \ref{lem:mpx-qsp})}
\label{app-sec:pf-mpx-qsp}
\mpxqsp*
\begin{proof}
    By QSP (Lemma \ref{lem:qsp}) there exists a set of phases $\boldsymbol{\Phi}^j \in \mathbb{R}^{d_j+1}$ for each $P_j$ such that 
    \begin{equation}
         e^{i {\Phi}^j_0 Z_\Pi}\prod_{k=1}^{d_j}\left[O_Ae^{i{\Phi}^j_k Z_\Pi}\right] = U_{P_j(A/\alpha)},
    \end{equation}
    where $Z_\Pi = 2\ket{0}\bra{0}_a - I$ and $O_A = U_A Z_\Pi$. $Z_\Pi$ can be implemented with an additional ancilla qubit controlled on the ancilla register of the block encoding (see e.g. Fig. 7.7 of Ref. \cite{Lin2022}).  Note that since each $d_j$ are of the same parity, we have that $d_{\max} - d_j \in 2\mathbb{N}$. Additionally, notice that 
    \begin{equation}
        O_A^{-1} = O_A^\dagger = e^{-i \pi/2 Z_\Pi}O_A e^{i\pi/2 Z_\Pi},
    \end{equation}
    thus, we may extend the phases $\boldsymbol{\Phi}^j \in \mathbb{R}^{d_j+1} \hookrightarrow \mathbb{R}^{d_{\max}+1}$ by setting $\Phi^{j}_k = (-1)^k \frac{\pi}{2}$ for all $k \in \{d_j+1, \ldots, d_{\max}+1\}$. With this choice of phases we have, 
    \begin{align*}
        &e^{i {\Phi}^j_0 Z_\Pi}\prod_{k=1}^{d_{\max}}\left[O_Ae^{i{\Phi}^j_k Z_\Pi}\right]\\
        &= e^{i {\Phi}^j_0 Z_\Pi}\prod_{k=1}^{d_j}\left[O_Ae^{i{\Phi}^j_k Z_\Pi}\right] \prod_{k=d_{j}+1}^{d_{\max}} \left[O_Ae^{i{\Phi}^j_k Z_\Pi}\right]\\
        &= U_{P_j(A)} \prod_{k=1}^{d_{\max}-d_j} \left[O_Ae^{i(-1)^{k} \frac{\pi}{2} Z_\Pi}\right],
    \end{align*}
    Notice that since $d_{\max} - d_j \in 2\mathbb{Z},$ write $d_{max} - d_j = 2d$
    \begin{align*}
         &\prod_{k=1}^{2d} \left[O_Ae^{i(-1)^{k} \frac{\pi}{2} Z_\Pi}\right]\\
         &=\left(O_Ae^{-i \frac{\pi}{2} Z_\Pi}O_Ae^{i(-1)^{k} \frac{\pi}{2} Z_\Pi}\right)^{d}\\
         &=\left(e^{i \frac{\pi}{2} Z_\Pi}O_A^\dagger O_Ae^{i(-1)^{k} \frac{\pi}{2} Z_\Pi}\right)^{d}\\
         &=I.
    \end{align*}

    Thus, the construction of the quantum state can be achieved as follows. Introduce an ancilla register of $m = \ceil{\log(M)}$ qubits and prepare the $m$ qubit quantum state $\ket{c}$. Additionally, introduce the controlled operation 
    \begin{equation}
        c\mathcal{S}^j_i: \ket{0}_{a+1}\ket{j} \rightarrow e^{i \Phi^j_i}\ket{0}_{a+1}\ket{j}.
    \end{equation}
    Then, 
    \begin{align*}
    &c\mathcal{S}^j_i\prod_{k=1}^{d_{\max}}\left[O_Ac\mathcal{S}^j_i\right]\ket{\psi}\ket{0}_{a+1}\ket{c}\\
    &=\sum_{j}c_jc\mathcal{S}^j_i\prod_{k=1}^{d_{\max}}\left[O_Ac\mathcal{S}^j_i\right]\ket{\psi}\ket{0}_{a+1}\ket{j}\\
    &=\sum_{j}c_jc\mathcal{S}^j_i\prod_{k=1}^{d_{\max}}\left[O_A\ket{\psi}\ket{0}_{a+1} e^{i\Phi^j_i}\right]\ket{j}\\
    &=\sum_{j}c_j U_{P_j(A)}\ket{\psi}\ket{0}_{a+1}\ket{j}\\
    &= \sum_{j}c_j P_j(A)\ket{\psi}\ket{0}_{a+1}\ket{j} + \ket{\perp}.
    \end{align*}
    A circuit diagram describing these operations is provided as Fig. \ref{fig:mpx-QSP}.
    The controlled operations require $O\left(m\right)$ Toffoli gates to implement, which are queried $O(M)$ times. The state preparation of $\ket{c}$ also requires no more than $\widetilde{O}(M)$ Toffoli or simpler gates, thus achieving the stated complexities.
\end{proof}

\section{Non-unitary overlap circuit}
\label{app-sec:non-unitary overlap circuit}
\begin{figure}[ht]
    \centering
    \begin{quantikz}
        \lstick{$\ket{0}$}&\gate{H}   &&\ctrl{0}\vqw{2}&\octrl{0}\vqw{1}&\ctrl{0}\vqw{1}      &\octrl{0}\vqw{1}&&\gate{H}&\meter{}\\
        \lstick{$\ket{0}$}&\qwbundle{n}&&               &\gate{O_\psi}   &\gate{O_\phi}        &\gate[2]{\textsc{sel}_{lc}}&&&\\
        \lstick{$\ket{0}$}&\qwbundle{k}&\gate{\textsc{prep}_{lc}}       &\gate{\textsc{phase}}&        &&&&
    \end{quantikz}
    \caption{Quantum circuit for estimating overlaps of the form $\bra{\phi}A\ket{\psi}$ for non-unitary matrices $A$ using the oracles for the block encoding of $A$.  This circuit also neglects possible ancilla qubits used in the construction of $\textsc{sel}_{lc}$, i.e. by using a block encoding and quantum signal processing to construct the Hamiltonian simulation unitaries. }
    \label{app-fig:overlap-circuit}
\end{figure}

We now show how the circuit in non-unitary overlap circuit (reproduced as fig. \ref{app-fig:overlap-circuit} above for convenience) can be used to estimate these overlaps. Given some matrix $H$ that we have written as an LCU $H = \sum_{l}|c_l|e^{i\phi_l} U_l$ with subnormalization factor $\alpha = \sum_l |c_l|$. We show by direct computation how the above circuit can be used to approximate $\Re{\bra{\phi}H\ket{\psi}}$. The imaginary part can also be obtained using the standard method by applying an additional $S$ gate on the ancilla qubit prior to the application of the final Hadamard gate. First, observe that the circuit performs the following sequence of transformations 
\begin{align*}
    \ket{0}\ket{0}_n\ket{0}_k &\xrightarrow{H\otimes I_n \otimes I_k} \frac{1}{\sqrt{2}}\left(\ket{0}\ket{0}_n\ket{0}_k + \ket{1}\ket{0}_n\ket{0}_k\right)\\
    &\xrightarrow{\textsc{prep}_{lc}\text{c-}\textsc{phase}} \frac{1}{\sqrt{2}}\left(\ket{0}\ket{0}_n\sum_{l}\sqrt{\frac{|c_l|}{\alpha}}\ket{l}_k + \ket{1}\ket{0}_n\sum_{l}e^{-i\phi_l}\sqrt{\frac{|c_l|}{\alpha}}\ket{l}_k\right)\\
    &\xrightarrow{\text{c-}\psi, \text{c-}\phi} \frac{1}{\sqrt{2}}\left(\ket{0}\ket{\psi}_n\sum_{l}\sqrt{\frac{|c_l|}{\alpha}}\ket{l}_k + \ket{1}\ket{\phi}_n\sum_{l}e^{-i \phi_l}\sqrt{\frac{|c_l|}{\alpha}}\ket{l}_k\right)\\
    &\xrightarrow{c-\textsc{sel}_{lc}} \frac{1}{\sqrt{2}}\left(\sum_{l}\sqrt{\frac{|c_l|}{\alpha}}\ket{0}U_l\ket{\psi}_n\ket{l}_k + \sum_{l}e^{-i \phi_l}\sqrt{\frac{|c_l|}{\alpha}}\ket{1}\ket{\phi}_n\ket{l}_k\right)\\
    &\xrightarrow{H\otimes I_n\otimes I_k} \frac{1}{2}\left(\sum_{l}\sqrt{\frac{|c_l|}{\alpha}}\ket{0}\left(U_l\ket{\psi}_n + e^{-i\phi_l}\ket{\phi}_n\right)\ket{l}_k + \sum_{l}\sqrt{\frac{|c_l|}{\alpha}}\ket{1}\left(U_l\ket{\psi}_n - e^{-i\phi_l}\ket{\phi}_n\right)\ket{l}_k\right).
\end{align*}

The probability of measuring the ancilla qubit at the top of the circuit in $\ket{0}$, $P(0)$, is 
\begin{equation}
    P(0) = \frac{1}{4}\norm{\sum_{l}\sqrt{\frac{|c_l|}{\alpha}}\ket{0}\left(U_l\ket{\psi}_n + e^{-i\phi_l}\ket{\phi}_n\right)\ket{l}_k}^2.
    \label{eq:P0 swap}
\end{equation}
Where 
\begin{equation}
\begin{aligned}
    &\norm{\sum_{l}\sqrt{\frac{|c_l|}{\alpha}}\ket{0}\left(U_l\ket{\psi}_n + e^{-i\phi_l}\ket{\phi}_n\right)\ket{l}_k}^2\\ &= \sum_{l'}\sqrt{\frac{|c_{l'}|}{\alpha}}\bra{0}\left(\prescript{}{n}{\bra{\psi}} U_{l'}^\dagger + \prescript{}{n}{\bra{\phi}}e^{i\phi_{l'}}\right)\prescript{}{k}{\bra{l'}}\sum_{l}\sqrt{\frac{|c_l|}{\alpha}}\ket{0}\left(U_l\ket{\psi}_n + e^{-i\phi_l}\ket{\phi}_n\right)\ket{l}_k \\
    &=\sum_{l}\frac{|c_{l}|}{\alpha}\left(\prescript{}{n}{\bra{\psi}} U_{l}^\dagger + \prescript{}{n}{\bra{\phi}}e^{i\phi_{l}}\right)\left(U_l\ket{\psi}_n + e^{-i\phi_l}\ket{\phi}_n\right)\\
    &=\sum_{l}\frac{|c_l|}{\alpha}\left(e^{-i\phi_l} \bra{\psi}U_l^\dagger \ket{\phi} + e^{i\phi_l}\bra{\phi} U_l \ket{\psi} + 2\right)\\
    &=\sum_{l} \left( \frac{|c_l|}{\alpha}e^{-i\phi_l} \bra{\psi}U_l^\dagger \ket{\phi} + \frac{|c_l|}{\alpha}e^{i\phi_l}\bra{\phi} U_l \ket{\psi} \right) + 2,
\end{aligned}
\end{equation}
where in the second to last line above we used that $\sum_{l}\frac{|c_l|}{\alpha}=1$. Now,
\begin{equation}
\begin{aligned}
    &\sum_{l} \left( \frac{|c_l|}{\alpha}e^{-i\phi_l} \bra{\psi}U_l^\dagger \ket{\phi} + \frac{|c_l|}{\alpha}e^{i\phi_l}\bra{\phi} U_l \ket{\psi} \right) \\
    &= \frac{1}{\alpha}\bra{\psi} \sum_l c_l^* U_l ^\dagger \ket{\phi}+\frac{1}{\alpha}\bra{\phi}\sum_l c_l U_l \ket{\psi}  \\
    &=  \bra{\phi}\frac{H}{\alpha}\ket{\psi} + \bra{\psi}\frac{H^\dagger}{\alpha} \ket{\phi}\\
    &=  \bra{\phi}\frac{H}{\alpha}\ket{\psi} + \left(\bra{\phi}\frac{H}{\alpha}\ket{\psi}\right)^*\\
    &= \frac{2}{\alpha}\Re{\bra{\phi}H\ket{\psi}}.
\end{aligned}
\end{equation}
So that overall \eqref{eq:P0 swap} simplifies to
\begin{equation}
    P(0) = \frac{1}{2}\left(1+\frac{1}{\alpha}\Re{\bra{\phi}H\ket{\psi}}\right).
\end{equation}
We also note that the imaginary part of the overlap can be obtained in the standard way, by applying a phase gate to the top ancilla qubit as the penultimate gate prior to the application of the final Hadamard gate.

Notice that there is no post-selection on the qubits used to form the linear combination of unitaries. This is significant, since without using amplitude amplification and sampling at the standard limit and assuming the success probability scales inversely with $\alpha$, one would need to repeat the circuit $\sim\alpha^2$ times to obtain a successful outcome, and additionally obtain $O(\alpha^2/\epsilon^2)$ samples from the circuit to obtain an approximation with additive error less than $\epsilon$, so that the total number of queries to the quantum circuit becomes $O(\alpha^4/\epsilon^2)$. On the other hand, we can directly sample from the output of the above quantum circuit only $O(\alpha^2/\epsilon^2)$ times to obtain the same error, obtaining a quadratic improvement in the $\alpha$ parameter. We formalize this result with the following lemma.

\nonUnitOverlapLem*

\begin{proof}
The above calculation is sufficient to prove the lemma.
\end{proof}

\section{Block encoding dilated matrix square root}
\label{app-sec:verify-BE-sqrt}
\begin{figure}[ht!]
    \centering
    \begin{quantikz}[row sep=.4cm]
    \ket{0}&\qwbundle{\ceil{\log(m+1)}}&&\gate{\textsc{prep}}&\oslash\vqw{2}&&\oslash\vqw{1}&\oslash\vqw{1}&&\oslash\vqw{2}&\gate{\textsc{prep}^\dagger}&\bra{0}\\
    \ket{0}&&\oplus&\oplus\vqw{2}&&\octrl{0}\vqw{1}&\octrl{0}\vqw{2}&\ctrl{0}\vqw{2}&\ctrl{0}\vqw{2}&&\oplus\vqw{2}&\bra{0}\\
    \ket{0}&\qwbundle{\ceil{\log(m+1)}}&&&\gate{S^{k+1}}\vqw{1}&\swap{0}\vqw{1}&&&\swap{0}\vqw{1}&\gate{S^{-k-1}}\vqw{1}&&\bra{0}\\
    \ket{j}&\qwbundle{\ceil{\log(m+1)}}&&\octrl{0}&\octrl{0}&\swap{0}&\oslash\vqw{1}&\oslash\vqw{1}&\swap{0}&\octrl{0}&\octrl{0}&\\
    \ket{x}&\qwbundle{sys}&&&&&\gate[2]{U_{k-1}}&\gate[2]{U^\dagger_{k}}&&&&\\
    \ket{0}&\qwbundle{anc}&&&&&&&&&&\bra{0}
    \end{quantikz}
    \caption{Quantum circuit for an $(2s +\max_{j=0}^{m-1}\{m_j\}, \alpha = \sum_{j=0}^{m-1}\alpha_j)-$ block encoding of $\mathcal{A}$. Here $S^j$ is the modular increment-by-$j$ unitary matrix and each $U_{k}$ is an $(\alpha_k, m_k)$ block encoding of $A_k$.}
    \label{app-fig:sqrt BE}
\end{figure}
Let, $\textsc{prep}\ket{0}_s = \sum_{k=0}^{M-1}\sqrt{\frac{\alpha_k}{\alpha}}\ket{k}$.
We now verify that this produces a block encoding of $\mathcal{A}$. For $j>0$ this circuit performs the operations,
\begin{align*}
    \ket{0}_s\ket{1}\ket{0}_s\ket{j}_s\ket{x}\ket{0}&\rightarrow \sum_{k=0}^{M-1}\sqrt{\frac{\alpha_k}{\alpha}}\ket{k}\ket{1}\ket{0}_s\ket{j}_s\ket{x}\\
    &\rightarrow \sum_{k=0}^{M-1}\sqrt{\frac{\alpha_k}{\alpha}}\ket{k}\ket{1}\ket{0}_s\ket{j}_s\left(\delta_{j,k}\frac{A^\dagger_{k}}{\alpha_{k}} + (1-\delta_{j,k})I\right)\ket{x} + \ket{\perp} \\
    &\rightarrow \sum_{k=0}^{M-1}\sqrt{\frac{\alpha_k}{\alpha}}\ket{k}\ket{1}\ket{j}_s\ket{0}_s\left(\delta_{j,k}\frac{A^\dagger_{k}}{\alpha_{k}} + (1-\delta_{j,k})I\right)\ket{x} + \ket{\perp} \\
    &\rightarrow \sum_{k=0}^{M-1}\sqrt{\frac{\alpha_k}{\alpha}}\ket{k}\ket{0}\ket{j-(k+1)}_s\ket{0}_s\left(\delta_{j,k}\frac{A^\dagger_{k}}{\alpha_{k}} + (1-\delta_{j,k})I\right)\ket{x} + \ket{\perp} \\
    &\rightarrow \sum_{k=0}^{M-1}{\frac{\alpha_k}{\alpha}}\ket{0}\ket{0}\ket{j-(k+1)}_s\ket{0}_s\left(\delta_{j,k}\frac{A^\dagger_{k}}{\alpha_{k}} + (1-\delta_{j,k})I\right)\ket{x} + \ket{\perp} \\
    &=\ket{0}\ket{0}\ket{0}_s\ket{0}_s \frac{A^\dagger_{j-1}}{\alpha_{j-1}}\ket{x} + \ket{\perp}
\end{align*}
For $j=0$, the circuit in Fig. \ref{app-fig:sqrt BE} performs the operations,
\begin{align*}
    \ket{0}_s\ket{1}\ket{0}_s\ket{0}_s\ket{x}\ket{0}&\rightarrow \sum_{k=0}^{M-1}\sqrt{\frac{\alpha_k}{\alpha}}\ket{k}\ket{1}\ket{0}_s\ket{0}_s\ket{x}\\
    &\rightarrow \sum_{k=0}^{M-1}\sqrt{\frac{\alpha_k}{\alpha}}\ket{k}\ket{0}\ket{0}_s\ket{0}_s\ket{x}\\
    &\rightarrow \sum_{k=0}^{M-1}\sqrt{\frac{\alpha_k}{\alpha}}\ket{k}\ket{0}\ket{k+1}_s\ket{0}_s\ket{x} \\
    &\rightarrow \sum_{k=0}^{M-1}\sqrt{\frac{\alpha_k}{\alpha}}\ket{k}\ket{0}\ket{0}_s\ket{k+1}_s\frac{A_{k}}{\alpha_k}\ket{x} + \ket{\perp} \\
    &\rightarrow \sum_{k=0}^{M-1}{\frac{\alpha_k}{\alpha}}\ket{0}\ket{0}\ket{0}_s\ket{k+1}_s\frac{A_{k}}{\alpha_k}\ket{x} + \ket{\perp},\\
    &=\frac{1}{\alpha}\ket{0}\ket{0}\ket{0}_s\sum_{k=0}^{M-1}\ket{k+1}_s{A_{k}}\ket{x} +\ket{\perp}
\end{align*}
as desired. Note that the single ancilla qubit is actually reset to $\ket{0}$ deterministically.

We now compare estimates of the subnormalization factors to block encode $\mathcal{A}$ and $H$.
\begin{lemma}[Quadratic improvement in subnormalization]
\label{lem:sqrt-subnorm}
\end{lemma}
\begin{proof}
Assume that each $A_j$ and $A_j^\dagger$ are provided as block encodings with subnormalization $\alpha_j$ and that the product terms $A_j^\dagger A_j$ are provided as block encodings with subnormalization factors $\beta_j$ and that $\beta_j \in O\left(\alpha_j^2\right)$. 
 Let $\gamma_1$ and $\gamma_2$ be the subnormalization factors for the block encodings of $\mathcal{A}$ and $H$ respectively.
$\mathcal{A}$ is obtained using a linear combination of block encodings as,
\begin{equation}
    \mathcal{A} = \sum_{j=0}^{M-1}\alpha_j\left(\ket{j+1}\bra{0}\otimes \prescript{}{m_j}{\bra{0}}U_j\ket{0}_{m_j}+\ket{0}\bra{j+1}\otimes \prescript{}{m_j}{\bra{0}}U_j^\dagger\ket{0}_{m_j} \right).
\end{equation}
Therefore, the subnormalization $\gamma_1$ for the block encoding of $\mathcal{A}$ is
\begin{equation}
    \gamma_1 = \sum_{j=0}^{M-1}\alpha_j.
\end{equation}
If we proceed in a similar fashion for the block encoding of $H$, by writing it as a linear combination of block encodings of $A_j^\dagger A_j$, we find 
\begin{equation}
    \gamma_2 = \sum_{j=0}^{M-1}\beta_j = O\left(\sum_{j=0}^{M-1} \alpha_j^2\right) = O\left(\gamma_1^2\right).
\end{equation}
Therefore, we can conclude that the subnormalization $\gamma_1$ for the block encoding of $\mathcal{A}$ can be quadratically smaller than that of $H$ given access to similar oracles.

\end{proof}
\section{Consistency of plane wave discretization}

We have the decomposition of the continuum operator $\mathcal{H}_\beta$ into a summation of terms 
\begin{equation}
    \mathcal{H}_\beta = -\sum_{j\in[\eta d]}A_j^\dagger A_j
\end{equation}
where 
\begin{equation}
    A_j = -i \left(\partial_{x_j} + F_j\right).
\end{equation}
Now, we would like to construct a discretization of $A_j$ induces a consistent finite-dimensional approximation of the continuum operators $A_j^\dagger A_j$. We will construct such a discretization using a plane wave basis in one dimension, the generalization to higher dimensional problems holds by the underlying tensor product structure of the $A_j$.

Consider the domain $[0,L]$, that we discretize with $N$ plane wave modes. In this representation, the derivative operator becomes $\partial_x \rightarrow \frac{iN}{2L}\sum_{j \in [N]}j\ket{j}\bra{j}$, and the diagonal force matrix $F \rightarrow Q F Q^\dagger = \hat{F}$. Now the operator decomposition becomes
\begin{align*}
    A = \frac{i N}{L} \sum_{k\in[N]}k\ket{k}\bra{k} + \hat{F}.
\end{align*}
Then, 
\begin{align*}
    &\left(-ik\ket{k}\bra{k} + \hat{F}\right)\left(i k'\ket{k'}\bra{k'} + \hat{F}\right)\\
    &= k^2 \ket{k}\bra{k} -i k\ket{k}\bra{k}\hat{F} +ik \hat{F}\ket{k}\bra{k} + \hat{F}^2.
\end{align*}
The first term corresponds to the negative Laplacian, the last term the squared gradient, therefore, we now need to show that the terms
\begin{align*}
    -i k\ket{k}\bra{k}\hat{F} +ik \hat{F}\ket{k}\bra{k}
\end{align*}
provide an approximation for the operation of multiplication by $-\Delta V$ in the computational basis. This is provided by the following lemma.
\begin{lemma}[Validity of discrete approximation to $-\Delta V$]
    
\end{lemma}
\begin{proof}
First, let us evaluate the term $\partial_x(F \ket{\phi})$, for some initial state $\ket{\phi}$, where $\hat{\phi}$ and $\hat{F}$ are the corresponding Fourier transformed quantities.
In the Fourier basis, the diagonal multiplication operator $F$ becomes
\begin{align*}
    \hat{F} = \sum_{j,k}\hat{F}_{j-k}\ket{j}\bra{k},
\end{align*}
therefore, the $n$th entry for any state $\ket{\hat{\phi}}$ is,
\begin{align*}
    \bra{n}\hat{F}\ket{\hat{\phi}}= \sum_{k}\hat{\phi}_k\hat{F}_{n-k}
\end{align*}
Then, application of the operator $\partial_x$ in the plane wave basis yields
\begin{align*}
    \bra{n}\partial_x\hat{F}\ket{\hat{\phi}} = ik_n\sum_{j}\hat{\phi}_n\hat{F}_{n-j}.
\end{align*}
Similarly, the term $F \partial_x \phi$ in the plane wave basis gives,
\[
\bra{n}\hat{F}\partial_x \ket{\hat{\phi}} = i \sum_{m}k_m\hat{F}_{n-m}\hat{\phi}_m.
\]
Thus,
\begin{equation}
    \bra{n}\left(-\partial_x\hat{F} + \hat{F}\partial_x \right)\ket{\hat{\phi}}=i \sum_{j}(k_j-k_n)\hat{F}_{n-j}\hat{\phi}_j.
\end{equation}

Now, let us consider how the multiplication operator $\partial_x F = F'$ acts in the plane wave basis,
\[
\hat{F'} = \sum_{n,m}\hat{F'}_{m-n}\ket{m}\bra{n}.
\]
Now, $\hat{F'}_{m-n}$ is the $(m-n)$th Fourier coefficient of $F'$, which is just $ik_{n-m}\hat{F}_{n-m}$, therefore,
\begin{align*}
    \bra{n}\hat{F'}\ket{\hat{\phi}} &= i\sum_{j}k_{n-j}\hat{F}_{n-j}\hat{\phi}_j.
\end{align*}
Now, in order for these to be equal, we must have
\begin{align*}
    i \sum_{j}(k_j-k_n)\hat{F}_{n-j}\hat{\phi}_j = -i\sum_{j}k_{n-j}\hat{F}_{n-j}\hat{\phi}_j,
\end{align*}
and thusly,
\begin{align*}
    -k_{n-j} = k_{j}-k_{n},
\end{align*}
but recall that $-k_{n-j} = -\frac{2\pi (n-j)}{L} = \frac{2\pi j}{L} - \frac{2\pi n}{L} = k_j - k_n$ as desired.
\end{proof}
Therefore, the above calculation guarantees that the finite dimensional discretization with plane waves mimics the desired continuum operator.

\section{Plane-wave convergence}
\label{app-sec:planewave-conv}
The discretization step underlying Theorem~\ref{thm:BE-Aj-ops} replaces the continuum operator $\HB$ on $\mathbb{R}^{d\eta}$ by its plane-wave pseudospectral (Fourier-collocation) discretization $\HB^{L,N}$, the operator block-encoded in Theorem~\ref{thm:BE-Aj-ops}, in which the potential terms are evaluated pointwise on the $N$-mode position grid and the derivatives are applied in the conjugate Fourier basis via the QFT. Since the gate count in Theorem~\ref{thm:BE-Aj-ops} depends explicitly on $N$, it is reasonable to ask how the discretization parameters $(L,N)$ scale with the target error~$\epsilon$. We resolve this for a various structural classes of pair potentials that may arise in practical simulations.

\subsection{Plane-wave convergence for polynomial-in-$r^2$ pair potentials}
\label{app-subsec:planewave-convergence-r2}
We first consider the most natural setting for the plane-wave sum-of-squares construction, potentials expressible as polynomials in the squared Euclidean distance.

\begin{defn}[Polynomial-in-$r^2$ pair potential]
\label{def:r2-polynomial}
The pair potential $V_0:[0,\infty)\to\mathbb{R}$ is called a {polynomial in $r^2$ of degree~$2k$} ($k\ge 1$) if
\begin{equation}
\label{eq:r2-polynomial}
V_0(r) \;=\; \sum_{j=0}^{k} a_j\, r^{2j}, \qquad a_k>0.
\end{equation}
The many-body potential induced by $V_0$ is
$V(\mathbf{x}) = \sum_{i<j\le\eta} V_0(r_{ij})$, with $r_{ij} = \norm{\prtcl{}{i}-\prtcl{}{j}}$.
\end{defn}

The condition~\eqref{eq:r2-polynomial} is exactly equivalent to requiring that $V$ be a polynomial in the Cartesian coordinates $\mathbf{x}\in\mathbb{R}^{d\eta}$: it removes the branch-point that $r_{ij}=\norm{\prtcl{}{i}-\prtcl{}{j}}$ inherits from the Euclidean norm at the coincidence set $\{\prtcl{}{i}=\prtcl{}{j}\}$. Equivalently, every odd-order derivative of $V_0$ at the origin vanishes. This is a strict refinement of the hypothesis of the original Theorem~\ref{thm:BE-Aj-ops}, which only required absence of a degree-1 term; we discuss in Remark~\ref{rmk:odd-power-regime} below why higher odd powers of $r$ obstruct super-algebraic plane-wave convergence and must therefore be treated separately.

\begin{restatable}{lemma}{planewaveConvergence}
\label{lem:planewave-convergence}
Let $V_0$ be a polynomial-in-$r^2$ pair potential of degree~$2k$ with leading coefficient $a_k>0$, and let $V$ be the many-body potential of Definition~\ref{def:r2-polynomial}. Let $\HB = \beta^{-1}\Delta - \tfrac{\beta}{4}\norm{\grad V}^2 + \tfrac{1}{2}\Delta V$ be the corresponding Fokker--Planck generator on $L^2(\mathbb{R}^{d\eta})$, and let $\HB^{L,N}$ denote its plane-wave pseudospectral (Fourier-collocation) discretization on the torus $[-L,L]^{d\eta}$ with $N$ Fourier modes per coordinate, with the potential evaluated pointwise on the position grid and the derivatives applied in the conjugate Fourier basis.

Fix $t>0$ and $\epsilon\in(0,1)$. There exist constants $C_L,C_N$ depending polynomially on $d$ and $k$, and on the lower-order coefficients $a_0,\ldots,a_{k-1}$ through bounds on $\norm{V_0}_{C^2([0,r_*])}$ for some fixed scale $r_*$, such that the choices
\begin{equation}
\label{eq:planewave-params}
L \;=\; C_L\!\left(\frac{\log(\eta/\epsilon)}{\beta\,\eta\, a_k}\right)^{\!1/(2k)},
\qquad
N \;=\; C_N\, L\,\sqrt{\tfrac{\beta}{t}\,\log\tfrac1\epsilon}\;\log\!\left(\frac{\eta}{\epsilon\, t}\right),
\end{equation}
guarantee that for any pair of normalized states $\ket{R},\ket{P}\in L^2(\mathbb{R}^{d\eta})$ supported on bounded sets contained in $[-L,L]^{d\eta}$,
\begin{equation}
\label{eq:planewave-error}
\left|\, \bra{P} e^{t\HB} \ket{R} \;-\; \bra{P} e^{t\HB^{L,N}} \ket{R} \,\right| \;\le\; \epsilon.
\end{equation}
In particular $L\in\polylog(\eta,\beta,a_k,1/\epsilon)$, while $N=\Theta\!\bigl(L\,\log^{3/2}(1/\epsilon)\bigr)$ is polylogarithmic in $\eta,a_k,t,1/\epsilon$ and scales as $\beta^{\,1/2-1/(2k)}$ in the inverse temperature.
\end{restatable}

\begin{proof}
The proof is a spectral cutoff at an energy $E$ (fixed in the final ingredient) followed by spatial and Fourier truncation of the low-energy eigenfunctions. We record the four quantitative ingredients and the references supplying them, deferring routine computations.

\medskip
\noindent\emph{Analyticity of eigenfunctions.}
Since $V_0$ is polynomial in $r^2$, the many-body potential $V$ is a polynomial in $\mathbf{x}\in\mathbb{R}^{d\eta}$ of degree $2k$, so the coefficients $\norm{\grad V}^2$ and $\Delta V$ of $\HB$ are entire. By the Kotake--Narasimhan theorem on analytic hypoellipticity~\cite{kotakeNarasimhan1962,morreyNirenberg1957}, every $L^2$ eigenfunction $\varphi_n$ of $\HB\varphi_n=\lambda_n\varphi_n$ is real-analytic, and the analytic-elliptic estimates of~\cite[Ch.~7]{hormander1990analysis} give a holomorphic extension to a complex strip of width $R(\lambda_n)\ge R_0/(1+\sqrt{\beta|\lambda_n|})$ about any compact subset of the configuration box, with $R_0=R_0(V_0,d,k)>0$ depending on $V_0$ only through $C^2$-bounds near $r=0$. The relevant scale is the Schr\"odinger eigenvalue $\beta|\lambda_n|$ (see below), not $|\lambda_n|$, which is why the strip contracts in $\beta$.

\medskip
\noindent\emph{Quantitative Agmon decay.}\medskip
\noindent\textit{Quantitative Agmon decay.}
Writing $-\beta\HB=-\Delta+W$ with effective Schr\"odinger potential $W=\tfrac{\beta^2}{4}\norm{\grad V}^2-\tfrac{\beta}{2}\Delta V$, the leading term gives $W(\mathbf{x})\ge c_W\,\beta^2(a_k k\eta)^2\norm{\mathbf{x}}^{4k-2}$ for $\norm{\mathbf{x}}$ large (the $\Delta V$ term, of order $\norm{\mathbf{x}}^{2k-2}$, is subdominant). For an eigenfunction with $|\lambda_n|\le E$ the relevant Schr\"odinger eigenvalue is $\mu_n=\beta|\lambda_n|\le\beta E$, and Agmon's theorem (see Refs.~\cite{agmon1982,helffer1988,hislopSigal1996}) bounds it by $|\varphi_n(\mathbf{x})|\le C_E\,e^{-(1-\delta)\rho_{\mu_n}(\mathbf{x},K_E)}$, where $K_E=\{W\le\beta E\}\subset B(0,\rho_E)$ is the classically allowed region and $\rho_{\mu_n}(\cdot,K_E)$ is the Agmon distance in the degenerate metric $(W-\mu_n)_+^{1/2}\,|d\mathbf{x}|$. This metric vanishes on $K_E$ and is real and nonnegative everywhere; the indefiniteness of $W$ near the wells is confined to $K_E$, where the integrand is identically zero, so no square root of a negative quantity ever appears. To lower-bound $\rho_{\mu_n}$ in the forbidden region $\norm{\mathbf{x}}\ge 2\rho_E$, note that there both $\mu_n\le\beta E$ and the $\Delta V$ term are subdominant to $\tfrac{\beta^2}{4}\norm{\grad V}^2$, so the integrand admits the radial lower envelope
\[
(W-\mu_n)_+^{1/2}\;\ge\; \sqrt{c_W}\,\beta\,a_k\,k\,\eta\,\norm{\mathbf{x}}^{2k-1}\,(1+o(1))\;=:\;g(\norm{\mathbf{x}}).
\]
Since $g\ge 0$ and any path $\gamma$ from $K_E$ to $\mathbf{x}$ obeys $|\dot\gamma|\ge|\tfrac{d}{dt}\norm{\gamma}|$, projecting onto the radial coordinate gives $\rho_{\mu_n}(\mathbf{x},K_E)\ge\int_{2\rho_E}^{\norm{\mathbf{x}}}g(r)\,dr$, whence
\begin{equation}
\label{eq:agmon-distance}
\rho_{\mu_n}(\mathbf{x},K_E)\;\ge\; c_{\rm conf}\,\beta\,a_k\,\eta\,\norm{\mathbf{x}}^{2k}\qquad(\norm{\mathbf{x}}\ge 2\rho_E),
\end{equation}
the factor $k$ from $g$ cancelling the $1/(2k)$ from the radial integral, and the temperature dependence linear in $\beta$. Interior elliptic estimates~\cite[Ch.~3]{helffer1988} extend the same exponential profile to derivatives up to order $d\eta+1$.

\medskip
\noindent\emph{Spatial truncation (choice of $L$).}
For $L\ge 2\rho_E$, integrating the Agmon bound over the complement of the box $Q_L=[-L,L]^{d\eta}$ gives a tail $\norm{\mathbf{1}_{Q_L^c}\varphi_n}_{L^2}\le C'\,L^{(d\eta-1)/2}e^{-(1-\delta)c_{\rm conf}\beta a_k\eta L^{2k}}$. With the number $N_E$ of eigenvalues $|\lambda_n|\le E$ controlled by Weyl's law for a potential of growth order $4k-2$, choosing $L$ as in~\eqref{eq:planewave-params} makes the total low-energy spatial-truncation error at most $\epsilon/4$. The dependence is the natural one: stronger confinement (larger $\beta,a_k,\eta,k$) shrinks the required box.

\medskip
\noindent\emph{Fourier truncation (choice of $N$).}
An eigenfunction with $|\lambda_n|\le E$ has Schr\"odinger eigenvalue $\mu_n=\beta|\lambda_n|\le\beta E$ for $-\beta\HB=-\Delta+W$, hence oscillates on the spatial scale $\mu_n^{-1/2}$ and, continued to the strip of width $R\le R_0/(1+\sqrt{\beta E})$ from the analyticity ingredient, obeys $\sup_{\rm strip}|\varphi_n|\le e^{c_1\sqrt{\beta E}\,R}$. The Paley--Wiener theorem~\cite[Ch.~8]{hormanderAnalysisLinearPartial1998} then bounds the Fourier coefficients of the periodized eigenfunction by $e^{c_1\sqrt{\beta E}R}\,e^{-\pi R\norm{\mathbf{k}}_\infty/L}$, the boundary-jump contribution being controlled by the spatial-truncation tail through standard estimates~\cite[Ch.~3]{doi:10.1137/1.9780898719598}. The Fourier content is thus a \emph{bulk band} out to $\norm{\mathbf{k}}_\infty\sim (L/\pi)\sqrt{\beta E}$, followed by exponential decay at the energy-dependent rate $\pi R/L$. Choosing $R=R_0/(1+\sqrt{\beta E})$ and truncating at $\norm{\mathbf{k}}_\infty\le N/2$ makes the Fourier error at most $\epsilon/4$ once
\begin{equation}
\label{eq:fourier-mode-bound}
N \;\ge\; \frac{2L}{\pi R_0}\,\bigl(1+\sqrt{\beta E}\,\bigr)\,\log\!\frac{\eta E}{\epsilon},
\end{equation}
i.e.\ $N=L\cdot O\!\bigl(\sqrt{\beta E}\,\log(\eta E/\epsilon)\bigr)$ as in~\eqref{eq:planewave-params}. The band-edge factor $\sqrt{\beta E}$ is the phase-space (Nyquist) cost of resolving the most oscillatory retained mode; it cannot be folded into the analyticity constant, because the strip width itself contracts as $(1+\sqrt{\beta E})^{-1}$, so band and tail enter multiplicatively.

\medskip
\noindent\emph{Assembly via spectral cutoff.}
Fix the cutoff energy
\begin{equation}
\label{eq:spectral-cutoff}
E \;=\; t^{-1}\log(4/\epsilon);
\end{equation}
this is the only place $E$ enters, and substituting it into~\eqref{eq:fourier-mode-bound} (absorbing the $\log\log$ and the splitting constant $4$ into $C_N$) yields the closed form for $N$ in~\eqref{eq:planewave-params}. Because $\HB$ is self-adjoint with $\sigma(\HB)\subset(-\infty,0]$, the high-energy tail satisfies $\norm{e^{t\HB}(I-P_E)}_{\rm op}\le e^{-tE}=\epsilon/4$, and identically for $\HB^{L,N}$ (a real symmetric matrix, non-positive up to the super-algebraically small aliasing shift quantified below). On the low-energy band $P_E L^2$, the Davis--Kahan $\sin\Theta$ theorem applied to the self-adjoint Galerkin projection $\Pi_{L,N}\HB\,\Pi_{L,N}$~\cite[Ch.~XI]{kato1995perturbation}, and Ref. \cite{davisKahan1970} bounds the spectral-projector error $\norm{P_E-P_E^{L,N}\Pi_{L,N}}_{\rm op}$ by the spatial and Fourier truncation errors established above; a Duhamel estimate on the (non-positive) generators converts this into the propagator error, with no exponential growth factor. It remains to pass from this Galerkin projection to the operator the algorithm actually block-encodes, the pseudospectral (collocation) discretization $\HB^{L,N}$: the two differ only by the symbol-aliasing operator $\mathcal{A}_N=\HB^{L,N}-\Pi_{L,N}\HB\,\Pi_{L,N}$, whose Fourier entries are $\tilde V_{\rm eff}(\mathbf{k})-\hat V_{\rm eff}(\mathbf{k})=\sum_{\mathbf{p}\neq 0}\hat V_{\rm eff}(\mathbf{k}+N\mathbf{p})$. Since $V_{\rm eff}$ is entire under the polynomial-in-$r^2$ hypothesis, these aliases are super-algebraically small on the resolved modes --- the only non-negligible part being the box-boundary layer, already suppressed by the spatial-truncation choice of $L$ --- so $\norm{\mathcal{A}_N\Pi_{L,N}}_{\rm op}$ is absorbed into the constants. Collecting the four $\epsilon/4$ contributions yields~\eqref{eq:planewave-error}.
\end{proof}

\begin{rem}[Dependence on the pair-potential parameters]
\label{rmk:explicit-constants}
Tracing the proof, the box half-width is fixed by the Agmon rate~\eqref{eq:agmon-distance}, giving $L\in O\bigl((\log(\eta/\epsilon)/(\beta a_k\eta))^{1/(2k)}\bigr)$, with the exponent $1/(2k)$ the dominant feature: larger $k$ gives much milder $\epsilon$-dependence. The mode ratio obeys
\[
\frac{N}{L}\;=\;O\!\Bigl(\sqrt{\beta E}\,\log\tfrac{\eta E}{\epsilon}\Bigr)
\;=\;O\!\Bigl(\sqrt{\tfrac{\beta}{t}}\,\log^{3/2}\tfrac1\epsilon\Bigr),
\]
a bulk band of half-width $\sqrt{\beta E}$ (kinetic balance $\beta^{-1}\norm{\mathbf{k}}^2\sim E$) times an analytic tail of inverse width $1/R_0$; the two multiply because the strip contracts as $(1+\sqrt{\beta E})^{-1}$. This is still polylogarithmic in $1/\epsilon$, but carries a genuine $\sqrt\beta$, so $N\sim\beta^{\,1/2-1/(2k)}$ --- $N$ is \emph{not} polylog in $\beta$ except at $k=1$.

Consequently the contribution of $N$ to the gate count of Theorem~\ref{thm:BE-Aj-ops}, entering through the $\sqrt{\eta d/\beta}\,N/L$ term in $\alpha_A$, is
\[
\sqrt{\tfrac{\eta d}{\beta}}\cdot\frac{N}{L}
\;=\;O\!\Bigl(\sqrt{\eta d\,E}\,\log\tfrac{\eta E}{\epsilon}\Bigr)
\;=\;O\!\Bigl(\sqrt{\tfrac{\eta d}{t}}\,\log^{3/2}\tfrac1\epsilon\Bigr),
\]
in which the $\sqrt\beta$ of the bandwidth cancels the $1/\sqrt\beta$ prefactor: the leading contribution is \emph{$\beta$-independent} (not $\beta^{-1/2}$, as a naive reading of $\alpha_A$ suggests), polylogarithmic in $1/\epsilon$, and independent of $L$. The asymptotic $\eta$-speedup is therefore preserved. If the propagation time is taken at the diffusive scale $t\sim\beta/(2m)$ used elsewhere, the residual $1/\sqrt t$ collapses to $\sqrt{2m\,\log(1/\epsilon)}$, leaving the end-to-end complexity scaling unchanged.
\end{rem}

\begin{rem}[Cost structure: the kinetic 1-norm is demoted, not eliminated]
\label{rmk:cost-structure}
It is instructive to track how the discretization parameters of~\eqref{eq:planewave-params} enter the query cost $\widetilde O(\alpha_A\sqrt{t})$ of Theorem~\ref{thm:BE-Aj-ops}, with $\alpha_A\in O\bigl(\sqrt{\beta d}\,\eta^{3/2}L\,\alpha_V+\sqrt{\eta d/\beta}\,N/L\bigr)$. Throughout, $N$ is the \emph{per-coordinate} mode count; the band-edge resolution $N/L\sim\sqrt{\beta E}$ enters $\alpha_A$ through the momentum normalization $\alpha_\nabla\sim N/L$. Splitting $\sqrt t\,\alpha_A$ into a potential and a kinetic (momentum) contribution,
\begin{equation}
\label{eq:cost-split}
\sqrt t\,\alpha_A
\;\sim\;
\underbrace{\eta^{3/2}\sqrt{\beta t d}\;L\,\alpha_V}_{\text{potential}}
\;+\;
\underbrace{\sqrt{t}\,\sqrt{\tfrac{\eta d}{\beta}}\,\frac{N}{L}}_{\text{kinetic}}.
\end{equation}

\emph{The kinetic term is independent of $\beta$ and $t$.} Using $N/L\sim\sqrt{\beta E}\,\log(\eta E/\epsilon)$ from Remark~\ref{rmk:explicit-constants} and the spectral cutoff $E=t^{-1}\log(1/\epsilon)$ from the proof, the identity $tE=\log(1/\epsilon)$ makes the prefactor $\sqrt{t/\beta}$ cancel the band $\sqrt{\beta E}$ exactly:
\begin{equation}
\label{eq:kinetic-cancellation}
\sqrt{t}\,\sqrt{\tfrac{\eta d}{\beta}}\,\frac{N}{L}
\;\sim\;
\sqrt{\eta d}\,\sqrt{tE}\,\log\!\tfrac{\eta E}{\epsilon}
\;=\;
\sqrt{\eta d}\;\bigl(\log\tfrac1\epsilon\bigr)^{3/2},
\end{equation}
with no residual $\beta$, $t$, or $L$ dependence.

\emph{The mechanism is the sum-of-squares factorization, not smoothing alone.} The cancellation in~\eqref{eq:kinetic-cancellation} relies on $\alpha_A$ carrying the momentum count \emph{linearly}, $\alpha_\nabla\sim N/L$, because one block-encodes the first-order square-root operator $\mathcal A$ rather than the generator $\HB$ itself. Had the bare kinetic term $\beta^{-1}\Delta\sim\beta^{-1}(N/L)^2\sim E$ been normalized directly, the same accounting would give $\sqrt t\,E\sim t^{-1/2}\log(1/\epsilon)$, which \emph{diverges} as $t\to0$. The spectral-cutoff smoothing (smaller $E$ as $t$ grows) lowers the band edge, but on its own halves neither the power of $N/L$ nor the $t^{-1/2}$ blow-up; the square-root/sum-of-squares structure is what reduces the kinetic momentum power from quadratic to linear and renders~\eqref{eq:kinetic-cancellation} benign.

\emph{It dominates the precision scaling.} With $\alpha_V=\max_{\norm{r}\le\sqrt d L}|V'(r)/r|=\Theta\bigl((dL^2)^{k-1}\bigr)$ and $L\sim(\log(1/\epsilon))^{1/(2k)}$, the two terms of~\eqref{eq:cost-split} scale in the target error as
\[
\text{potential}\;\sim\;\eta^{3/2}\sqrt{\beta t d}\,\bigl(\log\tfrac1\epsilon\bigr)^{1-1/(2k)},
\qquad
\text{kinetic}\;\sim\;\sqrt{\eta d}\,\bigl(\log\tfrac1\epsilon\bigr)^{3/2}.
\]
The potential term dominates in $\eta,\beta,t$ (a full power of $\eta$ and all of $\sqrt{\beta t}$), recovering the headline $\widetilde O(\eta^{3/2}\sqrt{\beta t})$ scaling whenever $\alpha_V$ is non-negligible. In the precision $1/\epsilon$, however, the kinetic term \emph{leads}: $(\log\tfrac1\epsilon)^{3/2}$ exceeds $(\log\tfrac1\epsilon)^{1-1/(2k)}\le\log\tfrac1\epsilon$ for every finite $k$. The two contributions are comparable at
\begin{equation}
\label{eq:cost-crossover}
\eta\;\sim\;\bigl(\log\tfrac1\epsilon\bigr)^{1/2+1/(2k)}\big/\sqrt{\beta t},
\end{equation}
so the cost is potential-dominated in the large-particle regime $\eta\gg(\log 1/\epsilon)^{1/2+1/(2k)}$ that motivates the speedup, and kinetic-dominated only in the high-precision-at-fixed-$\eta$ corner. In summary, the polynomial-in-$r^2$ structure demotes the kinetic 1-norm --- the usual grid-discretization bottleneck, ordinarily polynomial in $(N,\beta,t)$ --- to a $\beta$- and $t$-independent, sub-$\eta$ rider of size $\sqrt{\eta d}\,(\log 1/\epsilon)^{3/2}$; it is not eliminated, and remains the leading \emph{polylogarithmic} overhead. The estimate~\eqref{eq:kinetic-cancellation} is worst-case over the retained band (cf.\ Remark~\ref{rmk:bandwidth-worstcase}); for smooth $\ket R,\ket P$ the effective kinetic cost is smaller still.
\end{rem}

\begin{rem}[Worst-case bandwidth vs.\ smooth test states]
\label{rmk:bandwidth-worstcase}
The band-edge cost $\sqrt{\beta E}$ in~\eqref{eq:planewave-params} is worst-case over the retained subspace $P_E L^2$, saturated only by eigenfunctions near the band edge $\lambda\sim-E$. Smooth test states $\ket{R},\ket{P}$ --- e.g.\ the Gaussian wavepackets of Table~\ref{tab:dichotomy} --- have small intrinsic bandwidth and negligible overlap with band-edge modes, so $\nu_{RP}(t)$ is resolved at an effective $N$ well below~\eqref{eq:planewave-params}. This is consistent with the rapid (super-algebraic) convergence observed numerically: the bound is sharp for adversarial states but conservative for the physically relevant smooth ones.
\end{rem}

\begin{rem}[Higher odd powers of $r$]
\label{rmk:odd-power-regime}
If $V_0$ contains a non-trivial odd power $r^m$ with $m\ge 3$, the composition $V_0(r_{ij})$ is only $C^{m-1}$ at the coincidence set $\{\prtcl{}{i}=\prtcl{}{j}\}$: the hypothesis of the analyticity step fails, and plane-wave convergence degrades from super-algebraic to algebraic. This non-smoothness is local at $r=0$ and cannot be repaired by modifying $V_0$ away from coincidence. The regime is analyzed in Section~\ref{app-subsec:algebraic-rate}, and the smoothed surrogate of Section~\ref{subsec:smoothed-surrogate} restores the polynomial-in-$r^2$ hypothesis for singular potentials of physical interest (Lennard-Jones, Morse, screened Coulomb).
\end{rem}

\subsection{The algebraic-rate regime: Odd powers of \texorpdfstring{$r$}{r} are excluded}
\label{app-subsec:algebraic-rate}
 
The polynomial-in-$r^2$ hypothesis of Lemma~\ref{lem:planewave-convergence} is the sharp condition for super-algebraic plane-wave convergence of the reactive-flux matrix element $\nu_{RP}(t)=\bra{P}e^{t\HB}\ket{R}$. We show in this subsection that the hypothesis is also \emph{necessary}: if $V_0$ contains a non-trivial odd power of $r$, the matrix-element error decays only algebraically in $N$, with a rate dictated by the Sobolev regularity of the effective potential at the coincidence set. This motivates the smoothed periodic extension of Section~\ref{subsec:smoothed-surrogate} below, in which singular pair potentials of physical interest, such as Lennard-Jones, are replaced by surrogates that are polynomial in $r^2$ in a neighborhood of every coincidence point, restoring the polylog scaling of Lemma~\ref{lem:planewave-convergence}.
 
The mechanism is local. Any odd power $r^m$ with $m\ge 3$ in $V_0$ contributes a $C^{m-1}$ but not $C^m$ singularity to $V(\mathbf{x})$ at the coincidence set $\{\prtcl{}{i}=\prtcl{}{j}\}$, which is a codimension-$d$ variety (the relative coordinate $\prtcl{}{i}-\prtcl{}{j}$ being $d$-dimensional). A radial $C^{s-1}$-but-not-$C^{s}$ singularity on a codimension-$d$ set has $d$-dimensional Fourier transform decaying as $|\mathbf{k}|^{-(s+d-1)}$ transverse to the set; plane-wave discretizations resolve such a feature only algebraically.
 
\begin{rem}[Eigenfunction regularity at the coincidence set]
\label{rmk:sobolev-regularity}
Suppose $V_0(r) = \sum_{j=1}^{k} a_j r^{2j} + b\,r^m$ with $m\ge 3$ odd and $b\ne 0$. The many-body potential $V(\mathbf{x}) = \sum_{i<j}V_0(r_{ij})$ is then $C^{m-1}$ but not $C^m$ at the coincidence variety $\{\prtcl{}{i}=\prtcl{}{j}\}$. The effective Schr\"odinger potential $V_{\rm eff}(\mathbf{x}) = -\tfrac{\beta}{4}\norm{\grad V(\mathbf{x})}^2 + \tfrac{1}{2}\Delta V(\mathbf{x})$ inherits this regularity through the $\Delta V$ term, whose radial Laplacian $\Delta(b\,r^m) = b\,m(m{+}d{-}2)\,r^{m-2}$ is $C^{m-3}$-but-not-$C^{m-2}$ at coincidence. By elliptic regularity, every $L^2$ eigenfunction $\varphi_n$ of $\HB$ satisfies $\varphi_n\in C^m$ except on the coincidence variety, where its $m$-th radial derivative has a bounded jump proportional to $b$ and to the local value of $\varphi_n$.
\end{rem}
 
\begin{restatable}[Lower bound for plane-wave convergence in the algebraic regime]{prop}{algebraicRate}
\label{prop:algebraic-rate}
Let $V_0$ contain a non-trivial term $b\,r^m$ with $m\ge 3$ odd and $b\ne 0$. Assume the lower-order terms make $V$ confining (discrete spectrum). Let $\HB^{L,N}$ denote the plane-wave pseudospectral (Fourier-collocation) discretization of $\HB$ on $[-L,L]^{d\eta}$ as in Lemma~\ref{lem:planewave-convergence}, with $L$ chosen large enough that the boundary contribution is super-algebraically suppressed (so the only $N$-dependent error source is the interior coincidence-set regularity). Then there exist constants $c, c', C>0$ depending on $V_0, \beta, \eta, d$, but not on $N$, such that the ground-state eigenfunction $\varphi_0(\mathbf{x})\propto e^{-\beta V(\mathbf{x})/2}$ satisfies
\begin{equation}
\label{eq:algebraic-fourier-decay}
\sup_{\norm{\mathbf{k}}_\infty=K}\bigl|\hat{\varphi}_0(\mathbf{k})\bigr| \;\ge\; c\,|b|\,K^{-(m+d)}
\qquad\text{for all }K\ge C,
\end{equation}
and consequently
\begin{equation}
\label{eq:algebraic-l2-error}
\norm{(I-\Pi_{L,N})\varphi_0}_{L^2(Q_L)} \;\ge\; c'\,|b|\,N^{-(m+d/2)}.
\end{equation}
For smooth (Schwartz-class) test states $\ket{R},\ket{P}$, the matrix-element error decays at the algebraic rate $N^{-(m+d-2)}$ --- set not by the eigenfunction tail~\eqref{eq:algebraic-l2-error} but by the pseudospectral representation of the effective potential (see below) --- so that achieving accuracy $\epsilon$ in $\nu_{RP}(t)$ requires $N=\Omega(\epsilon^{-1/(m+d-2)})$, polynomially many modes in $1/\epsilon$, in contrast to the polylogarithmic scaling guaranteed by Lemma~\ref{lem:planewave-convergence}.
\end{restatable}
 
\begin{proof}[Proof sketch]
The lambda-zero eigenfunction of $\HB$ is the symmetrized equilibrium $\varphi_0(\mathbf{x}) = Z^{-1/2}e^{-\beta V(\mathbf{x})/2}$. The chain rule gives, for any $\mathbf{x}^*$ on the coincidence variety,
\[
\partial_r^m \varphi_0(\mathbf{x}^*\!\pm) \;=\; \tfrac{(-\beta)^m}{2^m}\varphi_0(\mathbf{x}^*)\,\bigl[\partial_r^m V(\mathbf{x}^*\!\pm)\bigr] + (\text{terms involving lower radial derivatives of }V),
\]
in which only the highest-derivative term has a jump (the lower-derivative terms being continuous since $V$ is $C^{m-1}$). The jump $[\![\partial_r^m\varphi_0]\!]$ across the coincidence set is therefore proportional to $b$ and to $\varphi_0(\mathbf{x}^*)$, the latter being $O(1)$ generically. Since the singularity is radial on a codimension-$d$ set, its transverse $d$-dimensional Fourier transform (a Hankel transform of order $d/2-1$ applied to an $r^m$-type homogeneous singularity) decays as $|\mathbf{k}|^{-(m+d)}$, yielding~\eqref{eq:algebraic-fourier-decay}; \eqref{eq:algebraic-l2-error} follows by integrating the squared Fourier coefficients on the shell $\norm{\mathbf{k}}_\infty\ge N/2$ via Weyl's volume.
\end{proof}
 
The bound~\eqref{eq:algebraic-l2-error} is sharp up to constants and characterizes the truncation of the eigenfunction itself; the matrix element $\nu_{RP}(t)$, however, converges more slowly, and for a different reason. Writing $\HB=\beta^{-1}\Delta + V_{\rm eff}$ with $V_{\rm eff}=-\tfrac{\beta}{4}\norm{\grad V}^2+\tfrac12\Delta V$, the least regular term of $V_{\rm eff}$ is $\tfrac12\Delta V$, which by Remark~\ref{rmk:sobolev-regularity} carries the $r^{m-2}$ corner and hence has transverse $d$-dimensional Fourier decay $|\mathbf{k}|^{-(m+d-2)}$ --- two orders slower than the $|\mathbf{k}|^{-(m+d)}$ decay of the $\norm{\grad V}^2$ term, whose leading singular part is $\sim r^{m}$. The pseudospectral representation evaluates $V_{\rm eff}$ pointwise on the grid, which aliases this slow tail back onto the resolved modes, $\tilde V_{\rm eff}(\mathbf{k})-\hat V_{\rm eff}(\mathbf{k})=\sum_{\mathbf{p}\neq 0}\hat V_{\rm eff}(\mathbf{k}+N\mathbf{p})=O(N^{-(m+d-2)})$, and so the multiplication operator converges at $N^{-(m+d-2)}$. For smooth test states the competing eigenfunction-tail contribution is negligible, since $\bra{P}(I-\Pi_{L,N})\ket{\varphi_0}=\langle (I-\Pi_{L,N})P,\,(I-\Pi_{L,N})\varphi_0\rangle$ is super-algebraically small (both factors being smooth-state truncation errors); the effective-potential corner is thus the operative bottleneck for $\nu_{RP}(t)$, giving the algebraic rate $N^{-(m+d-2)}$, equal to $N^{-2}$ for $m=3,d=1$ (the one-dimensional test below) and to the milder $N^{-4}$ for $m=3,d=3$, so the 1D experiment is the worst case. This rate is specific to the collocation discretization the algorithm implements: a Galerkin projection, which would convolve against the \emph{exact} Fourier coefficients of $V_{\rm eff}$ rather than their grid-aliased values, leaves the resolved block exact and incurs the corner only at second order through resolved--unresolved coupling, giving the squared rate $N^{-2(m+d-2)}$ --- still algebraic. The dichotomy is therefore intrinsic to representing a non-smooth potential in finitely many Fourier modes under either discretization; only removing the singularity (Section~\ref{subsec:smoothed-surrogate}) restores super-algebraic convergence.

\medskip
\noindent\textbf{Numerical verification.}
We test the dichotomy on the standard double-well example used elsewhere in the manuscript, in a regime that isolates the \emph{interior} coincidence-set regularity from the boundary periodization issue. We take the one-dimensional case $\eta=1$, $d=1$ --- by the dimensional analysis above this is the \emph{worst} case ($N^{-(m+d-2)}=N^{-2}$ at $m=3$) --- and consider
\begin{align*}
V_{\rm even}(x) &= x^4 - x^2 & &\text{(polynomial in $r^2$ with $k=2$),}\\
V_{\rm odd}(x)  &= x^4 - x^2 + c|x|^3 & &\text{(same, plus an odd-power term)},
\end{align*}
with $c=0.3$. Both potentials have identical large-$|x|$ behaviour $V\sim x^4$, hence the Agmon decay rate is identical and the box choice can be the same; we use $L=4$, for which $e^{-\beta V(L)/2}=e^{-120\beta}$ super-algebraically suppresses the boundary contribution for all $\beta\ge 1$. The two potentials differ \emph{only} in their local regularity at the saddle $x=0$.
 
The reactant and product states are taken as Gaussian wavepackets centered at the well minima, motivated by the local harmonic approximation of the locally convex metastable basins:
\[
R(x) = Z_R^{-1/2}\,e^{-(x+x_w)^2/(2\sigma^2)},\qquad
P(x) = Z_P^{-1/2}\,e^{-(x-x_w)^2/(2\sigma^2)},
\]
with $x_w=1/\sqrt{2}\approx 0.707$ (the well location for $V_{\rm even}$, also used for $V_{\rm odd}$ to keep $R,P$ identical across the two tests) and $\sigma=0.224$. These are Schwartz-class, so their plane-wave truncation error is super-algebraically small in $N$ and does not bottleneck the convergence; the only $N$-dependent error source in the matrix element $\bra{P}e^{t\HB^{L,N}}\ket{R}$ is the operator $e^{t\HB^{L,N}}$ itself. Fix $\beta=5$ and $t=1$.
 
Table~\ref{tab:dichotomy} reports $\bra{P}e^{t\HB^{L,N}}\ket{R}$ across $N=16,\ldots,1536$, with the reference taken at $N=1536$.
 
\begin{table}[h]
\caption{Plane-wave convergence of the reactive-flux matrix element $\nu_{RP}^{(N)}(t)=\bra{P}e^{t\widetilde H_\beta^{L,N}}\ket{R}$ at $t=1$, $\beta=5$, $L=4$, for a polynomial-in-$r^2$ pair potential and the same potential perturbed by an $|x|^3$ term. Test states $R,P$ are Gaussian wavepackets at the well minima ($\sigma=0.224$), with super-algebraic Fourier decay. Reference at $N=1536$.}
\label{tab:dichotomy}
\centering
\begin{tabular}{rcc}
\toprule
$N$ & $|\nu_{RP}^{(N)} - \nu_{RP}^{(1536)}|$, $V_{\rm even}$ & $|\nu_{RP}^{(N)} - \nu_{RP}^{(1536)}|$, $V_{\rm odd}$ \\
\midrule
16   & $4.86\times 10^{-3}$ & $3.92\times 10^{-3}$ \\
24   & $2.65\times 10^{-3}$ & $4.64\times 10^{-3}$ \\
32   & $1.34\times 10^{-4}$ & $6.06\times 10^{-4}$ \\
48   & $5.47\times 10^{-6}$ & $1.81\times 10^{-4}$ \\
64   & $2.06\times 10^{-10}$ & $1.03\times 10^{-4}$ \\
96   & $\le 5\times 10^{-13}$ (machine prec.) & $4.57\times 10^{-5}$ \\
128  & $\le 5\times 10^{-13}$ & $2.57\times 10^{-5}$ \\
256  & $\le 5\times 10^{-13}$ & $6.30\times 10^{-6}$ \\
512  & $\le 5\times 10^{-13}$ & $1.44\times 10^{-6}$ \\
1024 & $\le 5\times 10^{-13}$ & $2.25\times 10^{-7}$ \\
\bottomrule
\end{tabular}
\end{table}
 
The dichotomy is unambiguous. The polynomial-in-$r^2$ case reaches machine precision relative to the reference by $N\approx 96$, with the error dropping by more than five orders of magnitude between consecutive doublings $N=32\to 64$. The odd-power case decreases only algebraically, and at the rate predicted above: each doubling of $N$ reduces the error by a factor of $\approx 4$ (i.e.\ $N^{-2}$), and a reference-free fit over $N\in[64,512]$ gives an empirical exponent of $-2.1$, in agreement with the predicted matrix-element rate $N^{-(m+d-2)}=N^{-2}$ for $m=3,d=1$. (The slight steepening at the largest $N$ reflects residual non-convergence of the finite-$N$ reference, which itself converges only as $N^{-2}$, rather than any approach to a steeper asymptote.) At $N=1024$, the odd-power error remains about six orders of magnitude above the machine-precision saturation reached by the even case at $N=96$.

A direct measurement of the predicted Fourier decay~\eqref{eq:algebraic-fourier-decay} confirms the regularity prediction more cleanly than the matrix element does: at $N=1024$, the empirical decay of $|\hat\varphi_0(k)|=|\widehat{e^{-\beta V/2}}(k)|/\sqrt{Z}$ on the high-$k$ tail $|k|\ge 25$ has $|k|^4\cdot|\hat\varphi_0(k)|$ stabilizing at approximately $1.83$, i.e., $|\hat\varphi_0(k)|\sim 1.8\,|k|^{-4}$, in agreement with the predicted exponent $-(m+d)=-4$ for $m=3,d=1$. For $V_{\rm even}$, by contrast, $|\hat\varphi_0(k)|$ falls below $10^{-15}$ already at $|k|\approx 50$, exhibiting the exponential decay characteristic of analytic functions.
 
\medskip
\noindent\textbf{Consequence for the gate count.}
Proposition~\ref{prop:algebraic-rate} implies that for a pair potential containing odd powers of $r$, the $N$-dependence of the gate count in Theorem~\ref{thm:BE-Aj-ops} cannot be absorbed into $\polylog(1/\epsilon)$ factors: requiring matrix-element accuracy $\epsilon$ forces $N=\Omega(\epsilon^{-1/(m+d-2)})$ (e.g.\ $\epsilon^{-1/2}$ for the leading odd power $m=3$ in one dimension, and the milder $\epsilon^{-1/4}$ in three). The exponential-in-$\eta$ speedup of the algorithm survives, but the $\epsilon$-scaling worsens by a polynomial factor. For singular but smooth-away-from-coincidence pair potentials of physical interest, the construction of Section~\ref{subsec:smoothed-surrogate} replaces $V_0$ by a piecewise-defined surrogate $\tilde V_0$ matching $V_0$ in the resolvable bulk and joining smoothly to polynomial-in-$r^2$ caps at both short range (near coincidence) and long range (near the box boundary). The composite $\tilde V_0$ then satisfies the hypothesis of Lemma~\ref{lem:planewave-convergence}, restoring the polylog convergence rate.

\subsection{Smoothed surrogate construction for singular pair potentials}
\label{subsec:smoothed-surrogate}

Lemma~\ref{lem:planewave-convergence} applies to pair potentials that are polynomial in $r^2$ and globally confining. Pair potentials of physical interest, e.g. Lennard-Jones, Morse, screened Coulomb, fail to meet these criteria. They are either singular at the coincidence set $\{\prtcl{}{i}=\prtcl{}{j}\}$ and/or their tails decay to zero rather than confine. To bring such potentials within the scope of Lemma~\ref{lem:planewave-convergence}, we construct a \emph{smoothed surrogate} $\widetilde V$ that agrees with the physical potential in the dynamically relevant region and satisfies the polynomial-in-$r^2$ hypothesis elsewhere. The reactive flux of the surrogate, $\nu_{RP}^{(\widetilde V)}(t)$, then converges plane-wave-super-algebraically by Lemma~\ref{lem:planewave-convergence}, and a separate stability argument shows that it approximates the physical reactive flux $\nu_{RP}^{(V_{\rm phys})}(t)$ to within a target precision $\epsilon$.

We record a standing hypothesis underlying the construction. For a non-confining $V_{\rm phys}$ (tail $\to V_\infty<\infty$), the unmodified generator $H_\beta(V_{\rm phys})$ has no normalizable equilibrium: $e^{-\beta V_{\rm phys}/2}\notin L^2(\mathbb{R}^{d\eta})$, $0$ lies in the essential spectrum, and the warm-start state $\ket{R}\propto\int_R e^{-\beta V_{\rm phys}/2}$ is well-defined only because $R$ is a fixed bounded region (cf.~\cite{herauIsotropicHypoellipticityTrend2004a,helfferCriteriaPoincareInequality2003}). The wall of Definition~\ref{def:wall} restores a confining potential, hence a normalizable equilibrium $e^{-\beta\widetilde V/2}$, a discrete spectrum, and the Agmon decay used below; it is therefore structurally necessary, not an analytic convenience. For non-confining $V_{\rm phys}$ the operative reactive flux is the surrogate quantity $\nu_{RP}^{(\widetilde V)}(t)$, resolvable down to the dissociation floor of Theorem~\ref{thm:singular-convergence}; for confining $V_{\rm phys}$, Lemma~\ref{lem:stability-bound} additionally ties it to $\nu_{RP}^{(V_{\rm phys})}(t)$ with no floor.

The construction has two pieces: a \emph{short-range patch} that regularizes the coincidence-set singularity, and a \emph{long-range wall} that supplies global confinement.

\begin{defn}[Short-range polynomial-in-$r^2$ patch]
\label{def:short-range-patch}
Let $V_0:(0,\infty)\to\mathbb{R}$ be a pair potential, real-analytic on $(0,\infty)$, with possibly singular behavior as $r\to 0^+$. Fix a cutoff radius $r_c>0$. The \emph{short-range patch} of $V_0$ at $r_c$ is the function $V_0^{(p)}:[0,\infty)\to\mathbb{R}$ defined by
\begin{equation}
\label{eq:patch-def}
V_0^{(p)}(r) \;=\;
\begin{cases}
a_0 + a_1 r^2 + a_2 r^4, & r \le r_c,\\
V_0(r), & r \ge r_c,
\end{cases}
\end{equation}
where the coefficients $(a_0, a_1, a_2)$ are determined by the three $C^2$-matching conditions at $r=r_c$:
\begin{equation}
\label{eq:patch-matching}
V_0^{(p)}(r_c) = V_0(r_c),\qquad
\bigl(V_0^{(p)}\bigr)^\prime(r_c) = V_0^\prime(r_c),\qquad
\bigl(V_0^{(p)}\bigr)^{\prime\prime}(r_c) = V_0^{\prime\prime}(r_c).
\end{equation}
\end{defn}

The patch is a polynomial of degree 2 in $u=r^2$ (equivalently, degree 4 in $r$). It has exactly three free coefficients which are uniquely determined by the three matching conditions, so no additional parameters need to be chosen. The value $V_0^{(p)}(0)=a_0$ is whatever the matching forces; we record its qualitative role in Remark~\ref{rmk:patch-shape} below.

\begin{defn}[Long-range single-body wall]
\label{def:wall}
Fix box half-widths $0<L^\prime<L$, stiffness $\kappa>0$, and degree $p\in\mathbb{N}$. The \emph{long-range wall} on the domain $[-L,L]^{d\eta}$ is the single-body potential
\begin{equation}
\label{eq:wall-def}
V_{\rm wall}(\mathbf{x}) \;=\; \kappa \sum_{i=1}^\eta \mathbf{1}\!\left[\,\norm{\prtcl{}{i}}^2 > L^{\prime\,2}\,\right]\cdot \left(\frac{\norm{\prtcl{}{i}}^2 - L^{\prime\,2}}{L^2 - L^{\prime\,2}}\right)^{p}.
\end{equation}
\end{defn}

The wall is identically zero in the working region $\{\norm{\prtcl{}{i}}<L^\prime\text{ for all }i\}$ and ramps from $0$ at $\norm{\prtcl{}{i}}=L^\prime$ to $\kappa$ at $\norm{\prtcl{}{i}}=L$. The composite function $s\mapsto \mathbf{1}[s>0]\,s^p$ is $C^{p-1}$ but not $C^p$ at the seam $s=0$ (i.e., at $\norm{\prtcl{}{i}}=L^\prime$); this regularity is sufficient for our purposes, as we show in Lemma~\ref{lem:wall-seam-control} below.

\begin{defn}[Smoothed surrogate]
\label{def:smoothed-surrogate}
Let $V_{\rm phys}(\mathbf{x}) = \sum_{i<j} V_0(r_{ij})$ be a many-body potential built from pair interactions $V_0$. The \emph{smoothed surrogate} of $V_{\rm phys}$ with parameters $(r_c, L^\prime, L, \kappa, p)$ is
\begin{equation}
\label{eq:surrogate-def}
\widetilde V(\mathbf{x}) \;=\; \sum_{i<j} V_0^{(p)}(r_{ij}) \;+\; V_{\rm wall}(\mathbf{x}),
\end{equation}
with $V_0^{(p)}$ as in Definition~\ref{def:short-range-patch} and $V_{\rm wall}$ as in Definition~\ref{def:wall}.
\end{defn}

\begin{rem}[Cartesian smoothness and the polynomial-in-$r^2$ structure]
\label{rmk:cartesian-smoothness}
A pair potential expressed as a polynomial in $r_{ij}^2$ extends to a polynomial in the Cartesian coordinates of $\prtcl{}{i},\prtcl{}{j}$, since $r_{ij}^2 = \norm{\prtcl{}{i}-\prtcl{}{j}}^2$ is itself a polynomial of degree 2 in those coordinates. By contrast, a polynomial in $r_{ij}$ contains a factor of $|r_{ij}|$ that, in dimension $d\ge 2$, is non-smooth at the coincidence set; the Cartesian gradient $\nabla V_0(r_{ij}) = V_0^\prime(r_{ij}) (\prtcl{}{i}-\prtcl{}{j})/r_{ij}$ then contains the discontinuous radial unit vector. The polynomial-in-$r^2$ structure in Definition~\ref{def:short-range-patch} is precisely the requirement that $V_0^{(p)}$ define a smooth Cartesian potential. Analogously, $V_{\rm wall}$ in~\eqref{eq:wall-def} uses $\norm{\prtcl{}{i}}^2$ (not $\norm{\prtcl{}{i}}$) inside the polynomial, so its Cartesian smoothness is determined by the regularity of $s\mapsto\mathbf{1}[s>0]s^p$ alone.
\end{rem}

\begin{rem}[Patch shape and absence of spurious bound states]
\label{rmk:patch-shape}
For a typical singular pair potential like Lennard-Jones, the matching conditions~\eqref{eq:patch-matching} encode a steeply descending value and large positive curvature at $r_c$: $V_0^\prime(r_c) < 0$ with $|V_0^\prime(r_c)|$ large, and $V_0^{\prime\prime}(r_c) > 0$ large. The minimal-degree polynomial reconciling these conditions with monotonicity on $[0,r_c]$ produces an inverted-parabola-like shape with $V_0^{(p)}(0) \gg V_0(r_c)$; the value at the origin is dictated by the matching and is typically of order $V_0^\prime(r_c)\,r_c$ in magnitude. Numerical inspection (Section~\ref{subsec:LJ-numerics}) shows that the patch is monotone on $[0,r_c]$ for Lennard-Jones at $r_c=0.85\sigma$, and that the resulting $H_\beta$ has no spurious eigenstates localized inside the patched region: the Boltzmann weight at the patched region's minimum ($r=r_c$) is bounded by $e^{-\beta V_0(r_c)}$, which is below the target precision by construction. A higher-degree patch with extra matching conditions (e.g., $V_0^{(p)\,\prime\prime\prime}(r_c)=V_0^{\prime\prime\prime}(r_c)$) is also possible, but introduces Runge-type oscillations that are best avoided.
\end{rem}
The next two lemmas establish the two key properties of $\widetilde V$: (i) on the patch and wall regions it is polynomial in the Cartesian coordinates, so the hypotheses of Lemma~\ref{lem:planewave-convergence} hold verbatim, while on the bulk it is real-analytic with an analyticity strip bounded below by the coincidence cutoff $r_c$, to which the real-analytic refinement of that lemma applies; and (ii) its dynamics differ from those of $V_{\rm phys}$ by a controllably small amount.

\begin{restatable}[Regularity of the smoothed surrogate]{lemma}{surrogateRegularity}
\label{lem:surrogate-regularity}
Let $V_0$ be real-analytic on $(0,\infty)$ and $\widetilde V$ be the smoothed surrogate of Definition~\ref{def:smoothed-surrogate}. Then:
\begin{enumerate}
\item On the open patched region $\Omega_p := \{\mathbf{x}: r_{ij}<r_c\text{ for some pair }(i,j)\}$, the contribution $\sum_{i<j}V_0^{(p)}(r_{ij})$ is polynomial in the Cartesian coordinates and hence real-analytic.
\item On the open bulk region $\Omega_b := \{\mathbf{x}: r_{ij}>r_c\text{ for all pairs}\;\text{and}\;\norm{\prtcl{}{i}}<L^\prime\text{ for all }i\}$, $\widetilde V(\mathbf{x}) = V_{\rm phys}(\mathbf{x})$ and inherits the real-analyticity of $V_0$.
\item On the open wall region $\Omega_w := \{\mathbf{x}: \norm{\prtcl{}{i}}>L^\prime\text{ for some }i\}$, the wall term is polynomial in the Cartesian coordinates and hence real-analytic; $V_{\rm phys}$ is real-analytic in the same region.
\item At the seams $\{r_{ij}=r_c\}$ and $\{\norm{\prtcl{}{i}}=L^\prime\}$, $\widetilde V$ is $C^2$ and $C^{p-1}$ respectively.
\end{enumerate}
Moreover, the low-lying eigenfunctions of $\widetilde H_\beta$ obey the pointwise Agmon bound
\[
|\varphi(\mathbf{x})|\;\le\; C\,\exp\!\Bigl(-\tfrac{\beta}{2}\bigl(\widetilde V(\mathbf{x})-\textstyle\inf_{\Omega_b}\widetilde V\bigr)\Bigr)
\qquad(\mathbf{x}\in\overline{\Omega_w}),
\]
whose rate is linear in $\beta$. In particular the suppression accrued across the wall, where $V_{\rm wall}$ rises from $0$ at $\norm{\prtcl{}{i}}=L^\prime$ to $\kappa$ at the box edge, is at least $e^{-\beta\kappa/2}$, linear in both $\beta$ and $\kappa$.
\end{restatable}

\begin{proof}
Parts (i)--(iii) follow from Remark~\ref{rmk:cartesian-smoothness} and the polynomial structure of the patch and the wall. Part (iv) is by direct inspection of the matching conditions~\eqref{eq:patch-matching} and the regularity of $s\mapsto\mathbf{1}[s>0]s^p$. The Agmon claim has two parts. The reactive flux is dominated by the symmetric zero mode, which is known exactly: $\varphi_0(\mathbf{x})=Z^{-1/2}e^{-\beta\widetilde V(\mathbf{x})/2}$ solves $\widetilde H_\beta\varphi_0=0$ and is normalizable because $\widetilde V$ is confining, so its amplitude relative to the bulk is exactly $e^{-\tfrac{\beta}{2}(\widetilde V(\mathbf{x})-\inf_{\Omega_b}\widetilde V)}$, with no spectral estimate required. For the low-lying excited states $0<|\lambda_n|\le E$ we use the Agmon estimate~\cite{agmon1982,helffer1988} for $-\beta\widetilde H_\beta=-\Delta+W$, $W=\tfrac{\beta^2}{4}\norm{\grad\widetilde V}^2-\tfrac{\beta}{2}\Delta\widetilde V$, in the degenerate metric $(W-\beta|\lambda_n|)_+^{1/2}\,|d\mathbf{x}|$, which vanishes on the classically allowed region $\{W\le\beta|\lambda_n|\}$ and is real and nonnegative elsewhere. On the wall region $\Omega_w$, where $\widetilde V$ carries the leading $\norm{\prtcl{}{i}}^{2p}$ growth so that $\tfrac{\beta^2}{4}\norm{\grad\widetilde V}^2$ dominates both $\beta|\lambda_n|$ and $\tfrac{\beta}{2}\Delta\widetilde V$, the integrand equals $\tfrac{\beta}{2}\norm{\grad\widetilde V}\,(1+o(1))$; the radial projection used in Lemma~\ref{lem:planewave-convergence} then gives accrued decay $\tfrac{\beta}{2}(\widetilde V(\mathbf{x})-\widetilde V|_{\partial K})\,(1+o(1))$, matching the exact zero-mode rate up to the constant absorbed in $C$. In particular the suppression accrued from $\norm{\prtcl{}{i}}=L^\prime$ (where $V_{\rm wall}=0$) to the box edge (where $V_{\rm wall}=\kappa$) is at least $e^{-\beta\kappa/2}(1+o(1))$. Both parts give a rate linear in $\beta$; the $\sqrt{\beta E}$ of Lemma~\ref{lem:planewave-convergence} is the Fourier bandwidth, a distinct quantity entering the mode count $N$, not the decay rate.
\end{proof}
\begin{restatable}[Plane-wave convergence under seam non-smoothness]{lemma}{wallSeamControl}
\label{lem:wall-seam-control}
Under the hypotheses of Lemma~\ref{lem:surrogate-regularity}, fix $t>0$, and let $\widetilde H_\beta^{L,N}$ be the plane-wave pseudospectral (Fourier-collocation) discretization of $H_\beta(\widetilde V)$ on $[-L,L]^{d\eta}$ with $N$ modes per coordinate. There is a constant $C>0$, depending on $(V_0,\beta,\eta,t)$ but not on $N$, with the following property: for every analytic-region tolerance $\epsilon_\star\in(0,1)$, taking the spectral cutoff $E=t^{-1}\log(4/\epsilon_\star)$ of~\eqref{eq:spectral-cutoff} (instantiated at $\epsilon_\star$) and the associated mode bound~\eqref{eq:fourier-mode-bound}, there is a threshold
\[
N_0(\epsilon_\star)\;\in\;\polylog(\eta,\beta,1/\epsilon_\star,t),
\]
such that for all $N\ge N_0(\epsilon_\star)$,
\begin{equation}
\label{eq:surrogate-convergence}
\bigl|\,\bra{P} e^{t\widetilde H_\beta}\ket{R} - \bra{P}e^{t\widetilde H_\beta^{L,N}}\ket{R}\bigr|
\;\le\;
\underbrace{\;\epsilon_\star\;}_{\text{analytic regions}}
\;+\;
\underbrace{\;C\,N^{-s}\,\sup_{\mathbf{x}\in\Sigma}|\varphi(\mathbf{x})|\;}_{\text{seam contribution}},
\end{equation}
where $\Sigma$ ranges over the seams $\Sigma_{ij}=\{r_{ij}=r_c\}$ and $\Sigma_i=\{\norm{\prtcl{}{i}}=L^\prime\}$, with $s=2$ for the patch seams (across which $\widetilde V$ is $C^2$, so $\Delta\widetilde V$ has a jump in its zeroth radial derivative) and $s=p-1$ for the wall seam, and the supremum taken over symmetric eigenfunctions of $\widetilde H_\beta$ with appreciable participation in $\nu_{RP}$.

The first term is the analytic-region (Fourier-truncation) error: on the patch and wall regions $\widetilde V$ is polynomial in the Cartesian coordinates and Lemma~\ref{lem:planewave-convergence} applies verbatim, while on the bulk $\widetilde V=V_{\rm phys}$ is real-analytic with analyticity strip at least $r_c$, so its real-analytic refinement applies; in all three regions the contribution, including the high-energy tail $\norm{e^{t\widetilde H_\beta}(I-P_E)}_{\rm op}\le e^{-tE}=\epsilon_\star/4$ from~\eqref{eq:spectral-cutoff}, is driven below $\epsilon_\star$ once $N\ge N_0(\epsilon_\star)$. Here $E$ is fixed by $(t,\epsilon_\star)$ alone, not by the surrogate potential. The second term is the operator-aliasing (collocation) contribution of Proposition~\ref{prop:algebraic-rate}: each seam contributes a corner to $\tfrac12\Delta\widetilde V$ whose grid-sampled Fourier coefficients are aliased at the rate $N^{-s}$. This term decays only algebraically in $N$; we do not rely on that decay, but instead bound $N^{-s}\le 1$ and control the seam \textit{amplitude}, which the Agmon decay of Lemma~\ref{lem:surrogate-regularity} bounds as
\begin{equation}
\label{eq:seam-bound}
\sup_{\Sigma_{ij}}|\varphi| \le C\,e^{-\beta V_0(r_c)/2},
\qquad
\sup_{\Sigma_i}|\varphi|\le C\,e^{-\tfrac{\beta}{2}\bigl(V_{\rm phys}(L^\prime)-\inf_{\Omega_b}V_{\rm phys}\bigr)}.
\end{equation}
The patch-seam amplitude is the Boltzmann cost of forcing a pair to separation $r_c$, made small by decreasing $r_c$. The wall-seam amplitude is the value of the confined eigenfunction at $\norm{\prtcl{}{i}}=L^\prime$, where $V_{\rm wall}=0$: it is the \textit{physical} decay from the reactant basin out to $L^\prime$, governed by $V_{\rm phys}$ and not by the wall height $\kappa$. It is made small by placing $L^\prime$ beyond the basin, and is bounded below by $e^{-\tfrac{\beta}{2}(V_\infty-\inf_{\Omega_b}V_{\rm phys})}$, the dissociation floor of Theorem~\ref{thm:singular-convergence}. The wall height $\kappa$ governs only the further decay from $L^\prime$ to the box edge, hence the spatial truncation at $L$, not the seam amplitude.
\end{restatable}

\begin{proof}[Proof sketch]
Decompose $\norm{e^{t\widetilde H_\beta}-\Pi_N e^{t\widetilde H_\beta^{L,N}}\Pi_N}_{\rm op}$ with the spectral cutoff~\eqref{eq:spectral-cutoff} at tolerance $\epsilon_\star$: the high-energy tail obeys $\norm{e^{t\widetilde H_\beta}(I-P_E)}_{\rm op}\le e^{-tE}=\epsilon_\star/4$ by self-adjointness and $\sigma(\widetilde H_\beta)\subset(-\infty,0]$, with $E=t^{-1}\log(4/\epsilon_\star)$ fixed by $(t,\epsilon_\star)$ alone. On the retained band $\{|\lambda_n|\le E\}$ the error decomposes by region. On $\Omega_p,\Omega_w$, $\widetilde V$ is polynomial in the Cartesian coordinates, so Lemma~\ref{lem:planewave-convergence} and~\eqref{eq:fourier-mode-bound} apply verbatim and the Fourier-truncation error is super-algebraically small in $N$. On $\Omega_b$, $\widetilde V=V_{\rm phys}$ is real-analytic but not polynomial; the eigenfunctions remain real-analytic by Kotake--Narasimhan, with analyticity strip bounded below by the distance to the nearest complex singularity of $V_{\rm phys}$, at least $r_c$. The Paley--Wiener estimate with strip width $\gtrsim r_c$ gives Fourier-truncation error $\lesssim e^{-c\,r_c N/L}$; since $r_c$ shrinks only polylogarithmically in $1/\epsilon$, this adds only polylog factors to $N$, so the combined analytic-region contribution (including the $\epsilon_\star/4$ tail) drops below $\epsilon_\star$ for $N\ge N_0(\epsilon_\star)\in\polylog(\eta,\beta,1/\epsilon_\star,t)$. The remaining contributions are the seams, where $\tfrac12\Delta\widetilde V$ acquires a corner aliased at $N^{-s}$ (Proposition~\ref{prop:algebraic-rate}, $s=2$ patch, $s=p-1$ wall) with prefactor the eigenfunction amplitude on the seam, bounded by~\eqref{eq:seam-bound} via Lemma~\ref{lem:surrogate-regularity}. Summing gives~\eqref{eq:surrogate-convergence}.
\end{proof}

\begin{restatable}[Stability of the reactive flux under smoothing, confining case]{lemma}{stabilityBound}
\label{lem:stability-bound}
Let $V_{\rm phys}=\sum_{i<j}V_0(r_{ij})$ be \textit{confining} (so $H_\beta(V_{\rm phys})$ has discrete spectrum and a normalizable equilibrium), and let $\widetilde V$ be its smoothed surrogate, differing from $V_{\rm phys}$ only on the patch region $\Omega_p$. Assume $\ket{R},\ket{P}$ are supported in $\Omega_b$. Then
\begin{equation}
\label{eq:stability-bound}
\bigl|\,\nu_{RP}^{(V_{\rm phys})}(t) - \nu_{RP}^{(\widetilde V)}(t)\bigr| \;\le\; C\,e^{-\beta V_0(r_c)/2},
\end{equation}
where $C$ depends on $t$ and $\norm{V_{\rm phys}-\widetilde V}_{L^\infty(\Omega_p)}$ only through Duhamel multiplicative factors of the form $1+t\cdot\norm{\cdot}$.
\end{restatable}

\begin{proof}[Proof sketch]
Write $H_\beta(V_{\rm phys}) - H_\beta(\widetilde V) = \Delta H$, supported on $\Omega_p$ where $V_{\rm phys}\neq\widetilde V$. The Duhamel identity gives
\[
e^{tH_\beta(V_{\rm phys})} - e^{tH_\beta(\widetilde V)} = \int_0^t e^{(t-s)H_\beta(V_{\rm phys})}\,\Delta H\,e^{sH_\beta(\widetilde V)}\,ds.
\]
Both generators are confining, so both have Agmon-decaying eigenfunctions; sandwiching between $\bra{P}$ and $\ket{R}$ (supported in $\Omega_b$) and using that the participating eigenfunctions are suppressed on $\Omega_p$ by the Boltzmann factor $e^{-\beta V_0(r_c)/2}$ (Lemma~\ref{lem:surrogate-regularity} applied to the confined operator) gives~\eqref{eq:stability-bound}. No wall is needed for confinement, so no $\kappa$-dependent term arises; the argument is structurally that of Lemma~\ref{lem:kappa-scaling-steady-state}. For non-confining $V_{\rm phys}$ this comparison is unavailable, because $H_\beta(V_{\rm phys})$ then has no normalizable equilibrium and no localized eigenfunctions to which an Agmon estimate applies; in that case the operative quantity is the surrogate flux $\nu_{RP}^{(\widetilde V)}(t)$ itself, per the standing hypothesis, and the wall supplies the confinement that makes it well-posed.
\end{proof}

Combining Lemma~\ref{lem:surrogate-regularity} and Lemma~\ref{lem:wall-seam-control} with Lemma~\ref{lem:planewave-convergence}, we obtain the convergence statement for singular, non-confining pair potentials, in which the surrogate flux is the operative target.

\begin{restatable}[Plane-wave convergence for singular pair potentials]{thm}{singularConvergence}
\label{thm:singular-convergence}
Let $V_0$ be real-analytic on $(0,\infty)$ with $V_0(r)\to+\infty$ as $r\to 0^+$ (singular core) and $V_0(r)\to V_\infty$ as $r\to+\infty$ for finite $V_\infty$ (non-confining tail), and let $D=V_\infty-\inf_{\Omega_b}V_{\rm phys}$ be the basin depth. Fix $\beta,\eta,t>0$ and a target precision
\begin{equation}
\label{eq:dissociation-floor}
\epsilon\;\ge\;\epsilon_{\min}\;:=\;C_0\,e^{-\beta D/2}
\end{equation}
above the dissociation floor. There exist constants $C_0,C_1,C_2>0$ (depending on $V_0,\beta,\eta,t$ but not on $\epsilon$ or $N$) and a choice of design parameters
\begin{align*}
r_c &= V_0^{-1}\bigl(\,\tfrac{2}{\beta}\log(C_1/\epsilon)\,\bigr),\\
\kappa &= \tfrac{2}{\beta}\log(C_2/\epsilon),\\
L^\prime &\text{ such that } V_{\rm phys}(L^\prime)-\inf_{\Omega_b}V_{\rm phys}\ge \tfrac{2}{\beta}\log(C_0/\epsilon)\ \ (\text{feasible iff } \epsilon\ge\epsilon_{\min}),\\
L &= L^\prime + O\bigl((\log(1/\epsilon)/(\beta\kappa))^{1/(2p)}\bigr),\\
N &\in \polylog(\eta,\beta,1/\epsilon,t)
\end{align*}
such that the surrogate matrix element satisfies
\begin{equation}
\label{eq:singular-convergence}
\bigl|\,\nu_{RP}^{(\widetilde V)}(t) - \bra{P}e^{t\widetilde H_\beta^{L,N}}\ket{R}\bigr| \;\le\; \epsilon,
\end{equation}
provided $\ket{R},\ket{P}$ are supported in $\Omega_b$.
\end{restatable}

\begin{proof}
Since the target is the surrogate flux, the error is the discretization error alone, bounded by Lemma~\ref{lem:wall-seam-control} at tolerance $\epsilon_\star=\epsilon/2$. For $N\ge N_0(\epsilon/2)\in\polylog(\eta,\beta,1/\epsilon,t)$ the analytic-region contribution is at most $\epsilon/2$. The seam contribution is bounded, using $N^{-s}\le 1$, by
\[
C\Bigl(e^{-\beta V_0(r_c)/2}+e^{-\tfrac{\beta}{2}(V_{\rm phys}(L^\prime)-\inf_{\Omega_b}V_{\rm phys})}\Bigr)
= C\bigl(\epsilon/C_1+\epsilon/C_0\bigr)\le\epsilon/2,
\]
the patch term by the choice of $r_c$ and the wall term by the choice of $L^\prime$, which is feasible precisely because $\epsilon\ge\epsilon_{\min}$ guarantees $\tfrac{2}{\beta}\log(C_0/\epsilon)\le D$. Summing gives~\eqref{eq:singular-convergence}.
\end{proof}

For confining $V_{\rm phys}$, Lemma~\ref{lem:stability-bound} supplies the additional comparison $|\nu_{RP}^{(V_{\rm phys})}-\nu_{RP}^{(\widetilde V)}|\le C\,e^{-\beta V_0(r_c)/2}$, with no wall and no dissociation floor; combined with~\eqref{eq:singular-convergence} this yields $|\nu_{RP}^{(V_{\rm phys})}(t)-\bra{P}e^{t\widetilde H_\beta^{L,N}}\ket{R}|\le\epsilon$ for every $\epsilon\in(0,1)$ at $N\in\polylog(1/\epsilon)$.

The roles of the two seams are distinct. The patch amplitude $e^{-\beta V_0(r_c)/2}=\epsilon/C_1$ is tunable to any target by decreasing $r_c$, so the patch never floors. The wall amplitude $e^{-\tfrac{\beta}{2}(V_{\rm phys}(L^\prime)-\inf_{\Omega_b}V_{\rm phys})}$ is tunable only down to $e^{-\beta D/2}$, the Boltzmann weight of the full basin depth: this is the dissociation floor~\eqref{eq:dissociation-floor}, the precision below which the reactive flux of a non-confining potential at temperature $\beta^{-1}$ is not itself sharply defined, the reactant population escaping over the finite barrier. Above the floor, $N$ is polylog in $1/\epsilon$. Below it, the wall seam can still be resolved through its $N^{-(p-1)}$ decay at the cost $N=\Omega\bigl((\epsilon^{-1}e^{-\beta D/2})^{1/(p-1)}\bigr)$, polynomial in $1/\epsilon$ but with exponent $1/(p-1)$ controlled by the wall degree $p$, so a modest $p$ keeps the sub-floor cost mild.

The algebraic seam floor of Lemma~\ref{lem:wall-seam-control} is not merely small at the tested parameters but structurally below target. Both seam amplitudes are tied to the design parameters as functions of $\epsilon$: at the patch seams $e^{-\beta V_0(r_c)/2}=\epsilon/C_1$ since $r_c=V_0^{-1}(\tfrac2\beta\log(C_1/\epsilon))$, and at the wall seam $e^{-\beta\kappa/2}=\epsilon/C_2$ since $\kappa=\tfrac2\beta\log(C_2/\epsilon)$. Hence the $N^{-2}$ patch floor and the $N^{-(p-1)}$ wall floor each lie below $\epsilon$ for every target precision, so the surrogate never needs to resolve either seam and the polylog-$N$ scaling of Theorem~\ref{thm:singular-convergence} is unaffected. A genuinely super-algebraic asymptotic rate at the patch seam is available if desired by raising the patch to a $C^{2K}$ match (degree $K$ in $r^2$), giving seam aliasing $N^{-2K}$; modest $K$ suffices and stays well below the order at which Runge oscillations on $[0,r_c]$ become a concern (Remark~\ref{rmk:patch-shape}).

This establishes the end-to-end applicability of the framework to physical pair potentials: for any singular pair potential satisfying mild analyticity hypotheses, the matrix element is computable on a plane-wave grid with $N$ polylogarithmic in $1/\epsilon$, preserving the exponential-in-$\eta$ speedup of Theorem~\ref{thm:BE-Aj-ops}. The constants depend on the geometry of the pair potential (through $V_0(r_c)$ at the chosen cutoff and the bulk analyticity strip $\gtrsim r_c$) but not on $\eta$.

\subsubsection{Worked example: the Lennard-Jones pair potential}
\label{subsec:LJ-numerics}

We illustrate the construction with the Lennard-Jones (LJ) pair potential
\[
V_0(r) \;=\; 4\varepsilon\Bigl[(\sigma/r)^{12} - (\sigma/r)^{6}\Bigr],
\]
in reduced units $\varepsilon=\sigma=1$, restricted to $\eta=2$ particles in $d=1$ (relative coordinate $r\in(-L,L)$, COM separating trivially). The LJ pair potential has its minimum at $r=2^{1/6}\approx 1.122$ with depth $-\varepsilon=-1$.

\paragraph{Patch construction.} Choose $r_c=0.85\sigma$, in the repulsive wall just inside the LJ minimum. At this radius, $V_0(r_c)\approx 17.5\varepsilon$, $V_0^\prime(r_c)\approx -322\varepsilon/\sigma$, $V_0^{\prime\prime}(r_c)\approx 5455\varepsilon/\sigma^2$. Solving~\eqref{eq:patch-matching} gives the minimal-degree polynomial-in-$r^2$ patch
\[
V_0^{(p)}(r) \;=\; 681.33\varepsilon - 1648.06\,(\varepsilon/\sigma^2)\,r^2 + 1009.39\,(\varepsilon/\sigma^4)\,r^4,
\]
which is monotone-decreasing on $[0,r_c]$ from $V_0^{(p)}(0)\approx 681\varepsilon$ down to $V_0(r_c)\approx 17.5\varepsilon$. Although $V_0^{(p)}(0)$ is much larger than $V_0(r_c)$, the worst Boltzmann weight inside the patched region occurs at the seam, where $e^{-\beta V_0(r_c)}\sim 10^{-15}$ at $\beta=2$.

\paragraph{Wall construction.} Choose $L=5\sigma$, $L^\prime=4\sigma$, $p=4$, and $\kappa=V_0(r_c)\approx 17.5\varepsilon$ to match the patch-seam suppression. The eigenfunction tail at $\norm{\prtcl{}{i}}=L$ is then $e^{-\beta\kappa/2}\approx 2.5\times 10^{-8}$ at $\beta=2$, comparable to the patch's $e^{-\beta V_0(r_c)/2}$.

\paragraph{Cost overhead.} The block-encoding subnormalization of the surrogate is
\[
\alpha_{\widetilde V} \;\le\; \max\{\alpha_{V_{\rm phys}}^{(\rm bulk)}, \alpha_{V_0^{(p)}}\}  + \alpha_{V_{\rm wall}},
\]
where $\alpha_{V_0^{(p)}}\le|2 a_1|\approx 3300$ on the patched region, $\alpha_{V_{\rm wall}}\le 2p\kappa/(L^2-L^{\prime\,2})\approx 16$, and $\alpha_{V_{\rm phys}}^{(\rm bulk)}$ is the LJ contribution from $r_c<r<L^\prime$, of order $|V_0^\prime(r_c)|/r_c\approx 380$. The patch dominates by a factor of $\sim 10$. All three contributions are $O(1)$ in $\eta$ for a single pair, and scale as $\eta^2$ for $\eta$ particles via the pair sum---the same scaling as the bulk LJ contribution---so the patch does not affect the exponential-in-$\eta$ speedup.

\paragraph{Numerical convergence.} Figure~\ref{fig:LJ-validation} shows the convergence of the matrix element $\bra{P}e^{t\widetilde H_\beta^{L,N}}\ket{R}$ at $\beta=2$, $t=0.5$, for Gaussian test states $\ket{R},\ket{P}$ centered at $r=+r_{\rm min}$ and $r=+r_{\rm min}+1.5\sigma$ respectively (a dissociation-type matrix element). The error decays from $2.1\times 10^{-2}$ at $N=128$ to $2.5\times 10^{-9}$ at $N=384$. Two distinct error channels are present, and it is worth identifying which limits this run. The patch seam ($C^2$, so $s=2$) contributes $e^{-\beta V_0(r_c)/2}N^{-2}$; with $\beta V_0(r_c)\approx 35$ its amplitude is $\sim 10^{-8}$, so this channel sits at $\sim 10^{-8}N^{-2}$. The wall seam ($C^{p-1}$ with $p=4$, so $s=3$) contributes $e^{-\beta(V_{\rm phys}(L^\prime)-\inf_{\Omega_b}V_{\rm phys})/2}N^{-3}$; here $L^\prime=4\sigma$ is in the flat LJ tail, $V_{\rm phys}(L^\prime)\approx 0$, and $\inf_{\Omega_b}V_{\rm phys}=-\varepsilon$ at the LJ minimum, so the relevant exponent is the basin depth $D=V_\infty-\inf_{\Omega_b}V_{\rm phys}\approx\varepsilon$, giving a \textit{worst-case} wall floor $e^{-\beta D/2}N^{-3}\approx 0.37\,N^{-3}\approx 6.5\times 10^{-9}$ at $N=384$. The wall seam, not the patch seam, is therefore the operative worst-case floor for this run, and at the dissociation floor of Theorem~\ref{thm:singular-convergence}: at $\beta=2$ and $D\approx\varepsilon$ the basin is shallow, so the worst-case bound is genuinely $O(e^{-\beta D/2})$ rather than $O(\epsilon)$.

That the realized convergence (super-algebraic, reaching $2.5\times 10^{-9}$) tracks the worst-case wall floor rather than saturating above it reflects Remark~\ref{rmk:bandwidth-worstcase}: the seam amplitude $\sup_\Sigma|\varphi|$ is worst-case over the retained band, but the smooth, bulk-supported test states used here have negligible projection onto the far-tail modes carrying weight at $L^\prime=4\sigma$, so their realized wall-seam amplitude lies well below the equilibrium worst case $e^{-\beta D/2}$. The large wall height $\kappa=V_0(r_c)\approx 17.5\varepsilon$ contributes to the clean convergence through a \textit{separate} channel, the spatial truncation at the box edge, suppressed as $e^{-\beta\kappa/2}\approx 2.5\times 10^{-8}$; it does not act on the wall-seam amplitude, which is fixed at $L^\prime$ where $V_{\rm wall}=0$. In summary, this run sits in the regime $\epsilon\sim\epsilon_{\min}$ of Theorem~\ref{thm:singular-convergence} (shallow basin at moderate $\beta$), where the wall seam is resolved through its $N^{-(p-1)}$ decay rather than driven below target by amplitude alone; the observed decay is consistent with that $N^{-3}$ worst case being beaten by the smooth-state realization.

\begin{figure}[ht!]
\centering
\includegraphics[width=\linewidth]{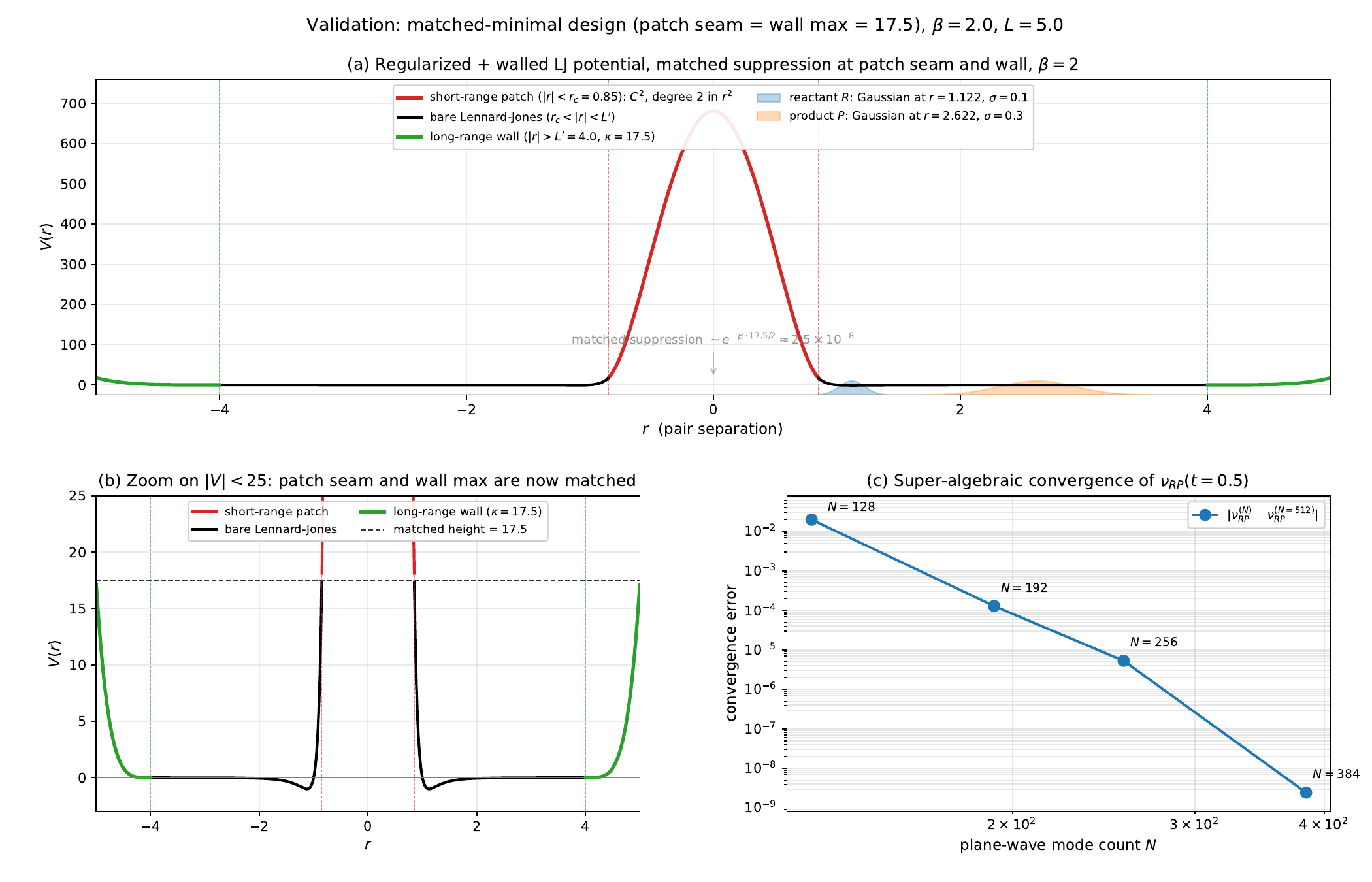}
\caption{Plane-wave convergence of the reactive-flux matrix element for the Lennard-Jones pair potential with the smoothed-surrogate construction of Definition~\ref{def:smoothed-surrogate}. \textit{(a)}~The total potential $\widetilde V(r)$, decomposed into the short-range polynomial patch (red, $|r|<r_c=0.85\sigma$), the bare LJ region (black, $r_c<|r|<L^\prime$), and the long-range polynomial wall (green, $|r|>L^\prime=4\sigma$); within $|r|<r_c$ the simulated potential is the red patch alone. \textit{(b)}~Zoom on $|V|<25\varepsilon$: the patch-seam value at $r=r_c$ equals the wall height at $|r|=L$, both at $\kappa=V_0(r_c)\approx 17.5\varepsilon$, matching the patch-seam Boltzmann suppression $e^{-\beta V_0(r_c)/2}$ to the outer-edge truncation $e^{-\beta\kappa/2}$. \textit{(c)}~Convergence of $\nu_{RP}(t=0.5)$ versus $N$, reaching $\sim 10^{-9}$ at $N=384$. The operative worst-case floor is the wall seam, $e^{-\beta D/2}N^{-(p-1)}$ with basin depth $D\approx\varepsilon$ and $p=4$, i.e.\ $\approx 0.37\,N^{-3}$ (Lemma~\ref{lem:wall-seam-control}); the patch-seam floor $e^{-\beta V_0(r_c)/2}N^{-2}\sim 10^{-8}N^{-2}$ is smaller and not operative. The realized error tracks below the worst-case wall floor because the smooth bulk test states have sub-floor projection onto the far-tail modes (Remark~\ref{rmk:bandwidth-worstcase}).}
\label{fig:LJ-validation}
\end{figure}

\section{Detailed construction of block encoding}
\label{app-sec:BE-A}

\subsection{Block Encoding Potential Gradient}
\label{app-subsec:BE V}
We will assume that the potential functions take the form of polynomials, which, due to the assumption of $V$ being confining implies that the polynomial must be of even degree. For symmetric pair wise potentials, $V_{ij} = V\left(r_{ij}\right)$, the gradients can be simply expressed,
\begin{align*}
    \nabla_i V_{ij} = V'\left(r_{ij}\right)\frac{\mathbf{x}^i-\mathbf{x}^j}{r_{ij}} = - \nabla_j V_{ij}.
\end{align*}
For simplicity, we will also assume that $V$ contains no terms of order 1 in $r_{ij}$ as this would necessitate the computation of the inverse square root, which we wish to avoid. Nevertheless, this assumption can be removed by applying the methods developed in Appendix K of \cite{su_fault-tolerant_2021} to prepare the inverse square root with overall $O(n^2)$ complexity.

To block encode the potential gradient, we will use a combination of arithmetic operations, the inequality testing method of \cite{sandersBlackBoxQuantumState2019}, and QSP. First, we compute the difference between each distinct pair of registers
\begin{align*}
    \ket{\mathbf{r}}_i \ket{\mathbf{r}}_j \rightarrow \ket{\mathbf{r}}_i\ket{\mathbf{r}_i - \mathbf{r}_j}_j.
\end{align*}
Then, we will calculate the sum of squares of these differences into a work register of $2n + \ceil{\log(d)}$ qubits, which prepares the state
\begin{align*}
    \ket{\mathbf{r}}_i\ket{\mathbf{r}_i - \mathbf{r}_j}_j\ket{0} \rightarrow \ket{\mathbf{r}}_i\ket{\mathbf{r}_i - \mathbf{r}_j}_j\ket{r_{ij}^2}.
\end{align*}
Next, we apply the technique from Ref. \cite{sandersBlackBoxQuantumState2019} to kickback an amplitude  $\propto \sqrt{r_{ij}^2}$. 

To apply this method, we append another $n_d = 2n + \ceil{\log(d)}$ qubit ancilla register plus one additional ancilla qubit and prepare the uniform superposition over the $n_d$ qubit register. This prepares the state
\begin{align*}
    \ket{\mathbf{r}}_i\ket{\mathbf{r}_i - \mathbf{r}_j}_j\ket{r_{ij}^2}\frac{1}{\sqrt{2^{n_d}}}\sum_{m \in [2^{n_d}]}\ket{m}\ket{0}.
\end{align*}
Then, we perform an inequality test between $m$ and $r_{ij}^2$,
\begin{align*}
    \textsc{comp}:\ket{n}\ket{m}\ket{0} &\rightarrow \begin{cases}
        \ket{n}\ket{m}\ket{0} & n < m\\
        \ket{n}\ket{m}\ket{1} & n \geq m.
    \end{cases}
\end{align*}
This can be performed through a subtraction circuit, and results in the state 
\begin{align*}
    \ket{\mathbf{r}}_i\ket{\mathbf{r}_i - \mathbf{r}_j}\ket{r_{ij}^2}\frac{1}{\sqrt{2^{n_d}}}\sum_{m <r_{ij}^2}\ket{m}\ket{0} + \ket{\perp}.
\end{align*}
We will forgo uncomputing the subtraction at this stage, as we will use this information in a forthcoming calculation.
Now, we perform the operation $\mathfrak{unif}^{-1}$, which is defined implicitly via
\begin{align*}
   \mathfrak{unif}:\ket{x}\ket{0}\rightarrow \frac{1}{\sqrt{x}}\ket{x}\sum_{j<x}\ket{j}. 
\end{align*}
Applying $\mathfrak{unif}^{-1}$ to the $\ket{r^2_{ij}}$ and $\ket{m}$ registers, we obtain
\begin{align*}
    \sqrt{\frac{r_{ij}^2}{2^{n_d}}}\ket{\mathbf{r}}_i\ket{\mathbf{r}_i - \mathbf{r}_j}_j\ket{r_{ij}^2}\ket{0}\ket{0} + \ket{\perp}.
\end{align*}
Then, uncomputing the sum of squares and the the difference of positions, we obtain
\begin{equation}
    U_r: \ket{\mathbf{r}}_i\ket{\mathbf{r}}_j\ket{0}\rightarrow \frac{r_{ij}}{\sqrt{2^{n_d}}}\ket{\mathbf{r}}_i\ket{\mathbf{r}}_j\ket{0} + \ket{\perp},
\label{eq:BE-dist-op}
\end{equation}
which is a block encoding of symmetric Euclidean distance between particles $i$ and $j$ with subnormalization factor $\alpha_r = \sqrt{2^{n_d}}$. The cost to perform this block encoding is characterized by Lemma \ref{lem:cost-BE-r}.

\begin{lemma}[Complexity of block encoding inter-particle distance]
\label{lem:cost-BE-r}
    There exists a quantum circuit which outputs an $(\alpha_r, 2n + \ceil{\log(d)}+1)$ block encoding of the Euclidean distance matrix $U_r$ in Eq. \eqref{eq:BE-dist-op}, with $\alpha_r = O(N^{d/2})$ using approximately
    \begin{equation}
        2dn^2+4nd  + 4n +O\left(\left(2n + 2\ceil{\log(d)}\right)\log\left(\frac{1}{\epsilon}\right)\right) + O(d)+ 2\ceil{\log(d)}
    \end{equation}
    Toffoli gates
    and 
    \begin{equation}
        2n + \ceil{\log(d)}
    \end{equation}
    work qubits in addition to those that flag the success of the block encoding. 
\end{lemma}
\begin{proof}
    This operation requires the usage of  $2d$ $n$ qubit subtraction circuits to compute and uncompute the difference. The Toffoli complexity of $n$ qubit subtraction is $2n+O(1)$, therefore, the Toffoli complexity of the $d$ $n$ qubit subtractions is $2nd + O(d)$. The summation of $d$ squares can be performed with $dn^2$ Toffoli gates per Lemma 9 of Ref. \cite{su_fault-tolerant_2021} and requires $n_d = 2n + \ceil{\log(d)}$ qubits to store the output. This is then followed by the subtraction circuit to compute $\textsc{ineq}$, which uses an addition $2nd+O(d)$ Toffoli gates. We then, append an additional $n_d$ qubits and prepare the uniform superposition over them. Preparation of the uniform superposition over a power of 2 has zero Toffoli cost, the antecedent inequality test requires $n_d$ Toffoli gates. $\mathfrak{unif}^{-1}$ can be performed with $n_d$ controlled Hadamard gates and a round of fixed point amplitude amplification with $\sim O\left(n_d\log\left(\frac{1}{\epsilon}\right)\right)$, Toffoli gates. The controlled Hadamard gates can be implemented with a single catalytically used $\ket{T}$ state, for a Toffoli complexity of $n_d$. Finally, we conclude by uncomputing the sum of squares which incurs an additional $dn^2$ Toffoli gates. In total, these operations entail
    \begin{align*}
            2dn^2+4nd  + 4n +O\left(\left(2n + 2\ceil{\log(d)}\right)\log\left(\frac{1}{\epsilon}\right)\right) + O(d)+ 2\ceil{\log(d)}
    \end{align*}
    Toffoli gates, and $2n + \ceil{\log(d)}$ additional work qubits.
\end{proof}
For $V(r)$ a polynomial of degree $2k$ with coefficients $c_j$ with no term of degree 1, then $\frac{V'(r)}{r}$ is a polynomial of degree $2(k-1)$ of the form $\sum_{j=2}^{2k-2}j c_j r^{j-2}$. Using, QSP we can query the block encoding of $r_{ij}$ $2k-2$ times to obtain the desired block encoding, with subnormalization factor $\alpha_V = \max_{\norm{r} \leq \sqrt{d}L} \frac{V'(r)}{r}$.

The last stage requires us to multiply by the difference of the shared degrees of freedom for the pair of particles. Fixing some $l \in [d]$, we seek to block encode the operation $\ket{r_{ij}}\rightarrow r_{ij}^{(l)}\ket{r_{ij}}$, where $r_{ij}^{(l)} \in [0,1]$. This can once again be accomplished using the techniques presented in Ref. \cite{sandersBlackBoxQuantumState2019} and a product of block encodings. In total, we perform the operations 
\begin{align*}
    \ket{\mathbf{r}_{i}-\mathbf{r}_j}\ket{0}\ket{0} &\xrightarrow{\textsc{unif}} \frac{1}{\sqrt{N}}\ket{\mathbf{r}_i-\mathbf{r}_j}\sum_{k=0}^{N-1}\ket{k}\ket{0}\\
    &\xrightarrow{\textsc{ineq}}\frac{1}{\sqrt{N}}\ket{\mathbf{r}_i-\mathbf{r}_j}\sum_{k\leq r_{ij}^{(l)}}\ket{k}\ket{0} + \ket{\perp}\\
    &\xrightarrow{\textsc{unif}^\dagger} r_{ij}^{(l)}\ket{\mathbf{r}_i-\mathbf{r}_j}\ket{0}\ket{0} + \ket{\perp}.
\end{align*}
The inequality in the second summation is calculated in two's complement arithmetic to account for the sign information in $r_{ij}^{(l)}$.
Since we have already included the cost of performing the subtraction in the Euclidean distance operator in Lemma \ref{lem:cost-BE-r}, the only additional non-Clifford cost is in performing the \textsc{ineq} operation, which requires $n$ Toffoli gates.  Using the compression gadget described in Ref. \cite{fang_time-marching_2023}, we can perform the product of these block encodings with a single additional ancilla qubit, $n+n_d$ additional Toffoli gates to control on the all zero state of the ancilla register, and one use of each of the block encodings. Overall, this entails a non-Clifford cost of 
\begin{equation}
    2dn^2+4nd  + 7n +O\left(\left(2n + 2\ceil{\log(d)}\right)\log\left(\frac{1}{\epsilon}\right)\right) + O(d)+ 4\ceil{\log(d)} +2.
    \label{eq:cost-implement-U_f}
\end{equation}

Now, to construct a block encoding of the full potential gradient operator, we will control swaps into an ancilla register and apply the above routines. The algorithm proceeds as follows. We prepare two quantum states encoding the uniform superposition over $\eta$ states using two registers of $n_\eta = \ceil{\log(\eta)}$ ancilla qubits, we also prepare a uniform superposition over $d$ states with a register of $n_d = \ceil{\log(d)}$ ancilla. Then, we perform a comparison test which is flagged by another ancilla qubit to mark if the two particle indices are the same. This entails the bulk of the prepare construction,
\begin{equation}
    \textsc{prep}:\ket{0}_{n_\eta}\ket{0}_{n_\eta}\ket{0}_{n_d}\ket{0} \rightarrow \frac{1}{\eta \sqrt{d}}\sum_{i \neq j \in [\eta] }\sum_{l \in [d]}\ket{i}\ket{j}\ket{l}\ket{0} + \ket{\perp}.
\end{equation}
Then, controlling on the two $n_\eta$ qubit registers, as well as the flag indicating the particle indices are unequal, we will swap the registers of particles $i$ and $j$ into ancilla. Then, we will perform two $nd$ qubit quantum Fourier transforms on the two $nd$ qubit registers that hold the degrees of freedom of particles $i$ and $j$. Then, we will apply the block encoding of the gradient operator, controlling on the $n_d$ qubit register indexing over the dimension index. Then, we will uncompute the $n_d$ qubit and one of the $n_\eta$ qubit registers encoding the uniform superpositions. We will use the notation $p_i$ to refer to the encoding of particle $i$'s degrees of freedom in the momentum representation and $r_i$ to refer to particle $i$'s degrees of freedom in the position representation. In all, the operations are as follows
\begin{align*}
    \ket{0}_{temp0}\ket{0}_{temp1}\ket{0}_{n_\eta}\ket{0}_{n_\eta}\ket{0}_{n_d}\ket{0} &\xrightarrow{\textsc{prep}} \ket{0}_{temp0}\ket{0}_{temp1}\frac{1}{\eta\sqrt{d}}\sum_{i \in [\eta]}\sum_{j\neq i}\sum_{k\in[d]}\ket{i}_{n_\eta}\ket{j}_{n_\eta}\ket{k}_{n_d}\ket{0} + \ket{\perp}\\
    &\xrightarrow{c\textsc{swap}_{ij}}\ket{p_i}_{temp0}\ket{p_j}_{temp1}\frac{1}{\eta\sqrt{d}}\sum_{i \in [\eta]}\sum_{j\neq i}\sum_{k\in[d]}\ket{i}_{n_\eta}\ket{j}_{n_\eta}\ket{k}_{n_d}\ket{0} + \ket{\perp}\\
    &\xrightarrow{Q_{nd}\otimes Q_{nd}}\sum_{i \in [\eta]}\sum_{j\neq i}\sum_{k\in[d]}\ket{r_i}_{temp0}\ket{r_j}_{temp1}\frac{1}{\eta\sqrt{d}}\ket{i}_{n_\eta}\ket{j}_{n_\eta}\ket{k}_{n_d}\ket{0} + \ket{\perp}\\
    &\xrightarrow{U_{F}} \sum_{i \in [\eta]}\sum_{j\neq i}\sum_{k\in[d]}\frac{F^i_k(r_{ij})}{\alpha_F} \ket{r_i}_{temp0}\ket{r_j}_{temp1}\frac{1}{\eta\sqrt{d}}\ket{i}_{n_\eta}\ket{j}_{n_\eta}\ket{k}_{n_d}\ket{0} + \ket{\perp'}\\
    &\xrightarrow{Q_{nd}^\dagger \otimes Q_{nd}^\dagger}\sum_{i \in [\eta]}\sum_{j\neq i}\sum_{k\in[d]}\frac{\widehat{F^i_k(r_{ij})}}{\alpha_F} \ket{p_i}_{temp0}\ket{p_j}_{temp1}\frac{1}{\eta\sqrt{d}}\ket{i}_{n_\eta}\ket{j}_{n_\eta}\ket{k}_{n_d}\ket{0} + \ket{\perp'}\\
    &\xrightarrow{c\textsc{swap}^\dagger_{ij}}\ket{0}_{temp0}\ket{0}_{temp1}\sum_{i \in [\eta]}\sum_{j\neq i}\sum_{k\in[d]}\frac{\widehat{F^i_k(r_{ij})}}{\alpha_F} \frac{1}{\eta\sqrt{d}}\ket{i}_{n_\eta}\ket{j}_{n_\eta}\ket{k}_{n_d}\ket{0} + \ket{\perp'}\\
    &\xrightarrow{I\otimes \textsc{unif}_\eta^\dagger \otimes I}\ket{0}_{temp0}\ket{0}_{temp1}\sum_{i \in [\eta]}\sum_{j\neq i}\sum_{k\in[d]}\frac{\widehat{F^i_k(r_{ij})}}{\alpha_F} \frac{1}{\sqrt{d \eta}\eta }\ket{i}_{n_\eta}\ket{0}_{n_\eta}\ket{k}_{n_d}\ket{0} + \ket{\perp''}.
\end{align*}

The subnormalization factor, gate complexity, and number of ancilla qubits needed to perform this block encoding is characterized by the following theorem.
\begin{lemma}[Cost to block encode potential gradient term]
\label{lem:BE-grad-V}
Let $V$ be a polynomial of degree $2k$ containing no term of degree $1$. Then, there exists a quantum circuit which constructs a block encoding of the potential gradient with subnormalization factor $\alpha_F = \eta^{3/2} \sqrt{d} L \alpha_V$, with  $\alpha_V = \max_{\norm{r}\leq \sqrt{d}L} \left|\frac{V'(r)}{r}\right|$, using $O\left(kdn^2(d+1) + \eta n d \right)$ Toffoli gates and  Then, the operations outlined above construct a $(\alpha_F, O(n_\eta + n))$- block encoding of the potential gradient and uses $O(2n + \ceil{\log(d)})$ additional qubits. 
\end{lemma}
\begin{proof}
    The construction of $\textsc{prep}$ requires two applications of the preparation of the uniform superposition over $\eta$ labels, and one application of the uniform superposition of $d$ labels. These operations have complexity $O\left(\left(2n_\eta + n_d\right)\log\left(\frac{1}{\epsilon}\right)\right)$. The controlled swap operations require $2(\eta d n + n_\eta)$ Toffoli gates operations. To implement the $nd$ qubit quantum Fourier transforms require $O((nd)^2)$ Toffoli gates. The implementation of $U_F$ is given by the polynomial degree times the cost reported Eq. \eqref{eq:cost-implement-U_f} and is $O(kdn^2)$. Finally, we uncompute these operations which entails another $O((nd)^2)$ Toffolis for the QFT and $O(2\eta d n)$ to uncompute the uniform superposition. 
\end{proof}

\subsection{Extension to inverse-power law potentials}
\label{app-subsec:BE-realistic-pots}

For more realistic potentials, such as the Morse or Lennard-Jones potentials, Lemma \ref{lem:BE-grad-V} does not apply. In order to extend this work to potentials that are more representative of practical problem instances, we now discuss how to extend our block encoding for non-confining singular potentials. We focus on the particular example of the Lennard-Jones potential, being representative of the kinds of potentials one might use in more realistic simulation contexts.

To accomplish this, we need two steps, first is the calculation of the Lennard-Jones on $[r_c,L]$, second the is the stitching with the polynomial used to smoothly interpolate the cutoff. The gradient of the bare potential
\begin{equation}
    \nabla V_{\text{LJ}}(r_{ij})= \frac{V'(r_{ij})}{r_{ij}}(\mathbf{x}_i -\mathbf{x}_j),
\end{equation}
where we define the scalar function
\begin{equation}
    g'(r) = 4\epsilon\left(-\frac{12 \sigma^{12}}{r^{14}}+6\frac{\sigma^6}{r^8}\right).
\end{equation}
Following the form given in the construction from Sec. \ref{subsec:LJ-numerics}, on $[0,r_c)$, the interstitial part of the potential $V_{\text{int}}$ is given as a polynomial in $r^2$ . The gradient therefore is
\begin{align*}
    \nabla V_{\text{int}}(r_{ij}) = \sum_{k=0}^{d}s_k r_{ij}^{2k-2} (\mathbf{x}_i-\mathbf{x}_j),
\end{align*}
where we define the scalar interstitial polynomial
\begin{align*}
    s'(r) = \sum_{k=0}^{d}s_k r_{ij}^{2k-2}.
\end{align*}

Using the same construction from Sec. \ref{app-subsec:BE V} above, we perform arithmetic to construction the operation
\begin{align*}
    O_{r}\ket{\mathbf{r}}_i\ket{\mathbf{r}}_j\ket{0}\rightarrow \ket{\mathbf{r}}_i\ket{\mathbf{r}_i -\mathbf{r}_j}_j\ket{r_{ij}^2}.
\end{align*}
Then, through an inequality test we compare
\begin{align*}
    \textsc{comp}_r:\ket{r_{ij}}\ket{0} \rightarrow \begin{cases}
        \ket{r_{ij}}\ket{0} & r_{ij} < r_c\\
        \ket{r_{ij}}\ket{1} & r_{ij} \geq r_c,
    \end{cases}
\end{align*}
where we do not require an additional register to store $r_c$ since this is classically known and efficiently precomputable quantity. The cost of this inequality test is just the cost of subtraction on $2n + \ceil{\log(d)}$ bits with Toffoli cost equaling the number of bits. 

Controlling on the $\ket{1}$ state of the inequality test register, we first produce a block encoding of $r_{ij}^{-2}$. Note that for inverse power-law potentials that are of odd power it is a straightforward modification of the following to produce a block encoding of $r_{ij}^{-1}$. We follow the construction provided in Appendix $K$ of Ref. \cite{su_fault-tolerant_2021}. Let $M$ be an integer power of 2 and $m \in 1\ldots M$, we perform inequality testing to check 
\begin{align*}
    \frac{m v_{\max}}{M}<\frac{1}{r_{ij}^2},
\end{align*}
where $v_{\max}=r_c^{-2}$ is the maximum value the Lennard-Jones potential takes before the interstitial region above the cutoff. To check this inequality, we rearrange the above as
\begin{align}
\label{app-eq:ineq-getting-tested}
    m v_{\max} r_{ij}^2 < M.
\end{align}

The algorithm then proceeds as follows: 
\begin{enumerate}
    \item Given $M$ a power of 2, we prepare the uniform superposition over the $m$ qubit ancillary register
    \item Test the inequality $ m v_{\max} r_{ij}^2 < M$
    \item Flip the qubit flagging \textsc{true} on the inequality so that it is in the zero state for all $m$ that satisfy the inequality
    \item Uncompute uniform superposition.
\end{enumerate} 
Letting
\begin{align}
\label{app-eq:bin-variable-ineq-test}
    v(m,\mathbf{x}_i,\mathbf{x}_j) = \begin{cases}
        1 & \frac{m v_{\max}}{M} < \frac{1}{\norm{\mathbf{x}_i -\mathbf{x}_j}^2}\\
        0 & \text{ else },
    \end{cases}
\end{align}
this procedure amounts to the following approximation,
\begin{align*}
   &\frac{v_{\max}}{M}\sum_{m=1}^{M}v(m,\mathbf{x}_i,\mathbf{x}_j)\\
    &= \frac{v_{\max}}{M}\sum_{m=1}^{\floor{\frac{M}{v_{\max}r_{ij}^{2}}}}\\
    &= \frac{v_{\max}}{M}\floor{\frac{M}{v_{\max}r_{ij}^{2}}}\\
    &\approx  \frac{r_{ij}^{-2}}{v_{\max}}.
\end{align*}
By choosing $M$ large enough, this error can be made arbitrarily small, since,
\begin{align*}
    \left|\frac{1}{\norm{\mathbf{x}_i -\mathbf{x}_j}^2} - \frac{v_{\max}}{M}\sum_{m=1}^{M}v(m,\mathbf{x}_i,\mathbf{x}_j)\right| \leq \frac{v_{\text{max}}}{M}.
\end{align*}
Furthermore, the scaling in $\epsilon$ is logarithmic, since the number of Toffoli gates scales linear with $\log_2(M)$ to test the inequality. Indeed, the algorithm applied in this setting is actually cheaper than that of Ref. \cite{su_fault-tolerant_2021} since we do not need to compute squares of $v_{\max}$ and $m$ to test the inequality leading to an overall  complexity of $O\left(\log^2(1/\epsilon)\right)$.

\begin{lemma}[Cost to block encode inverse square]
\label{lem:BE-inv-square}
Let $r_c>0$ be a cutoff parameter controlling the distance of the inverse square from zero. Let $v_{\max} = r_c^{-2}$ and $M$ an integer power of $2$. The for every $\epsilon>0$, there exists a $(v_{\max}, \log_2(M)+1, \frac{v_{\max}}{M})$ block encoding of the diagonal matrix
\begin{align*}
    r_{ij}^{-2}:\ket{\mathbf{x}}_i\ket{\mathbf{x}}_j \rightarrow \frac{\norm{\mathbf{x}_i-\mathbf{x}_j}^{-2}}{v_{\max}} \ket{\mathbf{x}}_i\ket{\mathbf{x}}_j,
\end{align*}
using $O(dn^2 + \log(1/\epsilon))$ Toffoli gates.
\end{lemma}
\begin{proof}
    First, we perform the necessary arithmetic to construct the oracle
    \begin{align*}
        O_r:\ket{\mathbf{x}}_i\ket{\mathbf{x}}_j\ket{0} \rightarrow \ket{\mathbf{x}}_i\ket{\mathbf{x}}_j\ket{r_{ij}^2}
    \end{align*}
    which can be done using $dn^2 + 2nd + O(d)$ Toffoli gates according to Lemma \ref{lem:cost-BE-r}. Then, introducing an ancilla register of $m = \log_2(M)$ qubits we prepare the uniform superposition over them which costs zero Toffoli gates, and prepares the state
    \begin{align*}
        \frac{1}{\sqrt{M}}\sum_{m=1}^{M}\ket{m}\ket{\mathbf{x}}_i\ket{\mathbf{x}}_j\ket{r_{ij}^2}.
    \end{align*}
    Appending an additional ancilla qubit acting as flag, we test the inequality \ref{app-eq:ineq-getting-tested} which prepares
    \begin{align*}
        \frac{1}{\sqrt{M}}\sum_{m=1}^{M}\ket{m}\ket{\mathbf{x}}_i\ket{\mathbf{x}}_j\ket{r_{ij}^2}\ket{v(m,\mathbf{x}_i,\mathbf{x}_j)},
    \end{align*}
    recalling that $v(m,\mathbf{x}_i,\mathbf{x}_j)$ is given in Eq. \eqref{app-eq:bin-variable-ineq-test}. The cost to test this inequality is $\max\{m, 2n + \ceil{\log(d)}\}$ Toffoli gates. Then, applying an $X$ gate on the flag, and applying Hadamards on the $m$ qubit ancilla register we obtain
    \begin{align*}
        &\ket{0}\ket{\mathbf{x}}_i\ket{\mathbf{x}}_j\ket{r_{ij}^2}\frac{1}{M}\sum_{m=1}^{M}\ket{\neg v(m,\mathbf{x}_i,\mathbf{x}_j)} + \ket{\perp}\\
        &=\ket{0}\ket{\mathbf{x}}_i\ket{\mathbf{x}}_j\ket{r_{ij}^2}\frac{1}{M}\sum_{m=1}^{\floor{M\frac{r_{ij}^{-2}}{v_{\max}}}}\ket{0} + \ket{\perp}\\
        &= \frac{1}{M}\floor{M\frac{r_{ij}^{-2}}{v_{\max}}}\ket{0}\ket{\mathbf{x}}_i\ket{\mathbf{x}}_j\ket{r_{ij}^2}\ket{0} +\ket{\perp}\\
    \end{align*}
    without any additional Toffoli complexity.

    Since,
    \begin{align*}
        &\left|\frac{1}{M}\floor{M\frac{r_{ij}^{-2}}{v_{\max}}}- \frac{r_{ij}^{-2}}{v_{\max}}\right|\leq\frac{1}{M},
    \end{align*}
    it suffices to choose $M > \frac{v_{\max}}{\epsilon}$ to ensure that the block encoding is within $\epsilon$ in spectral norm of the diagonal matrix encoding the inverse squared distance for all $r_{ij} \geq r_c$ following the convention of Eq. \eqref{eq:defn-BE}.
\end{proof}

Now, with the block encoding of the inverse square, we can produce a block encoding of 
\begin{align*}
    g'(r) =4\epsilon\left(-\frac{12 \sigma^{12}}{r^{14}}+6\frac{\sigma^6}{r^8}\right),
\end{align*}
with subnormalization factor given by the maximum value attained by the interpolating polynomial $s'(r)$ on $[0,r_c]$. Given that $g'(r)$ is a polynomial of degree 7 in $r^{-2}$, and that $|g'(r)/\max_{r\in[0,r_c]}s'(r)| < 1$ for all $r \in [0,L]$ this can be accomplished with QSP and a constant number of additional queries to the block encoding described in Lemma \ref{lem:BE-inv-square}.

Finally, since the interpolating polynomial is a polynomial in $r^2$, we may perform the technique from Ref. \cite{sandersBlackBoxQuantumState2019}, without the invocation of the $\mathfrak{unif}^{-1}$ operation to produce a block encoding of $r^2$. Then, we apply QSP on this to approximate the interpolating polynomial.

Performing all of the above operations in a controlled manner, applying $s'(r)$ when $\ket{r<r_c} = \ket{1}$ and $g'(r)$ when $\ket{r< r_c} = \ket{0}$, we produce a block encoding of the smoothed Lennard-Jones potential of Sec. \ref{def:smoothed-surrogate}. The following theorem computes the total cost.

The enforcement of the confining potential $g^{\text{(wall)}}$ is a single-body potential with costs that are equivalent to the cost of producing a block encoding of the square Euclidean distance and thereby contributes an additive cost of $O(\eta dn^2)$ Toffoli gates and is therefore subdominant to the costs associated with block encoding the pair potential.

\begin{thm}[Overall cost to block encode gradient of smoothed surrogate Lennard-Jones potential]
    Let $\widetilde{V}$ be the smoothed surrogate of the Lennard-Jones potential
    \begin{align*}
        \widetilde{V}(\mathbf{x}) = \sum_{i<j}V_{0}^{(p)}(r_{ij}) + V_{\rm{wall}}(\mathbf{x})
    \end{align*}
    where 
    \begin{align*}
     V_0^{(p)}(r) \;=\;
    \begin{cases}
        a_0 + a_1 r^2 + a_2 r^4, & r \le r_c,\\
        V_{\rm LJ}(r), & r \ge r_c,
    \end{cases}
    \end{align*}
    and let $L' < L$ be the parameter controlling the ramp-up region for the wall term
    \begin{align*}
    V_{\rm wall}(\mathbf{x}) \;=\; \kappa \sum_{i=1}^\eta \mathbf{1}\!\left[\,\norm{\prtcl{}{i}}^2 > L^{\prime\,2}\,\right]\cdot \left(\frac{\norm{\prtcl{}{i}}^2 - L^{\prime\,2}}{L^2 - L^{\prime\,2}}\right)^{p}.
    \end{align*}

    With $g'(r) = \partial_r V_0^{(p)}(r)/r$ and $w'(r) = \partial_r V_{\rm wall}(r)/r$, and $\alpha_{\rm pair}:=\max_{r\leq \sqrt{d}L}\left|g'(r)\right| $ and $\alpha_{\rm wall}:= \max_{\sqrt{d}L'\leq r \leq \sqrt{d}L}\left| w'(r)\right|$,
    there exists a quantum algorithm using 
    \begin{align}
        O\left((nd)^2 + dn^2 + \eta d + \eta + n +\ceil{\log(\eta)} + \ceil{\log(d)} + \log(\alpha_{\rm pair}/\epsilon)\right)
    \end{align}
    Toffoli gates to produce an $(\sqrt{\eta d} \left(\eta \alpha_{\rm pair} + \alpha_{\rm wall}\right),2\ceil{\log(\eta)} + \ceil{\log(d)} + \log(1/\epsilon)+4, \epsilon)$ block encoding of $\nabla\widetilde{V}(\mathbf{x})$.
\end{thm}
\begin{proof}
    
    Our strategy will be to produce a linear combination of block encodings of the pair and wall potentials. The prepare circuit for the LCU is formed by a single additional ancilla qubit in the state 
    \begin{align*}
        \ket{\alpha_V} &= \frac{1}{\sqrt{\alpha_{\text{wall}} + \alpha_{\rm pair}} }\left(\sqrt{\alpha_{\rm pair}}\ket{0} + \sqrt{\alpha_{\rm wall}}\ket{1}\right)\\
        &\equiv \frac{1}{\sqrt{\alpha_V}}\left(\sqrt{\alpha_{\rm pair}}\ket{0} + \sqrt{\alpha_{\rm wall}}\ket{1}\right).
    \end{align*}
    We then proceed by preparing the uniform superposition over $\eta$ labels and a uniform superposition over $d$ labels. Then, controlling on the $\ket{0}$ state of $\ket{\alpha_V}$ we perform the following operations
    \begin{enumerate}
        \item Prepare a uniform superposition over $\eta$ indices on another register 
        \item Perform an inequality test to check if two particle indices are unequal
        \item Controlling on the success state of the inequality test, control a swap of particles $i$ and $j$ into temporary ancilla
        \item Perform $2$, $nd$-qubit QFT's on each temporary register
        \item Perform binary subtraction on the temporary registers $\ket{\mathbf{x}_i}\ket{\mathbf{x}_j}\rightarrow \ket{\mathbf{x}_i}\ket{\mathbf{x}_{i} - \mathbf{x}_j}$
        \item Perform inequality test to check if $\norm{r_{ij}}<r_c$
        \item Controlling on the flag $\norm{r_{ij}}\geq r_c$, apply the construction from Lemma  \ref{lem:BE-inv-square} in conjunction with a degree 7 QSP
        \item Controlling on the flag $\norm{r_{ij}} < r_c$, apply the construction in Lemma \ref{lem:BE-grad-V} to the interstitial polynomial.
        \item Uncompute steps (1-6).
    \end{enumerate}

    To implement the above in a controlled manner only requires steps 7 and 8 needing to be implemented in a controlled manner. Within steps 7 and 8, only the phases in the QSP need to be implemented in a controlled manner. Overall, this results in an overhead of $\ceil{\log(v_{\max}/\epsilon)}+2$ additional Toffolis to those reported in the respective lemmas.
    
    Controlling on the $\ket{1}$ state of $\ket{\alpha_V}$, we perform the following operations
    \begin{enumerate}
        \item Compute $\ket{\mathbf{x}_i}\ket{0}\rightarrow \ket{\mathbf{x}_i}\ket{\norm{\mathbf{x}_i}^2}$
        \item Test if $\norm{r_i} > L'$
        \item Controlling on the qubit flagging the $\norm{r_i} > L'$ state, apply the construction from Lemma \ref{lem:BE-grad-V}
        \item Uncompute steps (1-3).
    \end{enumerate}

    Finally, we apply the inverse of the circuit used to prepare $\ket{\alpha_V}$, and post-select on being in the all-zero state.

    Neglecting uncomputed ancilla and letting $g'(r) = \partial_r V_0^{(p)}(r)/r$ and $w'(r) = \partial_r V_{\rm wall}(r)/r$, the above operations give the following
    \begin{align*}
        \ket{\mathbf{x}_i}\ket{\mathbf{x}_j}\ket{0} \rightarrow \frac{1}{\sqrt{\eta d}}\sum_{i =0}^{\eta -1}\sum_{k=0}^{d-1}\left(\sum_{j<i} \frac{g'(r_{ij})}{\eta \alpha_{\rm pair}} + \frac{w'(r_i)}{\alpha_{\rm wall}}\right)\ket{\mathbf{x}_i}\ket{\mathbf{x}_j}\ket{i}\ket{k}\ket{0} + \ket{\perp},
    \end{align*}
    resulting in an overall subnormalization factor of $\alpha_V = \sqrt{\eta d} \left(\eta \alpha_{\rm pair} + \alpha_{\rm wall}\right)$.

    The number of ancilla qubits flagging success of this block encoding is $2\ceil{\log(\eta)}  + \ceil{2\log(v_{\max}/\epsilon)} + 4$. The Toffoli costs are as follows
    \begin{enumerate}
        \item Construction of uniform superposition over $\eta$ indices: $\ceil{\log(\eta)}+O(1)$
        \item Inequality test to check if particle indices $i\neq j$: $\ceil{\log(\eta)}$
        \item Controlled swap of particle registers into ancilla (and unswap): $12 \eta d + 4\eta -8$
        \item QFTs: $2(nd)^2$ 
        \item Subtraction circuit: $nd$
        \item Inequality test to check $\norm{r_{ij}}\geq r_c$: $dn^2+2n+\ceil{\log(d)}$
        \item Construction of potential function: $O(dn^2 + \log(\alpha_{\rm pair}/\epsilon))$ Toffoli gates
        \item Uncompute using \cite{gidneyHalvingCostQuantum2018} (except swaps): no additional Toffoli gates.
    \end{enumerate}
    In the branch where we apply the wall potential gradient the costs that are distinct from the above are as follows:
    \begin{enumerate}
        \item Test $\norm{r_i} >L'$: $dn^2 + 2n + \ceil{\log(d)}$
        \item Apply wall Hamiltonian: $O(dn^2)$
    \end{enumerate}
    This leads to an overall Toffoli cost of 
    \begin{align}
        O\left((nd)^2 + dn^2 + \eta d + \eta + n +\ceil{\log(\eta)} + \ceil{\log(d)} + \log(\alpha_{\rm pair}/\epsilon)\right).
    \end{align}

\end{proof}

We remark that this is asymptotically the \textit{same} cost as implementing polynomial potentials, with perhaps a marginal increase in the subnormalization factor over the polynomial method. Furthermore, the constant factors in the Toffoli cost are essentially the same between both polynomial and inverse polynomial potentials. 

\subsection{Block encoding gradient operator}
\label{subsubsec:BE-sqrt-laplacian}
The block encoding of the $\eta d$-dimensional gradient operator is significantly simplified in the plane wave basis as the derivative operator takes the form of a scalar multiplication. Our goal now is to block encode the operation 
\begin{equation}
    \sum_{i\in[d\eta]} \ket{i} \otimes \nabla_i 
\end{equation}
where $k_m = \frac{2 \pi m}{L}$. This can be performed using the same block encoding strategy that we employed to construct the \textit{position} operator $\ket{r} \rightarrow r \ket{r}$. Indeed, this operation can be constructed with just a single additional ancilla qubit that is prepared in the quantum state $\frac{1}{\beta^{1/4}}\ket{0} + \beta^{1/4}\ket{1}$, with normalization factor $\frac{\beta+1}{\sqrt{\beta}} \in O(\max\{\sqrt{\beta},\sqrt{\beta^{-1}}\})$. On the $\ket{0}$ state we perform controlled application of the inequality test $\textsc{comp}$, and on the $\ket{1}$ state we perform controlled application of the operator $U_F$, which in turn requires $2k-2$ controlled applications of $U_r$.

\subsection{Proof of Theorem \ref{thm:BE-Aj-ops}}
\BEAj*
\begin{proof}
    The block encoding of each $A_j$ term requires one usage of the quantum circuit which block encodes $\nabla V$ and one usage of the circuit which block encodes $\nabla$. By Lemma \ref{lem:BE-grad-V} this is done with $O(kdn^2(d+1) + \eta n d)$ Toffoli or simpler gates, and the gradient operator requires an additional $O(n)$ Toffoli or simpler gates. Each of these block encodings are obtained with subnormalization factor $\alpha_F$ and $\alpha_\nabla$ respectively. Then, we form the LCU of these block encodings using a straightforward variation of the quantum circuit from Fig. \ref{app-fig:sqrt BE} and using an additional ancilla qubit which is prepared in the state 
    \begin{equation}
        \frac{1}{\sqrt{\alpha_A}}\left(\sqrt{\alpha_F}\beta^{1/4}\ket{0} + \sqrt{\alpha_\nabla} \beta^{-1/4} \ket{1}\right),
    \end{equation}
    where 
    \begin{equation}
        \alpha_A = \alpha_F \sqrt{\beta} + \alpha_\nabla \sqrt{\beta^{-1}}.
    \end{equation}
    Thus, with $\alpha_F \in O\left(\eta^{3/2}\sqrt{d}L \alpha_V\right)$, and $\alpha_\nabla \in O\left(\sqrt{\eta d} \frac{N}{L}\right)$, we have 
    \begin{equation}
        \alpha_A \in O\left(\sqrt{\beta d} \eta^{3/2} L \alpha_V + \sqrt{\frac{\eta d}{\beta}} \frac{N}{L}\right).
    \end{equation}
\end{proof}

\end{document}